\documentclass[10pt]{article}

\usepackage{amsfonts,amsmath,amssymb,amsthm}
\usepackage[pdftex]{graphicx}
\usepackage[hmargin=1in,vmargin=1in]{geometry}
\usepackage{natbib,setspace,enumerate,mathtools}
\usepackage{multirow}
\usepackage{booktabs}
\allowdisplaybreaks
\usepackage[toc,title,titletoc,header]{appendix}
\usepackage{adjustbox}
\usepackage[colorlinks,citecolor=blue]{hyperref}
\usepackage{etoolbox} % This package goes here to add \qed to remarks automatically.
\usepackage{threeparttable}
\def\qed{\rule{2mm}{2mm}}
\def\independent{\perp \!\!\! \perp}

\newtheorem{theorem}{Theorem}[section]
\newtheorem{lemma}{Lemma}[section]

\theoremstyle{definition}

\newtheorem{remark}{Remark}[section]
\newtheorem{assumption}{Assumption}[section]

\AtEndEnvironment{remark}{~\qed}
\AtEndEnvironment{example}{~\qed}

\DeclareMathOperator*{\var}{Var}
\DeclareMathOperator*{\cov}{Cov}

\newcommand{\mycomment}[1]{}

\begin{document}

\author{
Yuehao Bai \\
Department of Economics\\
University of Southern California \\
\url{yuehao.bai@usc.edu}
\and
Jizhou Liu \\
Booth School of Business \\
\textcolor{white}{111} University of Chicago \textcolor{white}{11}\\
\url{jliu32@chicagobooth.edu}
\and
Azeem M.\ Shaikh\\
Department of Economics\\
University of Chicago \\
\url{amshaikh@uchicago.edu}
\and
Max Tabord-Meehan\\
Department of Economics\\
University of Chicago \\
\url{maxtm@uchicago.edu}
}

\bigskip

\title{Inference in Cluster Randomized Trials with Matched Pairs \thanks{We would like to thank seminar and conference participants at Aarhus University, Canadian Economics Association Conference, CIREQ, Indiana University, NAWMES, NYU, Ohio State University, Princeton University, Southern Economic Assocation Conference, University of Southern California, University of Wisconsin-Madison, and Yale University for helpful comments on this paper. We thank Xun Huang for providing excellent research assistance. The fourth author acknowledges support from NSF grant SES-2149408.}}

\maketitle

\vspace{-0.3in}

\begin{spacing}{1.2}
\begin{abstract}
This paper studies inference in cluster randomized trials where treatment status is determined according to a ``matched pairs'' design.  Here, by a cluster randomized experiment, we mean one in which treatment is assigned at the level of the cluster; by a ``matched pairs'' design, we mean that a sample of clusters is paired according to baseline, cluster-level covariates and, within each pair, one cluster is selected at random for treatment. We study the large-sample behavior of a weighted difference-in-means estimator and derive two distinct sets of results depending on if the matching procedure does or does not match on cluster size. We then propose a single variance estimator which is consistent in either regime. Combining these results establishes the asymptotic exactness of tests based on these estimators. Next, we consider the properties of two common testing procedures based on $t$-tests constructed from linear regressions, and argue that both are generally conservative in our framework. We additionally study the behavior of a randomization test which permutes the treatment status for clusters within pairs, and establish its finite-sample and asymptotic validity for testing specific null hypotheses. Finally, we propose a covariate-adjusted estimator which adjusts for additional baseline covariates not used for treatment assignment, and establish conditions under which such an estimator leads to strict improvements in precision. A simulation study confirms the practical relevance of our theoretical results.
\end{abstract}
\end{spacing}



\thispagestyle{empty} 
\newpage
\setcounter{page}{1}

\section{Introduction}
This paper studies the problem of inference in cluster randomized experiments where treatment status is determined according to a ``matched pairs'' design.  Here, by a cluster randomized experiment, we mean one in which treatment is assigned at the level of the cluster; by a ``matched pairs'' design, we mean that the sample of clusters is paired according to baseline, cluster-level covariates and, within each pair, one cluster is selected at random for treatment. Cluster matched pair designs feature prominently in all parts of the sciences: examples in economics include \cite{angrist2009effects}, \cite{fryer2014injecting}, \cite{banerjee2015miracle}, \cite{crepon2015estimating}, \cite{bruhn2016impact}, \cite{glewwe2016better}, \cite{fryer2018pupil} and \cite{romero2020outsourcing}.

Following recent work in \cite{bugni2024inference}, we develop our results in a sampling framework where clusters are realized as a random sample from a population of clusters. Importantly, in this framework cluster sizes are modeled as random and ``non-ignorable," meaning that ``large" clusters and ``small" clusters may be heterogeneous, and, in particular, the effects of the treatment may vary across clusters of differing sizes. The framework additionally allows for the possibility of two-stage sampling, in which a subset of units is sampled from the set of units within each sampled cluster.

We first study the large-sample behavior of a weighted difference-in-means estimator under two distinct sets of assumptions on the matching procedure. Specifically, we distinguish between settings where the matching procedure does or does not match on a function of cluster size. For both cases, we establish conditions under which our estimator is asymptotically normal and derive simple, closed-form expressions for the asymptotic variance. Using these results, we establish formally that employing cluster size as a matching variable in addition to baseline covariates delivers a weak (and often strict) improvement in asymptotic efficiency relative to matching on baseline covariates alone, and in fact achieves full efficiency in a broad class of experimental designs: see Remark \ref{rem:N_efficient} for further discussion. We then propose a variance estimator which is consistent for either asymptotic variance depending on the nature of the matching procedure. Combining these results establishes the asymptotic exactness of tests based on our estimators. 
 
 We then consider the asymptotic properties of two commonly recommended inference procedures based on linear regressions of the individual-level outcomes on a constant and cluster-level treatment. The first inference procedure clusters at the level of treatment assignment. The second inference procedure clusters at the level of assignment pairs, as recently recommended in \cite{de2019level}. We establish that both procedures are generally conservative in our framework. 
%\footnote{We note that \cite{de2019level} consider an alternative finite-population ``design-based" inferential framework, which may be particularly attractive in settings where the experimental sample is not explicitly drawn from a larger population.} 

Next, we study the behavior of a randomization test which permutes the treatment status for clusters within pairs. We establish the finite-sample validity of such a test for testing a certain null hypothesis related to the equality of potential outcome distributions under treatment and control, and then establish asymptotic validity for testing null hypotheses about the size-weighted average treatment effect. We emphasize, however, that the latter result relies heavily on our choice of test statistic, which is studentized using our novel variance estimator. In simulations, we find that this randomization test controls size more reliably than any of the other inference procedures we consider in the paper, while delivering comparable power.

Finally, we derive large-sample results for a covariate-adjusted version of our estimator, which is designed to improve precision by exploiting additional baseline covariates which were not used for treatment assignment. As discussed in \cite{bai2024covariate} and \cite{cytrynbaum2023covariate}, standard covariate adjustments based on a regression using treatment-covariate interactions \citep[see, for instance,][for a succinct treatment]{negi2021revisiting} are not guaranteed to improve efficiency when treatment assignment is not completely randomized. For this reason, we consider a modified version of the estimator developed in \cite{bai2024covariate} for individual-level matched pair experiments. Our results show that our covariate-adjusted estimator is guaranteed to improve asymptotic efficiency relative to the unadjusted estimator. 
%whether or not the matching procedure matches on cluster size. 
%Interestingly, we also find that this improvement in efficiency is \emph{not} guaranteed when cluster size is excluded as a matching variable, and document in a simulation study that in fact such covariate adjustments may increase variance.

The analysis of data from cluster randomized experiments and data from experiments with matched pairs has received considerable attention \citep[see][for general overviews]{donner2000design, athey2017econometrics, hayes2017cluster}, but most recent work has focused on only one of these two features at a time. Recent work on the analysis of cluster randomized experiments includes  \cite{middleton2015unbiased}, \cite{su2021model}, \cite{schochet2021design}, and \cite{wang2022model} \citep[see][for a general discussion of this literature as well as further references]{bugni2024inference}. We note in particular that both \cite{middleton2015unbiased} and \cite{su2021model} discuss the benefits of using cluster size as a covariate in regression adjustment in the context of completely randomized experiments. Recent work on the analysis of matched pairs experiments includes \cite{jiang2020bootstrap}, \cite{cytrynbaum2021designing}, \cite{bai2024inference}, and \cite{bai2022optimality} \citep[see][for a discussion of this literature as well as further references]{bai2022inference}. Two papers which focus specifically on the analysis of cluster randomized experiments with matched pairs are \cite{imai2009essential} and \cite{de2019level}. Both papers maintain a finite-population perspective, where the primary source of uncertainty is ``design-based," stemming from the randomness in treatment assignment. In such a framework, both papers study the finite and large-sample behavior of difference-in-means type estimators and propose corresponding variance estimators which are shown to be conservative. In contrast, our paper maintains a ``super-population" sampling framework and proposes a novel variance estimator which is shown to be asymptotically exact in our setting. In Appendix \ref{sec:sims-finpop}, we repeat some of the simulation exercises we consider in the main text in a design-based framework. There we illustrate that our estimator may have benefits in the design-based framework as well.

The remainder of the paper is organized as follows. In Section \ref{sec:setup} we describe our setup and notation. Section \ref{sec:main} presents our main results. Section \ref{sec:simulations} studies the finite-sample behavior of our proposed tests via a simulation study. We conclude with recommendations for empirical practice in Section \ref{sec:recommendations}.

\section{Setup and Notation} \label{sec:setup}
In this section we introduce the notation and assumptions which are common to both matching procedures considered in Section \ref{sec:main}. We broadly follow the setup and notation developed in \cite{bugni2024inference}. Let $Y_{i,g} \in \mathbf{R}$ denote the (observed) outcome of interest for the $i$th unit in the $g$th cluster, $D_g \in \{0, 1\}$ denote the treatment received by the $g$th cluster, $X_g \in \mathbf{R}^k$ the observed, baseline covariates for the $g$th cluster, and $N_g \in \mathbf{Z}_+$ the size of the $g$th cluster. In what follows we sometimes refer to the vector $(X_g, N_g)$ as $W_g$. Further denote by $Y_{i,g}(d)$ the potential outcome of the $i$th unit in cluster $g$, when all units in the $g$th cluster receive treatment $d \in \{0, 1\}$. As usual, the observed outcome and potential outcomes are related to treatment assignment by the relationship
\begin{equation}\label{eq:potential}
  Y_{i,g} = Y_{i,g}(1)D_g + Y_{i,g}(0)(1 - D_g)~.
\end{equation}
In addition, define $\mathcal{M}_g$ to be the (possibly random) subset of $\{1, 2, \ldots, N_g\}$ corresponding to the observations within the $g$th cluster that are sampled by the researcher. We emphasize that a realization of $\mathcal{M}_g$ is a \emph{set} whose cardinality we denote by $|\mathcal{M}_g|$, whereas a realization of $N_g$ is a positive integer. For example, in the event that all observations in a cluster are sampled, $\mathcal{M}_g = \{1, \ldots, N_g\}$ and $|\mathcal{M}_g| = N_g$. We assume throughout that our sample consists of $2G$ clusters and denote by $P_G$ the distribution of the observed data $$Z^{(G)} := (((Y_{i,g} : i \in \mathcal{M}_g), D_g, X_g, N_g) : 1 \leq g \leq 2G )~,$$
and by $ Q_G$ the distribution of $$(((Y_{i,g}(1),Y_{i,g}(0) : 1 \le i \le N_g), \mathcal{M}_g, X_g, N_g) : 1 \leq g \leq 2G )~.$$  Note that $P_G$ is determined jointly by \eqref{eq:potential} together with the distribution of $D^{(G)} := (D_g : 1 \leq g \leq 2G)$ and $Q_G$, so we will state our assumptions below in terms of these two quantities.

We now describe some preliminary assumptions on $Q_G$ that we maintain throughout the paper.  In order to do so, it is useful to introduce some further notation.  To this end, for $d \in \{0,1\}$, define $$\bar Y_g(d) := \frac{1}{|\mathcal{M}_g|} \sum_{i \in \mathcal{M}_g} Y_{i,g}(d)~.$$  Further define $R_G(\mathcal{M}_g^{(G)}, X^{(G)}, N^{(G)})$ to be the distribution of $$((Y_{i,g}(1), Y_{i,g}(0) : 1 \le i \le N_g) : 1 \leq g \leq 2G) ~\big |~ \mathcal{M}_g^{(G)}, X^{(G)}, N^{(G)}~,$$ where $\mathcal{M}_g^{(G)} := (\mathcal{M}_g : 1 \leq g \leq 2G)$, $X^{(G)} := (X_g : 1 \leq g \leq 2G)$ and $N^{(G)} := (N_g : 1 \leq g \leq 2G)$.  Note that $Q_G$ is completely determined by $R_G(\mathcal{M}_g^{(G)}, X^{(G)}, N^{(G)})$ and the distribution of $(\mathcal{M}_g^{(G)}, X^{(G)}, N^{(G)})$.  The following assumption states our main requirements on $Q_G$ using this notation.
\begin{assumption} \label{ass:QG}
The distribution $Q_G$ is such that \vspace{-.25cm}
\begin{itemize}
\item[(a)] $\{(\mathcal{M}_g,X_g,N_g), 1 \leq g \leq 2G\}$ is an i.i.d.\ sequence of random variables.
\item[(b)] For some family of distributions $\{R(m,x,n) : (m,x,n) \in \text{supp}(\mathcal{M}_g,X_g,N_g)\}$, $$R_G(\mathcal{M}_g^{(G)}, X^{(G)}, N^{(G)}) = \prod_{1 \leq g \leq 2G} R(\mathcal{M}_g,X_g,N_g)~.$$
\item[(c)] $P\{|\mathcal{M}_g| \geq 1\} = 1$ and $E[N_g^2] < \infty$.
\item[(d)] For some $c < \infty$, $P\{E[Y^2_{i,g}(d)| X_g, N_g] \leq c \text{ for all } 1 \leq i \leq N_g \} = 1$ for all $d \in \{0,1\}$ and $1 \leq g \leq 2G$.
\item[(e)] $\mathcal{M}_g \independent (Y_{i,g}(1),Y_{i,g}(0) : 1 \leq i \leq N_g) ~ \big | ~ X_g, N_g$ for all $1 \leq g \leq 2G$.
\item[(f)] For $d \in \{0,1\}$ and $1 \leq g \leq 2G$, $$E[\bar Y_g(d)|N_g] = E\left [\frac{1}{N_g}\sum_{1 \leq i \leq N_g} Y_{i,g}(d) \Big |N_g \right] ~\text{w.p.1} ~.$$
\end{itemize}
\end{assumption}
For completeness, we reproduce some of the observations from \cite{bugni2024inference} regarding these assumptions. Assumptions \ref{ass:QG}(a)--(b) formalize the idea that our sample consists of an i.i.d\ sample of clusters whose cluster sizes are random and potentially related to the potential outcomes. As shown in \cite{bugni2024inference}, an important implication of Assumptions \ref{ass:QG}(a)--(b) for our purposes is that 
\begin{equation} \label{eq:iidclusters}
\left \{(\bar{Y}_g(1), \bar{Y}_g(0), |\mathcal{M}_g|, X_g, N_g \right), 1 \leq g \leq 2G\}~,
\end{equation}
is an i.i.d.\ sequence of random vectors. Assumptions \ref{ass:QG}(c)--(d) impose some mild regularity on the (conditional) moments of the distribution of cluster sizes and potential outcomes, in order to permit the application of relevant laws of large numbers and central limit theorems. Note that Assumption \ref{ass:QG}(c) does not rule out the possibility of observing arbitrarily large clusters, but does place restrictions on the heterogeneity of cluster sizes. For instance, two consequences of Assumptions \ref{ass:QG}(a) and (c) are that
\[\frac{\sum_{1\le g \le G}N_g^2}{\sum_{1\le g\le G}N_g} = O_P(1)~,\]
and
\[\frac{\max_{1\le g \le G}N_g^2}{\sum_{1 \le g \le G}N_g} \xrightarrow{P} 0~,\]
which mirror heterogeneity restrictions imposed in the analysis of clustered data when cluster sizes are modeled as non-random \citep[see for example Assumption 2 in][]{hansen2019asymptotic}. We use Assumption \ref{ass:QG}(c) extensively when establishing asymptotic normality in Theorems \ref{thm:normal_X} and \ref{thm:normal_N}; recent work by \cite{sasaki2022non} and \cite{chiang2023cluster}, however, suggests that one may be able to sometimes obtain asymptotic normality even when $E[N_g^2] = \infty$, provided that certain delicate conditions about the tail behavior of $N_g$ are satisfied. When the tails of the distribution of $N_g$ are so heavy that asymptotic normality fails, it may be possible to extend the recent work on subsampling based inference in \cite{chiang2023cluster} to our setting, but we leave this extension for future work. 

Assumptions \ref{ass:QG}(e)--(f) impose high-level restrictions on the two-stage sampling procedure. Assumption \ref{ass:QG}(e) allows the subset of observations sampled by the experimenter to depend on $X_g$ and $N_g$, but rules out dependence on the potential outcomes within the cluster itself.
Assumption \ref{ass:QG}(f) is a high-level assumption which guarantees that we can extrapolate from the observations that are sampled to the observations that are not sampled. It can be shown that Assumptions \ref{ass:QG}(e)--(f) are satisfied if $\mathcal{M}_g$ is drawn as a random sample without replacement from $\{1, 2, \ldots, N_g\}$ in an appropriate sense \citep[see Lemma 2.1 in][]{bugni2024inference}.

Our object of interest is the size-weighted cluster-level average treatment effect, which may be expressed in our notation as
\[\Delta(Q_G) =E\left[ \frac{N_g}{E[N_g]}\left(\frac{1}{N_g}\sum_{1 \le i \le N_g}(Y_{i,g}(1) - Y_{i,g}(0))\right)\right] = E\left[ \frac{1}{E[N_g]}\sum_{1 \le i \le N_g}(Y_{i,g}(1) - Y_{i,g}(0))\right]~.\]
This parameter, which weights the cluster-level average treatment effects proportional to cluster size, can be thought of as the average treatment effect where individuals are the unit of interest. Note that Assumptions \ref{ass:QG}(a)--(b) imply that we may express $\Delta(Q_G)$ as a function of $R$ and the common distribution of $(\mathcal{M}_g, X_g, N_g)$. In particular, this implies that $\Delta(Q_G)$ does not depend on $G$. Accordingly, in what follows we simply denote $\Delta = \Delta(Q_G)$. 

In Sections \ref{sec:normal}--\ref{sec:rand-test}, we study the asymptotic behavior of the following size-weighted difference-in-means estimator:
\begin{equation} \label{eq:size-weighted}
    \hat{\Delta}_G := \hat{\mu}_G(1) - \hat{\mu}_G(0)~,
\end{equation}
where
\[\hat{\mu}_G(d) := \frac{1}{N(d)}\sum_{1 \le g \le 2G}I\{D_g = d\}\frac{N_g}{|\mathcal{M}_g|}\sum_{i\in \mathcal{M}_g}Y_{i,g}~,\]
with
\[N(d) := \sum_{1 \le g \le 2G}N_gI\{D_g = d\}~.\]
Note that this estimator may be obtained as the estimator of the coefficient of $D_g$ in a weighted least squares regression of $Y_{i,g}$ on a constant and $D_g$ with weights equal to $N_g/|\mathcal{M}_g|$.
In the special case that all observations in each cluster are sampled, so that $\mathcal{M}_g = \{1, 2, \ldots, N_g\}$ for all $1 \le g \le G$ with probability one, this estimator collapses to the standard difference-in-means estimator. However, it is important to note that outside of this special case, the standard difference-in-means estimator is \emph{not} consistent for the size-weighted average treatment effect $\Delta$, and is instead consistent for an ``$|\mathcal{M}_g|$-weighted" treatment effect; see \cite{bugni2024inference} for details. In Section \ref{sec:adjust} we consider a covariate-adjusted modification of $\hat{\Delta}_G$ which is designed to incorporate additional baseline covariates which were not used for treatment assignment. 

%\begin{remark}\label{rem:tuples}
%Throughout the main text we maintain the assumption that the treatment $D_g \in \{0, 1\}$ is binary. In Appendix \ref{sec:supp-tuples} we consider a generalization of the results presented in Sections \ref{sec:normal} and \ref{sec:variance-estimate} to settings with more than two treatments (i.e. matched-\emph{tuple} designs). Examples of clustered matched-tuples experiments include [{\bf REFS HERE }]. 
%\end{remark}

\begin{remark}\label{rem:BFE}
Following the recommendations in \cite{bruhn2009pursuit} and \cite{glennerster2013running}, it is common practice to conduct inference in matched pair experiments using the standard errors obtained from a regression of individual level outcomes on treatment and a collection of pair-level fixed effects. We do not analyze the asymptotic properties of such an approach for two reasons. First, in the context of individual-level randomized experiments, \cite{bai2022inference} and \cite{bai2024inference} argue that such a regression estimator is in fact numerically equivalent to the simple difference-in-means estimator, but that the resulting standard errors are generally conservative (and in some cases possibly invalid). This result generalizes immediately to the clustered setting in the special case where all clusters are the same size and $\mathcal{M}_g = \{1, 2, \ldots, N_g\}$ so that all units in each cluster are sampled. Second, when cluster sizes vary, this numerical equivalence no longer holds, and in such cases \cite{de2019level} argue (in an alternative inferential framework) that the corresponding regression estimator may no longer be consistent for the average treatment effect of interest. 
\end{remark}

\begin{remark}
\cite{bugni2024inference} also define an alternative treatment effect parameter given by
\[\Delta^{\rm eq}(Q_G) = E\left[ \frac{1}{N_g}\sum_{1 \le i \le N_g}(Y_{i,g}(1) - Y_{i,g}(0))\right]~.\]
This parameter, which weights the cluster-level average treatment effects equally regardless of cluster size, can be thought of as the average treatment effect where the clusters themselves are the units of interest. Note that since we do not assume that cluster sizes are ``ignorable," i.e. we allow for the average treatment effect to vary with cluster size, $\Delta^{\rm eq}$ and $\Delta$ are indeed distinct parameters with differing policy implications; see \cite{bugni2024inference} for a detailed discussion and relevant empirical examples. We focus exclusively on the analysis of $\Delta$ for two reasons: first, as discussed further in \cite{bugni2024inference}, we view  $\Delta$ as the parameter most likely to be of practical interest; second, because the analysis of $\Delta^{\rm eq}$ for matched-pair designs follows directly from the analysis for individual-level randomized experiments developed in \cite{bai2022inference}, by applying their results to the data obtained from the cluster-level averages $\{(\bar{Y}_g, D_g, X_g, N_g): 1 \le g \le 2G\}$, where $\bar{Y}_g = \frac{1}{|\mathcal{M}_g|}\sum_{i \in \mathcal{M}_g}Y_{i,g}$. As a result, we do not pursue a detailed description of inference for this parameter in the paper. 
\end{remark}
%With this in mind, we focus exclusively on the analysis of $\Delta$ in this paper for two reasons: first, in our view, $\Delta$ will usually be the primary parameter of interest since this corresponds to a weighting that recovers the average treatment effect where individuals are the units of interest.

\section{Main Results} \label{sec:main}
\subsection{Asymptotic Behavior of $\hat{\Delta}_G$ for Cluster-Matched Pair Designs}\label{sec:normal}
In this section, we consider the asymptotic behavior of $\hat{\Delta}_G$ for two distinct types of cluster-matched pair designs. Section \ref{sec:no_size} studies a setting where cluster size is \emph{not} used as a matching variable when forming pairs. Section \ref{sec:size} considers the setting where we do allow for pairs to be matched based on cluster size in an appropriate sense made formal below.
\subsubsection{Not Matching on Cluster Size}\label{sec:no_size}
In this section, we consider a setting where cluster size is not used as a matching variable. First, we describe our formal assumptions on the mechanism determining treatment assignment. The $G$ pairs of matched clusters may be represented by the sets 
\[\{\pi(2j - 1), \pi(2j)\} \text{ for } j = 1, \dots, G,\]
where $\pi = \pi_G(X^{(G)})$ is a permutation of $2G$ elements, and the right-hand side of this equality emphasizes that, since the permutation represents the result of the matching procedure, it is in fact a function of the cluster-level covariates $X^{(G)}$. Given such a $\pi$, we assume that treatment status is assigned as follows:
%Not here that the dependence of the permutation on on $X^{(G)}$ occurs because, when matching the clusters into pairs, for example using a non-bipartite matching algorithm to minimize the Euclidean distances in terms of covariates, we need to use the covariates of all clusters. 
\begin{assumption}\label{ass:indep_pairsX} 
Treatment status is assigned so that 
\[\left\{\left((Y_{i,g}(1), Y_{i,g}(0): 1 \le i \le N_g), N_g, \mathcal{M}_g\right)\right\}_{g=1}^{2G} \independent D^{(G)} | X^{(G)}~.\]
Conditional on $X^{(G)}$, $(D_{\pi(2j-1)}, D_{\pi(2j)})$, $j = 1, ..., G$ are i.i.d.\ and each uniformly distributed over $\{(0,1), (1,0)\}$.
\end{assumption}

Assumption \ref{ass:indep_pairsX} states that, after pairs are formed according to the baseline covariates, which cluster is treated in a pair is determined by a coin flip independently of all other variables. We further require that the clusters in each pair be ``close" in terms of their baseline covariates in the following sense:
\begin{assumption}\label{ass:pair_formX}
The pairs used in determining treatment assignment satisfy
\[\frac{1}{G}\sum_{1 \leq j \leq G} \|X_{\pi(2j)} - X_{\pi(2j-1)}\|^2 \xrightarrow{P} 0~,\]
as $G \to \infty$.
\end{assumption}
\cite{bai2022inference} provide results which facilitate the construction of pairs which satisfy Assumption \ref{ass:pair_formX}. For instance, if $\mathrm{dim}(X_g) = 1$, then by simply pairing clusters by ordering them from smallest to largest according to $X_g$ and then pairing adjacent clusters, it follows from Theorem 4.1 in \cite{bai2022inference} that Assumption \ref{ass:pair_formX} is satisfied if $E[X_g^2] < \infty$. When $\dim(X_g) > 1$ and a suitable matching procedure is used (for instance the {\tt nbpmatching} package in \texttt{R}), it follows from the discussion in Appendix \ref{sec:match_sufficient} that Assumption \ref{ass:pair_formX} is satisfied when $E[\|X_g\|^d] < \infty$ for $d \geq \mathrm{dim}(X_g) + 1 $ . 

Next, we state the additional assumptions on $Q_G$ we require beyond those stated in Assumption \ref{ass:QG}:
\begin{assumption}\label{ass:DGP2}
The distribution $Q_G$ is such that 
\begin{itemize}
\item[(a)] $E[\bar{Y}_g^r(d)N_g^\ell|X_g = x]$, are Lipschitz for $d \in \{0, 1\}$, $r, \ell \in \{0, 1, 2\}$~,
\item[(b)] For some $C < \infty$, $P\{E[N_g|X_g] \le C\} = 1$~. %{\bf [I THINK this suffices (used to be $N_g^2$)]}
\end{itemize}
\end{assumption}
Assumption \ref{ass:DGP2}(a) is a smoothness requirement analogous to Assumption 2.1(c) in \cite{bai2022inference} that ensures that units within clusters which are ``close" in terms of their baseline covariates are suitably comparable. If $X_g$ is discrete and clusters are matched perfectly in that the distance between pairs in Assumption \ref{ass:pair_formX} is zero, Assumption \ref{ass:DGP2}(a) is not needed. Assumption \ref{ass:DGP2}(b) imposes an additional restriction on the distribution of cluster sizes beyond what is stated in Assumption \ref{ass:QG}(c). Under these assumptions, we obtain the following result:
\begin{theorem}\label{thm:normal_X}
Under Assumptions \ref{ass:QG} and \ref{ass:indep_pairsX}--\ref{ass:DGP2},
\[\sqrt{G}(\hat{\Delta}_G - \Delta) \xrightarrow{d} N(0, \omega^2)\]
as $G \rightarrow \infty$, where 
\[\omega^2 = E[\tilde Y_g^2(1)] + E[\tilde Y_g^2(0)] - \frac{1}{2} E[(E[\tilde{Y}_g(1) + \tilde{Y}_g(0) | X_g])^2]~,\]
with
\[\tilde Y_g(d) = \frac{N_g}{E[N_g]} \left ( \bar{Y}_g(d) - \frac{E[\bar{Y}_g(d) N_g]}{E[N_g]} \right )~.\]
\end{theorem}
The proof of Theorem \ref{thm:normal_X} proceeds by studying the joint distribution of the random numerators and denominators of $\hat{\mu}_G(d)$ for $d \in \{0, 1\}$ using techniques similar to those used in \cite{bai2022inference}, carefully taking into consideration the potential dependence between cluster sizes and outcomes, and then applying the Delta method. Remarkably, the resulting asymptotic variance we obtain in Theorem \ref{thm:normal_X} corresponds exactly to the asymptotic variance of the difference-in-means estimator for matched pairs designs with individual-level assignment \citep[as derived in][]{bai2022inference}, but with transformed cluster-level potential outcomes given by $\tilde{Y}_g(d)$. Accordingly, our result collapses exactly to theirs when $P\{N_g = 1\} = 1$. 

\begin{remark}\label{rem:X_efficient}
Theorem \ref{thm:normal_X} also quantifies the gain in precision obtained from using a matched pairs design versus complete randomization (i.e., assigning half of the clusters to treatment at random): it can be shown that the limiting distribution of $\hat{\Delta}_G$ under complete randomization is given by 
\[\sqrt{G}(\hat{\Delta}_G - \Delta) \xrightarrow{d} N(0, \omega^2_0)~,\]
where $\omega^2_0 = E[\tilde{Y}_g^2(1)]+ E[\tilde{Y}_g^2(0)]$. We thus immediately obtain that $\omega^2 \le \omega^2_0$. Moreover, this inequality is strict unless $E[\tilde{Y}_g(1) + \tilde{Y}_g(0) | X_g] = 0$, which holds for instance when the whole vector of individual potential outcomes, the cluster size, and sampling indicators are independent from $X_g$. This gain in precision echos similar findings for individual-level randomization in \cite{bai2022inference} and \cite{bai2022optimality}.
\end{remark}

\subsubsection{Matching on Cluster Size}\label{sec:size}
In this section, we repeat the exercise in Section \ref{sec:no_size} in a setting where the assignment mechanism matches on baseline characteristics \emph{and} (some function of) cluster size in an appropriate sense to be made formal below. Recall the definition $W_g = (X_g, N_g)$, and let $W^{(G)}:=(W_g: 1 \le g \le 2G)$. First, we describe how to modify our assumptions on the mechanism determining treatment assignment. The $G$ pairs of clusters are still represented by the sets 
\[\{\pi(2j - 1), \pi(2j)\} \text{ for } j = 1, \dots, G~,\]
however, now we allow the permutation $\pi = \pi_G(W^{(G)})$ which determines the pairing to depend on cluster sizes as well as $X^{(G)}$. Given such a $\pi$, we now assume that treatment status is assigned as follows:
\begin{assumption}\label{ass:indep_pairs} 
Treatment status is assigned so that 
\[\{((Y_{i,g}(1), Y_{i,g}(0): 1 \le i \le N_g),\mathcal{M}_g)\}_{g=1}^{2G} \independent D^{(G)} | W^{(G)}~.\]
Conditional on $W^{(G)}$, $(D_{\pi(2g-1)}, D_{\pi(2g)})$, $g = 1, ..., G$ are i.i.d.\ and each uniformly distributed over $\{(0,1), (1,0)\}$.
\end{assumption}

We also require some modifications on our regularity conditions for how pairs are formed and our smoothness requirements on the potential outcomes; we provide further discussion in Remark \ref{rem:N_regular} below:

\begin{assumption}\label{ass:pair_form}
The pairs used in determining treatment assignment satisfy $E[N_g^4] < \infty$ and
\begin{equation} \label{eq:r1234}
   \frac{1}{G}\sum_{1 \leq j \leq G} \|W_{\pi(2j)} - W_{\pi(2j-1)}\|^4 \xrightarrow{P} 0~. 
\end{equation}
\end{assumption}

\begin{assumption}\label{ass:DGP} The distribution $Q_G$ is such that $E[\bar{Y}_g^r(d)|W_g = w]$ are Lipschitz for $d \in \{0, 1\}$, $r \in \{1, 2\}$.
\end{assumption}
    %is not stronger than the ones usually imposed for \eqref{eq:r1234} to hold with $r = 1$ or $r = 2$. In particular, \cite{cytrynbaum2021designing} imposes this condition in order for \eqref{eq:r1234} to hold with $r = 2$.
    % Such algorithms include, for instance, the optimal non-bipartite matching algorithm in the \texttt{nbpmatching} package in R to minimize the left-hand side of \eqref{eq:r1234} for $r = 1$ or $r = 2$.
\begin{remark}\label{rem:match_sufficient}
    We show in Appendix \ref{sec:match_sufficient} that a sufficient condition for \eqref{eq:r1234} when using suitable matching algorithms is that $E[\|W_g\|^d] < \infty$ for some $d \geq \mathrm{dim}(W_g) + 3 = \mathrm{dim}(X_g) + 4$. Note further that if $W_g$ is bounded, then
    \[ \frac{1}{G}\sum_{1 \le j \le G}\|W_{\pi(2j)} - W_{\pi(2j-1)}\|^4 \le C\left(\frac{1}{G}\sum_{1 \le j \le G}\|W_{\pi(2j)} - W_{\pi(2j-1)}\|^2\right)~, \]
    for some constant $C > 0$, and therefore any algorithm that minimizes the right-hand of the above display (for instance, the {\tt nbpmatching} algorithm in \texttt{R}) will satisfy Assumption \ref{ass:pair_form}.
\end{remark}

Under our modified matching procedure and regularity conditions, we obtain the following analog to Theorem \ref{thm:normal_X}:
\begin{theorem}\label{thm:normal_N}
Under Assumptions \ref{ass:QG} and \ref{ass:indep_pairs}--\ref{ass:DGP}, 
\[\sqrt{G}(\hat\Delta_G - \Delta) \xrightarrow{d} N(0, \nu^2)~,\]
as $G \rightarrow \infty$, where 
\begin{equation} \label{eq:limitvar-n}
    \nu^2 = E[\tilde Y_g^2(1)] + E[\tilde Y_g^2(0)] - \frac{1}{2} E[(E[\tilde{Y}_g(1) + \tilde{Y}_g(0) | X_g, N_g])^2]~,
\end{equation}
with
\[\tilde Y_g(d) = \frac{N_g}{E[N_g]} \left ( \bar{Y}_g(d) - \frac{E[\bar{Y}_g(d) N_g]}{E[N_g]} \right )~.\]
\end{theorem}
Note that the asymptotic variance $\nu^2$ has exactly the same form as $\omega^2$ from Section \ref{sec:no_size}, with the only difference being that the final term of the expression conditions on both cluster characteristics $X_g$ and cluster size $N_g$. 

\begin{remark}\label{rem:N_efficient}
Theorem \ref{thm:normal_N} demonstrates the gain in precision obtained from matching on cluster size and cluster characteristics versus simply matching on cluster characteristics, thus formalizing a conjecture presented in \cite{imbens2011experimental}. To see this, note that by comparing $\omega^2$ and $\nu^2$ we obtain that 
\[\omega^2 - \nu^2 = -\frac{1}{2}\left(E[E[\tilde{Y}_g(1) + \tilde{Y}_g(0)|X_g]^2] - E[E[\tilde{Y}_g(1) + \tilde{Y}_g(0)|X_g, N_g]^2]\right)~.\]
It then follows by the law of iterated expectations and Jensen's inequality that $\omega^2 \ge \nu^2$, and the inequality is strict unless $E[\tilde{Y}_g(1) + \tilde{Y}_g(0)|X_g,N_g] = E[\tilde{Y}_g(1) + \tilde{Y}_g(0)|X_g]$ with probability one. A simplified sufficient condition for this to hold is that $E[N_g\bar{Y}_g(d)|X_g, N_g] = E[N_g\bar{Y}_g(d)|X_g]$ for $d \in \{0, 1\}$ and $N_g = E[N_g|X_g]$; the latter condition essentially implying that $N_g$ can be perfectly predicted by $X_g$. Moreover, it can be shown that $\nu^2$ attains the efficiency bound derived in \cite{bai2024efficiency} over a broad class of treatment assignments which maintain that each cluster is treated with marginal probability one-half, including in particular matched pairs as a special case. 
\end{remark}

\begin{remark}\label{rem:N_regular}
We note that Assumptions \ref{ass:pair_formX}--\ref{ass:DGP2} differ from Assumptions \ref{ass:pair_form}--\ref{ass:DGP} because of the special role that $N_g$ plays in the definition of $\Delta$ relative to the other observable characteristics. For instance, we impose Assumption \ref{ass:DGP} instead of \ref{ass:DGP2} to avoid assuming that $E[N_g^2\bar{Y}_g(d)|W_g = w]$ is a Lipschitz function in $w$, which would fail unless $N_g$ were bounded since $N_g$ is part of $W_g$.
\end{remark}

\subsection{Variance Estimation}\label{sec:variance-estimate}
In this section, we construct variance estimators for the asymptotic variances $\omega^2$ and $\nu^2$ obtained in Section \ref{sec:normal}. In fact, we propose a \emph{single} variance estimator that is consistent for \emph{both} $\omega^2$ and $\nu^2$ depending on the nature of the matching procedure. As noted in the discussion following Theorem \ref{thm:normal_X}, the expressions for $\omega^2$ and $\nu^2$ correspond exactly to the asymptotic variance obtained in \cite{bai2022inference} with the individual-level outcome replaced by a cluster-level transformed outcome. We thus follow the variance construction from \cite{bai2022inference}, but replace the individual outcomes with feasible versions of these transformed outcomes. To that end, consider the observed adjusted outcome defined as:
%\begin{align*}
%    \hat Y_{g} = \frac{N_{g}}{ \frac{1}{2G}\sum_{1\leq g \leq 2G} N_g }\left(\bar{Y}_{g}-D_g \left(\frac{\sum_{1\leq g \leq 2G} \bar{Y}_{g} D_g N_g}{\sum_{1\leq g \leq 2G} N_{g}}\right) - (1-D_g)\left(\frac{\sum_{1\leq g \leq 2G} \bar{Y}_{g} (1- D_g) N_g}{\sum_{1\leq g \leq 2G} N_{g}}\right)\right)~,
%\end{align*}
\begin{align*}
    \hat Y_{g} = \frac{N_{g}}{ \frac{1}{2G}\sum_{1\leq j \leq 2G} N_j }\left(\bar{Y}_{g}-\frac{\frac{1}{G}\sum_{1\leq j \leq 2G} \bar{Y}_{j} I\{D_j = D_g\} N_j}{\frac{1}{G}\sum_{1\leq j \leq 2G} I\{D_j = D_g\} N_j}\right)~,
\end{align*}
where
\begin{equation*}
    \bar{Y}_{g} = \frac{1}{|\mathcal{M}_g|} \sum_{i \in \mathcal{M}_g} Y_{i,g}~.
\end{equation*}
We then propose the following variance estimator:
\begin{equation}\label{eqn:define-variance-estimator}
    \hat{v}_{G}^{2}=\hat{\tau}_{G}^{2}-\frac{1}{2}\hat{\lambda}_{G}^{2}~,
\end{equation}
where
\begin{align*}
\hat{\tau}_{G}^{2}=& \frac{1}{G} \sum_{1 \leq j \leq G}\left(\hat Y_{\pi(2 j)}-\hat Y_{\pi(2 j-1)}\right)^{2} \\
\hat{\lambda}_{G}^{2}=& \frac{2}{G} \sum_{1 \leq j \leq\lfloor G/2\rfloor} \left(\hat Y_{\pi(4 j-3)}-\hat Y_{\pi(4 j-2)}\right) \left(\hat Y_{\pi(4 j-1)}- \hat Y_{\pi(4j)} \right)(D_{\pi(4 j-3)}-D_{\pi(4 j-2)})(D_{\pi(4 j-1)}-D_{\pi(4 j)})~.
\end{align*}
Note that the construction of $\hat{v}^2_G$ can be motivated using the same intuition as the variance estimators studied in \cite{bai2022inference} and \cite{bai2024inference}: to consistently estimate quantities like (for instance) $E[E[\tilde{Y}_g(1)|X_g]^2]$ which appear in $\omega^2$, ideally we would like to average over the products of the average outcomes of two treated clusters with similar values of covariates. By construction, however, only one cluster in each pair is treated, and our solution is to instead average across ``pairs of pairs" of clusters. As a consequence, we will additionally require that the matching algorithm satisfy the condition that ``pairs of pairs" of clusters are sufficiently close in terms of their baseline covariates/cluster size, as formalized in the following two assumptions:

\begin{assumption}\label{ass:close-pairs-of-pairsX}
The pairs used in determining treatment status satisfy
\begin{equation*}
    \frac{1}{G} \sum_{1 \leq j \leq\left\lfloor\frac{G}{2}\right\rfloor}\left\|X_{\pi(4 j-k)}-X_{\pi(4 j-\ell)}\right\|^{2} \stackrel{P}{\rightarrow} 0
\end{equation*}
for any $k \in \{2, 3\}$  and $\ell \in \{0, 1\}$.
\end{assumption}

\begin{assumption}\label{ass:close-pairs-of-pairs}
The pairs used in determining treatment status satisfy $E[N_g^4] < \infty$
\begin{equation*}
    \frac{1}{G} \sum_{1 \leq j \leq\left\lfloor\frac{G}{2}\right\rfloor} \left\|W_{\pi(4 j-k)}-W_{\pi(4 j-\ell)}\right\|^{4} \stackrel{P}{\rightarrow} 0
\end{equation*}
for any $k \in \{2, 3\}$  and $\ell \in \{0, 1\}$.
\end{assumption}

As noted in \cite{bai2022inference}, given pairs which satisfy Assumptions \ref{ass:pair_formX} or \ref{ass:pair_form}, it is possible to reorder the pairs so that Assumptions \ref{ass:close-pairs-of-pairsX} or \ref{ass:close-pairs-of-pairs} are satisfied. We then obtain the following two consistency results for the estimator $\hat{v}_G^2$:

\begin{theorem}\label{thm:variance-estimator}
Suppose Assumption \ref{ass:QG} holds. If additionally Assumptions \ref{ass:indep_pairsX}--\ref{ass:DGP2} and \ref{ass:close-pairs-of-pairsX} hold, then
\begin{equation*}
    \hat{v}_{G}^{2} \xrightarrow{P} \omega^2~.
\end{equation*}
Alternatively, if Assumptions \ref{ass:indep_pairs}--\ref{ass:DGP} and \ref{ass:close-pairs-of-pairs} hold, then
\begin{equation*}
    \hat{v}_{G}^{2} \xrightarrow{P} \nu^2~.
\end{equation*}
\end{theorem}

By combining Theorems \ref{thm:normal_X}--\ref{thm:normal_N} with Theorem \ref{thm:variance-estimator}, asymptotically exact tests and confidence intervals can be constructed using a $t$-statistic studentized by $\hat{v}_G$. Next, we derive the limits in probability of two commonly recommended variance estimators obtained from a (weighted) linear regression of the individual-level outcomes $Y_{i,g}$ on a constant and cluster-level treatment $D_g$. The first variance estimator we consider, which we denote by $\hat{\omega}^2_{\rm CR,G}$, is simply the cluster-robust variance estimator of the coefficient of $D_g$ as defined in equation \eqref{eq:omega_CR} in the appendix. Theorem \ref{thm:CR_limit} derives the limit in probability of $\hat{\omega}^2_{\rm CR,G}$ under a matched pair design which matches on baseline covariates as defined in Section \ref{sec:no_size}, and shows that it is generally too large relative to $\omega^2$.

\begin{theorem}\label{thm:CR_limit}
Under Assumptions \ref{ass:QG} and \ref{ass:indep_pairsX}--\ref{ass:DGP2},
\[\hat{\omega}^2_{\rm CR, G} \xrightarrow{P} E[\tilde{Y}_g(1)^2] + E[\tilde{Y}_g(0)^2] \ge \omega^2~,\]
with equality if and only if 
\begin{equation}\label{eq:CR_exact}
E[\tilde{Y}_g(1) + \tilde{Y}_g(0)|X_g] = 0~.
\end{equation}
\end{theorem}

The next variance estimator we consider, which we denote by $\hat{\omega}^2_{\rm PCVE,G}$, is the variance estimator of the coefficient of $D_g$ obtained from clustering on the assignment \emph{pairs} of clusters as defined in equation \eqref{eq:PCVE} in the appendix. \cite{de2019level} call this the pair-cluster variance estimator (PCVE)\footnote{We emphasize, however, that \cite{de2019level} propose their variance estimator in a finite population ``design-based" inferential framework, which is distinct from the superpopulation framework we consider here. In Appendix \ref{sec:sims-finpop} we repeat some of the simulation exercises we consider in Section \ref{sec:sims-unadj} in a design-based framework. There we illustrate that our estimator may have benefits in the design-based framework as well.}. Theorem \ref{thm:PCVE_limit} derives the limit in probability of $\hat{\omega}^2_{\rm PCVE,G}$ in the special case where $N_g = n$ for $g = 1,\ldots, 2G$ for some fixed $n$ and $|\mathcal{M}_g| = N_g$, and shows that it is generally too large relative to $\omega^2$.

\begin{theorem}\label{thm:PCVE_limit}
Suppose Assumptions \ref{ass:QG} and \ref{ass:indep_pairsX}--\ref{ass:DGP2} hold. If in addition we impose that $N_g = n$ for $g = 1, \ldots, 2G$ for some fixed positive integer $n$ and that $|\mathcal{M}_g| = N_g$, then
\[\hat{\omega}^2_{\rm PCVE, G} \xrightarrow{P} \omega^2 + \frac{1}{2}E\left[(E[\tilde{Y}_g(1) - \tilde{Y}_g(0)|X_g])^2\right] \ge \omega^2~,\]
with equality if and only if 
\begin{equation}\label{eq:PCVE_exact}
E[\tilde{Y}_g(1) - \tilde{Y}_g(0)|X_g] = 0~.
\end{equation}
\end{theorem}

Although we do not derive the limit in probability of $\hat{\omega}^2_{\rm PCVE, G}$ in the general case, our simulation evidence in Section \ref{sec:simulations} suggests that the limit of $\hat{\omega}^2_{\rm PCVE, G}$ remains conservative, and that the conditions under which it is consistent for $\omega^2$ are the same as those in equation \eqref{eq:PCVE_exact}. From Theorems \ref{thm:CR_limit} and \ref{thm:PCVE_limit} we obtain that neither cluster-robust standard error is consistent for $\omega^2$ unless the baseline covariates are irrelevant for the potential outcomes in an appropriate sense. In particular, equation \eqref{eq:PCVE_exact} holds when the average treatment difference for the sampled units in a cluster are homogeneous, in the sense that $\bar{Y}_g(1) - \bar{Y}_g(0)$ is constant. We further note that the conditions under which $\hat{\omega}^2_{\rm CR,G}$ and $\hat{\omega}^2_{\rm PCVE,G}$ are consistent for $\omega^2$ are exactly analogous to the conditions under which \cite{bai2022inference} derive (in the setting of an individual-level matched pairs experiment) that the two-sample $t$-test and matched pairs $t$-test are asymptotically exact, respectively.

\subsection{Randomization Tests}\label{sec:rand-test}
In this section, we study the properties of a randomization test based on the idea of permuting the treatment assignments for clusters within pairs. In Section \ref{sec:fin_sample} we present some finite-sample properties of our proposed test, and in Section \ref{sec:large_sample} we establish its large sample validity for testing the null hypothesis $H_0: \Delta = 0$.

First, we define the test. In words, the randomization test constructs its critical value from the empirical distribution of the test statistic obtained by permuting the treatment assignments within pairs. In practice, such a distribution can be approximated by randomly permuting the treatment status of clusters within the same pair: for each pair of clusters, the treatment status of the two clusters remains the same with probability one-half and is flipped otherwise. The test statistic is then calculated based on these permuted treatment assignments and the critical value is determined by the $1 - \alpha$ quantile of resulting distribution of all such permutation statistics. Formally, denote by $\mathbf{H}_G$ the group of all permutations on $2G$ elements and by $\mathbf{H}_G(\pi)$ the subgroup that only permutes elements within pairs defined by $\pi$:
\[\mathbf{H}_G(\pi) = \{h \in \mathbf{H}_G: \{\pi(2j-1), \pi(2j)\} = \{h(\pi(2j-1)), h(\pi(2j))\} \text{ for } 1 \le j \le G\}~.\]
Define the action of $h \in \mathbf{H}_G(\pi)$ on $Z^{(G)}$ as follows:
\[hZ^{(G)} = \{((Y_{i,g}: i \in \mathcal{M}_g), D_{h(g)}, X_g, N_g): 1 \le g \le 2G\}~.\]
The randomization test we consider is then given by
\[\phi_G^{\rm rand}(Z^{(G)}) = I\{T_G(Z^{(G)}) > \hat{R}^{-1}_G(1 - \alpha)\}~,\]
where 
\[\hat{R}_G(t) = \frac{1}{|\mathbf{H}_G(\pi)|}\sum_{h \in \mathbf{H}_G(\pi)}I\{T_G(hZ^{(G)}) \le t\}~,\]
with
\[T_G(Z^{(G)}) = \left|\frac{\sqrt{G}\hat{\Delta}_G}{\hat{v}_G}\right|~.\]
\begin{remark}\label{rem:random_draw}
As is often the case for randomization tests, $\hat{R}_G(t)$ may be difficult to compute in situations where $|\mathbf{H}_G(\pi)| = 2^{G}$ is large. In such cases, we may replace $\mathbf{H}_G(\pi)$ with a stochastic approximation $\hat{\mathbf{H}}_G = \{h_1, h_2, \ldots, h_B\}$, where $h_1$ is the identity transformation and $h_2, \ldots, h_B$ are i.i.d.\ uniform draws from $\mathbf{H}_G(\pi)$. The results in Section \ref{sec:fin_sample} continue to hold with such an approximation; the results in Section \ref{sec:large_sample} continue to hold provided $B \rightarrow \infty$ as $G \rightarrow \infty$.
\end{remark}

\subsubsection{Finite-Sample Results}\label{sec:fin_sample}
In this section we present some finite-sample properties of the proposed test. Consider testing the null hypothesis that the distribution of potential outcomes within a cluster are equal across treatment and control conditional on observable characteristics and cluster size:
\begin{equation}\label{eq:dist_XN}
H_0^{X, N}: (Y_{i,g}(1): 1 \le i \le N_g)|(X_g, N_g) \stackrel{d}{=} (Y_{i,g}(0): 1 \le i \le N_g) | (X_g, N_g)~.
\end{equation}
Note \eqref{eq:dist_XN} is stronger than the statement that the average treatment effect $\Delta = 0$. As a consequence, we are able to establish the following result on the finite sample validity of our randomization test for testing \eqref{eq:dist_XN}:
\begin{theorem}\label{thm:rand_finite}
Suppose Assumption \ref{ass:QG} holds and that the treatment assignment mechanism satisfies Assumption \ref{ass:indep_pairsX} or \ref{ass:indep_pairs}. Then, for the problem of testing \eqref{eq:dist_XN} at level $\alpha \in (0, 1)$, $\phi_G^{\rm rand}(Z^{(G)})$ satisfies
\[E[\phi^{\rm rand}_G(Z^{(G)})] \le \alpha~,\]
under the null hypothesis.
\end{theorem}

\begin{remark}
The proof of Theorem \ref{thm:rand_finite} follows classical arguments that underlie the finite sample validity of randomization tests more generally. Accordingly, as in those arguments, the result continues to hold if the test statistic $T_G$ is replaced by any other test statistic which is a function of $Z^{(G)}$. 
\end{remark}

\subsubsection{Large-Sample Results}\label{sec:large_sample}
In this section, we establish the large-sample validity of the randomization test $\phi_G^{\rm rand}$ for testing the null hypothesis
\begin{equation} \label{eq:delta_zero}
H_0: \Delta = 0~.
\end{equation}
Note \eqref{eq:delta_zero} is implied by \eqref{eq:dist_XN}. In Remark \ref{rem:nonzero_null} we describe how to modify the test for testing non-zero null hypotheses.
\begin{theorem}\label{thm:rand_large}
Suppose $Q_G$ satisfies Assumption \ref{ass:QG}, and either
\begin{itemize}
    \item Assumption \ref{ass:DGP2} with treatment assignment mechanism satisfying Assumption \ref{ass:indep_pairsX} and \ref{ass:close-pairs-of-pairsX}~,
    \item Assumption \ref{ass:DGP} with treatment assignment mechanism satisfying Assumptions \ref{ass:indep_pairs} and \ref{ass:close-pairs-of-pairs}~.
\end{itemize}
Further, suppose that the probability limit of $\hat{v}^2_G$ is positive, then
for the problem of testing \eqref{eq:delta_zero} at level $\alpha \in (0,1)$, $\phi_G^{\rm rand}(Z^{(G)})$ satisfies
\[\lim_{G \rightarrow \infty} E[\phi^{\rm rand}_G(Z^{(G)})] = \alpha~,\]
under the null hypothesis.
\end{theorem}
%\[\sup_{t \in \mathbf{R}}|\hat{R}_G(t) - (\Phi(t) - \Phi(-t))|\xrightarrow{P} 0~,\]
%where $\Phi(\cdot)$ is the standard normal CDF. Thus, 
Theorems \ref{thm:rand_finite} and \ref{thm:rand_large} highlight that the randomization test $\phi^{\rm rand}_G(Z^{(G)})$ is asymptotically valid for testing \eqref{eq:delta_zero} while additionally retaining the finite-sample validity described in Section \ref{sec:fin_sample} under the null hypothesis \eqref{eq:dist_XN}. In Section \ref{sec:sims-unadj} we illustrate the benefit of this additional robustness on the small-sample behavior of $\phi^{\rm rand}_G(Z^{(G)})$ relative to tests constructed using Gaussian critical values. We note that, unlike for the null hypothesis considered in Section \ref{sec:fin_sample}, the choice of test statistic $T_G$ is crucial for establishing Theorem \ref{thm:rand_large}. Similar observations have been made in related contexts in \cite{janssen1997studentized}, \cite{chung2013exact}, \cite{bugni2018inference} and \cite{bai2022inference}. 

\begin{remark}\label{rem:nonzero_null}
We briefly describe how to modify the test $\phi^{\rm rand}_G$ for testing general null hypotheses of the form
\[H_0: \Delta = \Delta_0~.\]
To this end, let 
$$\tilde{Z}^{(G)} := (((Y_{i,g} - D_g\Delta_0 : i \in \mathcal{M}_g), D_g, X_g, N_g) : 1 \leq g \leq 2G )~,$$
then it can be shown that under the assumptions given in Theorem \ref{thm:rand_large}, the test $\phi^{\rm rand}_G(\tilde{Z}^{(G)})$ obtained by replacing $Z^{(G)}$ with $\tilde{Z}^{(G)}$ satisfies
\[\lim_{G \rightarrow \infty} E[\phi^{\rm rand}_G(\tilde{Z}^{(G)})] = \alpha~,\]
under the null hypothesis.
\end{remark}

\subsection{Covariate Adjustment}\label{sec:adjust}
In this section, we consider a linearly covariate-adjusted modification of $\hat{\Delta}_G$ that is designed to improve precision by exploiting additional observed baseline covariates that were not used for treatment assignment. To that end, we consider a setting in which we observe two sets of baseline covariates, $X_g$ and $C_g$, where $X_g \in \mathbf R^k$ denotes the original set of baseline covariates used for treatment assignment, and $C_g \in \mathbf R^{\ell}$ denotes the covariates in addition to $X_g$ that were not used for treatment assignment. Note that $C_g$ could also include cluster-level aggregates of individual-level outcomes, including intracluster means and quantiles. Before proceeding, we note that for the remainder of Section \ref{sec:adjust}, Assumption \ref{ass:QG} should be understood to hold with $(X_g, C_g)$ in place of $X_g$. 
%the assumptions in Section \ref{sec:setup} should be modified such that $X_g$ is replaced by $(X_g, C_g)$ throughout. In particular, references to  below s

Our primary focus will be on settings in which the cluster size $N_g$ is used in determining the pairs.  We note that similar results continue to hold under suitable modifications of our assumptions when $N_g$ is not used in determining pairs by simply replacing $W_g$ with $X_g$ throughout. As in Section \ref{sec:size}, let $\pi = \pi_G(W^{(G)})$ denote the permutation that determines the pairs. We then assume that treatment status is assigned as follows:

\begin{assumption} \label{ass:indep_xz}
Treatment status is assigned so that
\[ \{((Y_{i, g}(1), Y_{i, g}(0): 1 \leq i \leq N_g), \mathcal M_g, C_g)\}_{g = 1}^{2G} \independent D^{(G)} | W^{(G)}~. \]
Conditional on $W^{(G)}$, $(D_{\pi(2g-1)}, D_{\pi(2g)})$, $g = 1, ..., G$ are i.i.d.\ and each uniformly distributed over $\{(0,1), (1,0)\}$.
\end{assumption}
We consider a linearly covariate-adjusted estimator of $\Delta$ based on a set of regressors generated by $W_g$ and $C_g$; define $\psi_g = \psi(W_g, C_g)$, where $\psi: (\mathbf R^k \times \mathbf Z_+) \times \mathbf R^\ell \to \mathbf R^p$. We impose the following assumptions on $\psi$:

\begin{assumption} \label{ass:psi}
    The function $\psi$ is such that
    \begin{enumerate}[(a)]
        \item No component of $\psi$ is a constant and $E[\var[\psi_g | W_g]]$ is nonsingular.
        \item $\var[\psi_g] < \infty$.
        \item For some $c < \infty$, $P \{E[\|\psi_g\|^2 \bar Y_g^2(d) | W_g] \leq c\} = 1$ for $d \in \{0, 1\}$.
        \item $E[\psi_g | W_g = w]$, $E[\psi_g \psi_g' | W_g = w]$, and $E[\psi_g \bar Y_g^r(d) | W_g = w]$ for $d \in \{0, 1\}$ and $r \in \{1, 2\}$ are Lipschitz.
    \end{enumerate}
\end{assumption}

Assumption \ref{ass:psi}(a) implies that none of the components of $\psi_g$ can be perfectly predicted only by $W_g$. Assumptions \ref{ass:psi}(b)--(c) form the counterpart to Assumption \ref{ass:QG}(d), and Assumption \ref{ass:psi}(d) is the counterpart to Assumption \ref{ass:DGP}.

As discussed in \cite{bai2024covariate} and \cite{cytrynbaum2023covariate}, standard covariate adjustments based on a regression using treatment-covariate interactions \citep[see, for instance,][for a succinct treatment]{negi2021revisiting} are not guaranteed to improve efficiency when treatment assignment is not completely randomized. For this reason, we consider a modified version of the adjusted estimator developed in \cite{bai2024covariate} for individual-level matched pair experiments. Let $\hat \beta_G$ denote the OLS estimator of the slope coefficient in the linear regression of $\left ( \frac{1}{2G} \sum_{1 \leq g \leq 2G} N_g \right ) (\hat Y_{\pi(2g - 1)} - \hat Y_{\pi(2g)})(D_{\pi(2g - 1)} - D_{\pi(2g)})$ on a constant and $(\psi_{\pi(2g - 1)} - \psi_{\pi(2g)})(D_{\pi(2g - 1)} - D_{\pi(2g)})$. We then define our covariate-adjusted estimator as
\begin{equation} \label{eq:adj}
\hat \Delta_G^{\rm adj} = \frac{\frac{1}{G} \sum_{1 \leq g \leq 2G} (\bar Y_g N_g - (\psi_g - \bar \psi_G)' \hat \beta_G) D_g}{\frac{1}{G} \sum_{1 \leq g \leq 2G} N_g D_g} - \frac{\frac{1}{G} \sum_{1 \leq g \leq 2G} (\bar Y_g N_g - (\psi_g - \bar \psi_G)' \hat \beta_G) (1 - D_g)}{\frac{1}{G} \sum_{1 \leq g \leq 2G} N_g (1 - D_g)}~,
\end{equation}
where
\[ \bar \psi_G = \frac{1}{2G} \sum_{1 \leq g \leq 2G} \psi_g~. \]
Theorem \ref{thm:adj} derives the limiting distribution of $\hat \Delta_G^{\rm adj}$, and, importantly, it shows that the limiting variance of $\hat \Delta_G^{\rm adj}$ is no larger than that of $\hat \Delta_G$ in \eqref{eq:size-weighted} and is strictly smaller unless $\psi$ is ``irrelevant'' for $\tilde Y_g(1) + \tilde Y_g(0)$ after ``controlling'' for $W_g$, in the sense made precise below.

\begin{theorem} \label{thm:adj}
    Under Assumptions \ref{ass:QG},  \ref{ass:pair_form}, \ref{ass:DGP}, \ref{ass:indep_xz}, and \ref{ass:psi},
    \[ \sqrt G(\hat \Delta_G^{\rm adj} - \Delta) \stackrel{d}{\to} N(0, \varsigma^2) \]
    as $G \to \infty$, where
    \begin{equation}\label{eq:varsigma}
    \varsigma^2 = E[\var[Y_g^\ast(1) | W_g]] + E[\var[Y_g^\ast(0) | W_g]] + \frac{1}{2} E[(E[Y_g^\ast(1) - Y_g^\ast(0) | W_g] - \Delta)^2]~, 
    \end{equation}
    with
    \[ Y_g^\ast(d) = \frac{\bar Y_g(d) N_g - (\psi_g - E[\psi_g])' \beta^\ast}{E[N_g]} - \frac{N_g}{E[N_g]} \frac{E[\bar Y_g(d) N_g - (\psi_g - E[\psi_g])' \beta^\ast]}{E[N_g]} = \tilde Y_g(d) - \frac{(\psi_g - E[\psi_g])' \beta^\ast}{E[N_g]}~, \]
    and
    \begin{equation} \label{eq:betastar}
        \beta^\ast = \left(2 E[\var[\psi_g | W_g]] \right)^{-1} E[  \cov [ \psi_g, \tilde{Y}_g(1) + \tilde{Y}_g(0) | W_g] ] E[N_g]~.
    \end{equation}
    Moreover, 
    \begin{equation*}\label{eq:kappa}
    \varsigma^2 = \nu^2 - \kappa^2~,
    \end{equation*} where $\nu^2$ is as in \eqref{eq:limitvar-n} and
    \[
    \kappa^2 = \frac{2}{E[N_g]^2} E\left[\var\left[\psi_g^\prime \beta^\ast |W_g\right]\right]~.
    \]
    As a consequence, $\varsigma^2 \le \nu^2$, with equality if and only if $\kappa^2 = 0$.
\end{theorem}

Note that the asymptotic variance $\varsigma^2$ has the same form as the variance $\nu^2$, but with new transformed outcomes $Y^*_g(d)$ which can be expressed as covariate-adjusted versions of the original transformed outcomes $\tilde{Y}_g(d)$. Exploiting this observation is what allows us to establish that $\varsigma^2 = \nu^2 - \kappa^2$. As a consequence, we find that the asymptotic variance of $\hat{\Delta}^{\rm adj}_G$ is lower than that of $\hat{\Delta}_G$ whenever the adjustment is appropriately ``relevant," in the sense that $\kappa^2 \ne 0$. 

\begin{remark}
    Although the estimator in \eqref{eq:adj} is closely related to the class of covariate-adjusted estimators in \cite{bai2024covariate}, we cannot directly apply their results in our context because the two denominators in \eqref{eq:adj} are the average cluster sizes of treated and untreated clusters and are therefore random. As a result, unlike in \cite{bai2024covariate}, the demeaning of $\psi$ in \eqref{eq:adj} is crucial for the results in Theorem \ref{thm:adj} to hold. In particular, some remainder terms in the proof of Theorem \ref{thm:adj} are no longer $o_P(1)$ without the demeaning. Moreover, unlike for individual-level experiments, $\hat\Delta_G^{\rm adj}$ cannot be interpreted as the intercept of a linear regression as in \cite{bai2024covariate}.
\end{remark}

For variance estimation, define
\begin{align*}
    \mathring Y_g = \frac{1}{\frac{1}{2G}\sum_{1\leq j \leq 2G} N_j }\left(N_g \bar Y_g- N_g \frac{\frac{1}{G}\sum_{1\leq j \leq 2G} \bar Y_j I\{D_j = D_g\} N_j}{\frac{1}{G}\sum_{1\leq j \leq 2G} I\{D_j = D_g\} N_j} - \psi_g' \hat \beta_G \right)~.
\end{align*}
We then propose the following variance estimator:
\begin{equation}\label{eq:adj-var}
    \mathring \varsigma_G^2 = \mathring \tau_G^2 - \frac{1}{2}\mathring \lambda_G^2~,
\end{equation}
where
\begin{align*}
\mathring \tau_G^2=& \frac{1}{G} \sum_{1 \leq j \leq G}\left(\mathring Y_{\pi(2j)}-\mathring Y_{\pi(2j-1)}\right)^{2} \\
\mathring \lambda_G^2 =& \frac{2}{G} \sum_{1 \leq j \leq\lfloor G/2\rfloor} \left(\mathring Y_{\pi(4j-3)}-\mathring Y_{\pi(4j-2)}\right) \left(\mathring Y_{\pi(4j-1)}- \mathring Y_{\pi(4j)} \right)(D_{\pi(4j-3)}-D_{\pi(4j-2)})(D_{\pi(4j-1)}-D_{\pi(4j)})~.
\end{align*}
The following theorem establishes the consistency of the variance estimator:

\begin{theorem} \label{thm:adj-var}
    Under Assumptions \ref{ass:pair_form}, \ref{ass:DGP}, \ref{ass:close-pairs-of-pairs}, \ref{ass:indep_xz}, and \ref{ass:psi},
    \[ \mathring \varsigma_G^2 \stackrel{P}{\to} \varsigma^2~. \]
\end{theorem}

\section{Simulations} \label{sec:simulations}
\subsection{Unadjusted Estimation}\label{sec:sims-unadj}
In this section, we examine the finite-sample behavior of the estimation and inference procedures considered in Sections \ref{sec:normal}-\ref{sec:rand-test}. We further compare these procedures to tests and confidence intervals constructed using the standard cluster-robust variance estimator (CR) and the pair cluster variance estimator (PCVE) proposed in \cite{de2019level}. For $d \in \{0, 1\}$, $1 \le g \le 2G$, the potential outcomes are generated according to the equation
\[Y_{i,g}(d) = \mu_d(X_g, C_g) + 2\epsilon_{d,i,g}~.\]
Where, in each specification, $(X_g, C_g)$, $g = 1, \ldots, 2G$ are i.i.d.\ with $X_g, C_g\sim\ Beta(2,4)$, and $(\epsilon_{0,i,g}, \epsilon_{1,i,g})$, $g = 1, \ldots, 2G$, $i = 1, \ldots, N_g$ are i.i.d.\ with $\epsilon_{0,i,g}, \epsilon_{1,i,g} \sim N(0,1)$ independently. Note that $C_g$ are additional cluster level covariates which are used to determine the cluster size $N_g$, but are not used directly for matching. Throughout Section \ref{sec:simulations} we assume that we observe the entire cluster, that is, we assume $\mathcal{M}_g = \{1, 2, \ldots, N_g\}$; in Appendix \ref{sec:sims-sample} we repeat the simulation exercise in Section \ref{sec:sims-unadj} for other choices of $\mathcal{M}_g$. We consider the following two specifications for $\mu_d$:

\begin{enumerate}[{\bf Model} 1:]
\item $\mu_{1}(X_g, C_g) = \mu_{0}(X_g, C_g) = 10(X_g-1/3)+6(C_g-1/3)+2$~.
\item $\mu_1(X_g, C_g) = 10(X_g^2-1/7)+6(C_g-1/3)+2$ and $\mu_0(X_g, C_g) = 0$~.
\end{enumerate}

Note that Model 1 satisfies the homogeneity condition in \eqref{eq:PCVE_exact} whereas Model 2 does not. In both cases, $N_g$, $g=1,\dots,2G$ are i.i.d.\ with $N_g \sim\ Binomial(R, C_g)+(500-R)$, where $R$ determines the difference in maximum and minimum cluster sizes. In particular $R$ satisfies the property that $N_g \in [N_{min}, N_{max}]$ with $N_{max}-N_{min}=R$ and we consider $R \in \{49,149,249,349,449\}$ with $N_{max}=500$ fixed. For each model and distribution of cluster sizes, we consider two alternative pair-matching procedures. First, we consider a design which matches clusters using $X_g$ only. To construct these pairs, we sort the clusters according to $X_g$ and pair adjacent clusters. Next, we consider a design which matches clusters using both $X_g$ and $N_g$. To construct these pairs, we match the clusters according to their Mahalanobis distance using the non-bipartite matching algorithm from the R package {\tt nbpMatching}.

Tables \ref{tab:model1X}--\ref{tab:model4XN} report the coverage and average length of $95\%$ confidence intervals constructed using our variance estimator as well as the CR and PCVE estimators. For Model 1 in Table \ref{tab:model1X}, we find that, in accordance with Theorems \ref{thm:variance-estimator}--\ref{thm:PCVE_limit}, the CR variance estimator is extremely conservative, whereas our proposed variance estimator (denoted $\hat{v}_G^2$) and the PCVE variance estimator have exact coverage asymptotically. This feature translates to significantly smaller confidence intervals: on average the confidence intervals constructed using $\hat{v}_G^2$ or PCVE are almost half the length of those constructed using CR when $G \ge 50$. However, the confidence intervals constructed using $\hat{v}_G^2$ or PCVE undercover when $G < 50$. We find similar results when matching on both $X_g$ and $N_g$ in Table \ref{tab:model1XN}. Comparing across Tables \ref{tab:model1X} and \ref{tab:model1XN} we find that, in line with the discussions following Theorems \ref{thm:normal_X} and \ref{thm:normal_N}, matching on $N_g$ in addition to $X_g$ results in a large reduction in the average length of confidence intervals constructed using $\hat{v}_G^2$ (or PCVE), but no change in the average length of confidence intervals constructed using CR.

Moving to Model 2 in Tables \ref{tab:model4X} and \ref{tab:model4XN}, here we find that confidence intervals constructed using CR continue to be conservative, but now the confidence intervals constructed using PCVE are \emph{also} conservative, and numerically very similar to those constructed using CR. In contrast, the confidence intervals constructed using $\hat{v}_G^2$ remain exact asymptotically. Once again this translates to smaller confidence intervals for $\hat{v}_G^2$: on average the confidence intervals constructed using $\hat{v}_G^2$ are approximately $25\%$ smaller than those constructed using CR or PCVE when $G \ge 50$. However, once again we find that the confidence intervals constructed using $\hat{v}^2_G$ can undercover when $G < 50$, with the size of the distortion growing as a function of the cluster size heterogeneity. 

Next, to further address the small-sample coverage distortions observed in Tables \ref{tab:model1X}-\ref{tab:model4XN}, we study the size and power of $0.05$-level hypothesis tests conducted using our proposed randomization test, as well as standard $t$-tests constructed using the CR and PCVE estimators, in Tables \ref{tab:RT_m1}--\ref{tab:RT_m4} below.\footnote{Here we move to studying the properties of hypothesis tests instead of confidence intervals to avoid having to perform test-inversion for our randomization test, but we expect that similar results would continue to hold for confidence intervals as well.} In Table \ref{tab:RT_m1} we find that tests based on the CR variance estimator are extremely conservative, and this translates to having essentially no power against our chosen alternative. Tests based on the PCVE estimator produce non-trivial power, but also size-distortions in small samples. In contrast, since Model 1 satisfies the null hypothesis considered in \eqref{eq:dist_XN}, our randomization test is valid in finite samples by construction, and displays comparable power to the PCVE-based test even when the latter does not control size. When moving to Model 2 in Table \ref{tab:RT_m4} we are only guaranteed that the randomization test is asymptotically valid, but we find that the test is still able to control size in small samples as long as cluster-size heterogeneity is not too large.  Importantly, in such cases, both the CR and PCVE-based tests also fail to control size. Finally, the randomization test displays favorable power relative to both the CR and PCVE-based tests throughout Table \ref{tab:RT_m4} except for some cases when $G = 12$.

\subsection{Covariate-Adjusted Estimation}\label{sec:simulations-adj}
In this section, we examine the finite-sample behavior of the covariate-adjusted estimator considered in Section \ref{sec:adjust}. We consider the following modification of Model 2: let $C_g = (C_{1g}, C_{2g})$,
\begin{enumerate}[{\bf Model} {\rm Adj.}:]
\item $\mu_1(X_g, C_{1g}, C_{2g}) = 10(C_{1g}^2-1/7)+6(C_{2g}-1/3)+25$ and $\mu_0(X_g, C_{1g}, C_{2g}) = 0$~,
\end{enumerate}
 with $X_g \sim U[0,1]$ generated independently of all other variables, and modify the distribution of $N_g$ so that $N_g \sim Binomial(R, 1 - C_{2g}) + (500 - R)$. 
 
Tables \ref{tab:modelH} and \ref{tab:modelHN} report the coverage and average length of $95\%$ confidence intervals constructed using our variance estimators when matching using $X_g$ and both $X_g$ and $N_g$, respectively, for $\hat{\Delta}_G$ versus $\hat{\Delta}^{\rm adj}_G$ with $\psi_g = C_g$. In accordance with Theorem \ref{thm:adj}, we find that for moderate to large samples ($G \ge 50$), covariate adjustment leads to smaller average CI lengths. 
%even as we increase the amount of cluster size heterogeneity. In contrast, Table \ref{tab:modelH} reports the coverage and average lengths of $95\%$ confidence intervals (CIs) constructed using our variance estimators when matching using \emph{only} $X_g$, for $\hat{\Delta}_G$ versus $\hat{\Delta}^{\rm adj}_G$ with $\psi_g = C_g$. In general, we find that when cluster-size heterogeneity is low, covariate adjustment leads to smaller average CI lengths. However, as the amount of heterogeneity increases, the average CI length for the adjusted estimator rapidly overtakes the length for the unadjusted estimator. We emphasize that this does not seem to be a small-sample issue: even with $G = 250$, the average CI length for the adjusted estimator is over two times larger than for the unadjusted estimator in the most extreme case. 

\section{Recommendations for Empirical Practice}\label{sec:recommendations}

Based on our theoretical results as well as the simulation study above, we conclude with some recommendations for practitioners when conducting inference for cluster matched pair designs. The methods in this paper are primarily tailored for inference in a super-population framework; as explained in \cite{bai2024primer}, such a sampling framework may be viewed as an approximation to a regime where a small fraction of the total population of clusters is sampled. Simulation evidence in Appendix \ref{sec:sims-finpop}, however, suggests that our methods compare favorably against existing methods even in finite-population settings. Formal results in a finite population framework can be established by following the general strategy presented in Appendix A.1 in \cite{bai2024primer}.

Our recommendations depend on whether the number of clusters is moderately large (e.g., at least 50 pairs) or small (e.g., less than 50 pairs). If the number of clusters is moderately large, then our recommendation is that practitioners should employ either the covariate-adjusted tests based on the covariate-adjusted estimator $\hat{\Delta}^{\rm adj}_G$ defined in Section \ref{sec:adjust} paired with its corresponding variance estimator $\mathring\varsigma_G^2$ and a normal critical value or the unadjusted tests based on the unadjusted estimator $\hat{\Delta}_G$ introduced in Section \ref{sec:setup} paired with its corresponding variance estimator $\hat{v}^2_G$ and a normal critical value. 

%Whenever cluster size is used in determining the pairs, our results show that covariate-adjusted tests are more powerful in large samples than unadjusted tests; in practice, this feature continues to hold even when cluster size was not used in determining the pairs, provided that cluster-size heterogeneity is not too great (i.e., in our simulations, a ratio of largest to smallest cluster size of less than 2).  Outside of these circumstances, we recommend that practitioners employ the unadjusted tests.

%tests based on the covariate-adjusted estimator $\hat{\Delta}^{\rm adj}_G$ defined in Section \ref{sec:adjust} paired with its corresponding variance estimator $\mathring\varsigma_G^2$ and a normal critical value.  Whenever cluster size is used in determining the pairs, our results show that tests based on these quantities are more powerful in large samples than tests based on their unadjusted counterparts, i.e., the unadjusted estimator $\hat{\Delta}_G$ introduced in Section \ref{sec:setup} paired with its corresponding variance estimator $\hat{v}^2_G$ and normal critical values.  In practice, this feature continues to hold even when cluster was not used in determining the pairs, provided that cluster-size heterogeneity is not too great (i.e., in our simulations, a ratio of largest to smallest cluster size of less than 2).  Outside of these circumstances, we recommend that practitioners employ the unadjusted tests.

If, on the other hand, the number of clusters is small, then we recommend instead that practitioners use the randomization test based on the un-adjusted estimator $\hat{\Delta}_G$ paired with its corresponding variance estimator $\hat{v}^2_G$ outlined in Section \ref{sec:rand-test}. In our simulations, this test controlled size more reliably than any of the other inference procedures we considered in the paper, while delivering comparable power. Note that by modifying the test as in Remark \ref{rem:nonzero_null}, the test could also be inverted to construct confidence intervals if desired.

In general, all of our results crucially hinge on the assumption that clusters in a pair are sufficiently “close” (Assumptions \ref{ass:pair_formX} and \ref{ass:pair_form}), and such a condition becomes difficult to satisfy as the dimension of $X_g$ increases. For this reason, we recommend that practitioners construct their pairs using a small subset of the baseline covariates that they believe have the highest explanatory power (including possibly cluster size itself). The experimental data can then be analyzed by using either the un-adjusted or adjusted estimators we propose in this paper.

%Finally, , we recommend that practitioners do not employ either the cluster-robust or pair-cluster variance estimators.  Our results show that tests based on these quantities are valid in large samples, but potentially conservative, leading to a loss of power relative to the tests described above.

\begin{table}[htbp]\centering
\begin{threeparttable}[b]
\def\sym#1{\ifmmode^{#1}\else\(^{#1}\)\fi}
\caption{Model 1 - Matching on $X_g$\tnote{*}}
\begin{tabular}{c*{9}{c}}
\hline\hline
    \multicolumn{1}{c}{$N_{max}/N_{min}$}&\multicolumn{1}{c}{}
    &\multicolumn{1}{c}{$G=12$}&\multicolumn{1}{c}{$G=26$}&\multicolumn{1}{c}{$G=50$}
    &\multicolumn{1}{c}{$G=100$}&\multicolumn{1}{c}{$G=150$}
    &\multicolumn{1}{c}{$G=200$}&\multicolumn{1}{c}{$G=250$}\\
\hline
\multicolumn{9}{c}{} \\
\multicolumn{9}{c}{\bf Coverage} \\
\hline
\multirow{3}{*}{1.11}
& $\hat{v}_G^2$ & $0.9185$ & $0.9290$ & $0.9420$ & $0.9465$ & $0.9375$ & $0.9460$ & $0.9515$ \\
& CR & $0.9985$ & $0.9990$ & $0.9995$ & $1$ & $1$ & $1$ & $1$ \\
& PCVE & $0.9230$ & $0.9310$ & $0.9385$ & $0.9405$ & $0.9395$ & $0.9480$ & $0.9520$ \\
\multicolumn{9}{l}{} \\
\multirow{3}{*}{1.42}
& $\hat{v}_G^2$ & $0.9005$ & $0.9345$ & $0.9345$ & $0.9480$ & $0.9490$ & $0.9545$ & $0.9615$ \\
& CR & $0.9980$ & $0.9995$ & $0.9985$ & $0.9995$ & $0.9995$ & $1$ & $1$ \\
& PCVE & $0.9035$ & $0.9380$ & $0.9375$ & $0.9490$ & $0.9495$ & $0.9550$ & $0.9595$ \\
\multicolumn{9}{l}{} \\
\multirow{3}{*}{1.99}
& $\hat{v}_G^2$ & $0.9130$ & $0.9330$ & $0.9380$ & $0.9385$ & $0.9490$ & $0.9455$ & $0.9365$ \\
& CR & $0.9985$ & $0.9985$ & $1$ & $1$ & $1$ & $1$ & $0.9995$ \\
& PCVE & $0.9095$ & $0.9230$ & $0.9420$ & $0.9420$ & $0.9495$ & $0.9460$ & $0.9350$ \\
\multicolumn{9}{l}{} \\
\multirow{3}{*}{3.31}
& $\hat{v}_G^2$ & $0.9065$ & $0.9180$ & $0.9340$ & $0.9415$ & $0.9470$ & $0.9450$ & $0.9520$ \\
& CR & $0.9950$ & $0.9980$ & $0.9980$ & $0.9985$ & $1$ & $0.9985$ & $0.9995$ \\
& PCVE & $0.8980$ & $0.9155$ & $0.9330$ & $0.9380$ & $0.9465$ & $0.9470$ & $0.9500$ \\
\multicolumn{9}{l}{} \\
\multirow{3}{*}{9.80}
& $\hat{v}_G^2$ & $0.9035$ & $0.9230$ & $0.9420$ & $0.9340$ & $0.9440$ & $0.9415$ & $0.9495$ \\
& CR & $0.9925$ & $0.9940$ & $0.9970$ & $0.9985$ & $0.9975$ & $0.9995$ & $0.9990$ \\
& PCVE & $0.8925$ & $0.9100$ & $0.9365$ & $0.9330$ & $0.9425$ & $0.9385$ & $0.9475$ \\
\hline
\multicolumn{9}{c}{} \\
\multicolumn{9}{c}{{\bf Average Length}} \\
\hline
\multirow{3}{*}{1.11}
& $\hat{v}_G^2$ & $1.72150$ & $1.16078$ & $0.84582$ & $0.59830$ & $0.48784$ & $0.42466$ & $0.37936$ \\
& CR & $3.20593$ & $2.21689$ & $1.61886$ & $1.15015$ & $0.94053$ & $0.81591$ & $0.73010$ \\
& PCVE & $1.69494$ & $1.15171$ & $0.84119$ & $0.59746$ & $0.48744$ & $0.42415$ & $0.37895$ \\
\multicolumn{9}{l}{} \\
\multirow{3}{*}{1.42}
& $\hat{v}_G^2$ & $1.75019$ & $1.18859$ & $0.86476$ & $0.61378$ & $0.50112$ & $0.43567$ & $0.38917$ \\
& CR & $3.21821$ & $2.22957$ & $1.62982$ & $1.15829$ & $0.94732$ & $0.82180$ & $0.73543$ \\
& PCVE & $1.72075$ & $1.17840$ & $0.86140$ & $0.61286$ & $0.50024$ & $0.43527$ & $0.38897$ \\
\multicolumn{9}{l}{} \\
\multirow{3}{*}{1.99}
& $\hat{v}_G^2$ & $1.80502$ & $1.23175$ & $0.89937$ & $0.63958$ & $0.52250$ & $0.45322$ & $0.40566$ \\
& CR & $3.24165$ & $2.25077$ & $1.64811$ & $1.17207$ & $0.95862$ & $0.83166$ & $0.74408$ \\
& PCVE & $1.77287$ & $1.21936$ & $0.89602$ & $0.63843$ & $0.52133$ & $0.45352$ & $0.40524$ \\
\multicolumn{9}{l}{} \\
\multirow{3}{*}{3.31}
& $\hat{v}_G^2$ & $1.90111$ & $1.30589$ & $0.96060$ & $0.68446$ & $0.55910$ & $0.48664$ & $0.43505$ \\
& CR & $3.27892$ & $2.28895$ & $1.68064$ & $1.19654$ & $0.97928$ & $0.84959$ & $0.76030$ \\
& PCVE & $1.85679$ & $1.29128$ & $0.95566$ & $0.68299$ & $0.55824$ & $0.48568$ & $0.43437$ \\
\multicolumn{9}{l}{} \\
\multirow{3}{*}{9.80}
& $\hat{v}_G^2$ &  $2.09510$ & $1.45719$ & $1.08057$ & $0.77340$ & $0.63320$ & $0.55071$ & $0.49226$ \\
& CR & $3.35580$ & $2.36729$ & $1.75068$ & $1.24963$ & $1.02275$ & $0.88759$ & $0.79443$ \\
& PCVE & $2.03228$ & $1.43576$ & $1.07565$ & $0.77259$ & $0.63171$ & $0.54976$ & $0.49203$ \\
\hline\hline
\end{tabular}
\label{tab:model1X}
\begin{tablenotes}
\item [*] Number of clusters $=2G$ with $G=12, 26, 50, 100, 150, 200, 250$. Number of replications for each $G$ is $2000$. $N_{max}=500$.
\end{tablenotes}
\end{threeparttable}
\end{table}

\begin{table}[htbp]\centering
\begin{threeparttable}[b]
\def\sym#1{\ifmmode^{#1}\else\(^{#1}\)\fi}
\caption{Model 1 - Matching on $X_g$ and $N_g$\tnote{*}}
\begin{tabular}{c*{9}{c}}
\hline\hline
    \multicolumn{1}{c}{$N_{max}/N_{min}$}&\multicolumn{1}{c}{}
    &\multicolumn{1}{c}{$G=12$}&\multicolumn{1}{c}{$G=26$}&\multicolumn{1}{c}{$G=50$}
    &\multicolumn{1}{c}{$G=100$}&\multicolumn{1}{c}{$G=150$}
    &\multicolumn{1}{c}{$G=200$}&\multicolumn{1}{c}{$G=250$}\\
\hline
\multicolumn{9}{c}{} \\
\multicolumn{9}{c}{\bf Coverage} \\
\hline
\multirow{3}{*}{1.11}
& $\hat{v}_G^2$ & $0.9105$ & $0.9285$ & $0.9345$ & $0.9430$ & $0.9470$ & $0.9495$ & $0.9565$ \\
& CR & $1$ & $1$ & $1$ & $1$ & $1$ & $1$ & $1$ \\
& PCVE & $0.9100$ & $0.9260$ & $0.9360$ & $0.9460$ & $0.9460$ & $0.9480$ & $0.9555$ \\
\multicolumn{9}{l}{} \\
\multirow{3}{*}{1.42}
& $\hat{v}_G^2$ & $0.9210$ & $0.9410$ & $0.9400$ & $0.9510$ & $0.9490$ & $0.9300$ & $0.9445$ \\
& CR & $1$ & $1$ & $1$ & $1$ & $1$ & $1$ & $1$ \\
& PCVE & $0.9215$ & $0.9405$ & $0.9425$ & $0.9555$ & $0.9465$ & $0.9325$ & $0.9425$ \\
\multicolumn{9}{l}{} \\
\multirow{3}{*}{1.99}
& $\hat{v}_G^2$ & $0.9170$ & $0.9460$ & $0.9420$ & $0.9505$ & $0.9485$ & $0.9495$ & $0.9570$ \\
& CR & $1$ & $1$ & $1$ & $1$ & $1$ & $1$ & $1$ \\
& PCVE & $0.9110$ & $0.9440$ & $0.9395$ & $0.9520$ & $0.9490$ & $0.9510$ & $0.9555$ \\
\multicolumn{9}{l}{} \\
\multirow{3}{*}{3.31}
& $\hat{v}_G^2$ & $0.9220$ & $0.9280$ & $0.9295$ & $0.9430$ & $0.9440$ & $0.9480$ & $0.9390$ \\
& CR & $1$ & $1$ & $1$ & $1$ & $1$ & $1$ & $1$ \\
& PCVE & $0.9150$ & $0.9290$ & $0.9325$ & $0.9470$ & $0.9435$ & $0.9510$ & $0.9405$ \\
\multicolumn{9}{l}{} \\
\multirow{3}{*}{9.80}
& $\hat{v}_G^2$ & $0.9015$ & $0.9260$ & $0.9320$ & $0.9505$ & $0.9485$ & $0.9405$ & $0.9435$ \\
& CR & $1$ & $1$ & $1$ & $1$ & $1$ & $1$ & $1$ \\
& PCVE & $0.8860$ & $0.9225$ & $0.9380$ & $0.9495$ & $0.9485$ & $0.9420$ & $0.9475$ \\
\hline
\multicolumn{9}{c}{} \\
\multicolumn{9}{c}{{\bf Average Length}} \\
\hline
\multirow{3}{*}{1.11}
& $\hat{v}_G^2$ & $1.20496$ & $0.64428$ & $0.39514$ & $0.24765$ & $0.19157$ & $0.16045$ & $0.14069$ \\
& CR & $3.21594$ & $2.22170$ & $1.62079$ & $1.15081$ & $0.94092$ & $0.81621$ & $0.73031$ \\
& PCVE & $1.18192$ & $0.63873$ & $0.39376$ & $0.24689$ & $0.19111$ & $0.16028$ & $0.14062$ \\
\multicolumn{9}{l}{} \\
\multirow{3}{*}{1.42}
& $\hat{v}_G^2$ & $1.16805$ & $0.58866$ & $0.34117$ & $0.19821$ & $0.14670$ & $0.12020$ & $0.10335$ \\
& CR & $3.23229$ & $2.23499$ & $1.63182$ & $1.15901$ & $0.94776$ & $0.82214$ & $0.73561$ \\
& PCVE & $1.14574$ & $0.58388$ & $0.34065$ & $0.19783$ & $0.14622$ & $0.12000$ & $0.10327$ \\
\multicolumn{9}{l}{} \\
\multirow{3}{*}{1.99}
& $\hat{v}_G^2$ & $1.18988$ & $0.60685$ & $0.34699$ & $0.19474$ & $0.14244$ & $0.11466$ & $0.09729$ \\
& CR & $3.25786$ & $2.25761$ & $1.65083$ & $1.17312$ & $0.95917$ & $0.83201$ & $0.74440$ \\
& PCVE & $1.16373$ & $0.59889$ & $0.34582$ & $0.19426$ & $0.14229$ & $0.11456$ & $0.09728$ \\
\multicolumn{9}{l}{} \\
\multirow{3}{*}{3.31}
& $\hat{v}_G^2$ & $1.27089$ & $0.64963$ & $0.37337$ & $0.20857$ & $0.15167$ & $0.12110$ & $0.10157$ \\
& CR & $3.29929$ & $2.29885$ & $1.68464$ & $1.19841$ & $0.98016$ & $0.85013$ & $0.76067$ \\
& PCVE & $1.23316$ & $0.64188$ & $0.37129$ & $0.20767$ & $0.15108$ & $0.12084$ & $0.10134$ \\
\multicolumn{9}{l}{} \\
\multirow{3}{*}{9.80}
& $\hat{v}_G^2$ & $1.41981$ & $0.75053$ & $0.43329$ & $0.24285$ & $0.17464$ & $0.13851$ & $0.11558$ \\
& CR & $3.38816$ & $2.38329$ & $1.75642$ & $1.25248$ & $1.02442$ & $0.88868$ & $0.79508$ \\
& PCVE & $1.36449$ & $0.73612$ & $0.42992$ & $0.24197$ & $0.17401$ & $0.13826$ & $0.11549$ \\
\hline\hline
\end{tabular}
\label{tab:model1XN}
\begin{tablenotes}
\item [*] Number of clusters $=2G$ with $G=12, 26, 50, 100, 150, 200, 250$. Number of replications for each $G$ is $2000$. $N_{max}=500$.
\end{tablenotes}
\end{threeparttable}
\end{table}

\begin{table}[htbp]\centering
\begin{threeparttable}[b]
\def\sym#1{\ifmmode^{#1}\else\(^{#1}\)\fi}
\caption{Model 2 - Matching on $X_g$\tnote{*}}
\begin{tabular}{c*{9}{c}}
\hline\hline
    \multicolumn{1}{c}{$N_{max}/N_{min}$}&\multicolumn{1}{c}{}
    &\multicolumn{1}{c}{$G=12$}&\multicolumn{1}{c}{$G=26$}&\multicolumn{1}{c}{$G=50$}
    &\multicolumn{1}{c}{$G=100$}&\multicolumn{1}{c}{$G=150$}
    &\multicolumn{1}{c}{$G=200$}&\multicolumn{1}{c}{$G=250$}\\
\hline
\multicolumn{9}{c}{} \\
\multicolumn{9}{c}{\bf Coverage} \\
\hline
\multirow{3}{*}{1.11}
& $\hat{v}_G^2$ & $0.9260$ & $0.9375$ & $0.9420$ & $0.9420$ & $0.9460$ & $0.9465$ & $0.9510$ \\
& CR & $0.9570$ & $0.9635$ & $0.9755$ & $0.9790$ & $0.9825$ & $0.9835$ & $0.9800$ \\
& PCVE & $0.9560$ & $0.9645$ & $0.9750$ & $0.9785$ & $0.9825$ & $0.9835$ & $0.9805$ \\
\multicolumn{9}{l}{} \\
\multirow{3}{*}{1.42}
& $\hat{v}_G^2$ & $0.9280$ & $0.9395$ & $0.9455$ & $0.9405$ & $0.9490$ & $0.9495$ & $0.9490$ \\
& CR & $0.9525$ & $0.9705$ & $0.9705$ & $0.9715$ & $0.9795$ & $0.9860$ & $0.9820$ \\
& PCVE & $0.9535$ & $0.9710$ & $0.9705$ & $0.9735$ & $0.9795$ & $0.9860$ & $0.9820$ \\
\multicolumn{9}{l}{} \\
\multirow{3}{*}{1.99}
& $\hat{v}_G^2$ & $0.9180$ & $0.9325$ & $0.9385$ & $0.9455$ & $0.9480$ & $0.9420$ & $0.9465$ \\
& CR & $0.9415$ & $0.9595$ & $0.9680$ & $0.9765$ & $0.9770$ & $0.9805$ & $0.9800$ \\
& PCVE & $0.9415$ & $0.9605$ & $0.9675$ & $0.9770$ & $0.9780$ & $0.9800$ & $0.9805$ \\
\multicolumn{9}{l}{} \\
\multirow{3}{*}{3.31}
& $\hat{v}_G^2$ & $0.8965$ & $0.9290$ & $0.9390$ & $0.9480$ & $0.9440$ & $0.9400$ & $0.9495$ \\
& CR & $0.9325$ & $0.9615$ & $0.9700$ & $0.9750$ & $0.9775$ & $0.9750$ & $0.9765$ \\
& PCVE & $0.9315$ & $0.9615$ & $0.9685$ & $0.9755$ & $0.9780$ & $0.9745$ & $0.9770$ \\
\multicolumn{9}{l}{} \\
\multirow{3}{*}{9.80}
& $\hat{v}_G^2$ & $0.8850$ & $0.9085$ & $0.9295$ & $0.9380$ & $0.9360$ & $0.9375$ & $0.9445$ \\
& CR & $0.9155$ & $0.9460$ & $0.9640$ & $0.9660$ & $0.9660$ & $0.9685$ & $0.9755$ \\
& PCVE & $0.9175$ & $0.9450$ & $0.9635$ & $0.9660$ & $0.9665$ & $0.9680$ & $0.9755$ \\
\hline
\multicolumn{9}{c}{} \\
\multicolumn{9}{c}{{\bf Average Length}} \\
\hline
\multirow{3}{*}{1.11}
& $\hat{v}_G^2$ & $1.64579$ & $1.11414$ & $0.80852$ & $0.57317$ & $0.46677$ & $0.40525$ & $0.36269$ \\
& CR & $1.88285$ & $1.31397$ & $0.96438$ & $0.68747$ & $0.56044$ & $0.48713$ & $0.43634$ \\
& PCVE & $1.88367$ & $1.31373$ & $0.96432$ & $0.68752$ & $0.56044$ & $0.48718$ & $0.43636$ \\
\multicolumn{9}{l}{} \\
\multirow{3}{*}{1.42}
& $\hat{v}_G^2$ & $1.67055$ & $1.13171$ & $0.81934$ & $0.58015$ & $0.47436$ & $0.41154$ & $0.36739$ \\
& CR & $1.90602$ & $1.32885$ & $0.97303$ & $0.69262$ & $0.56755$ & $0.49258$ & $0.44032$ \\
& PCVE & $1.90579$ & $1.32897$ & $0.97283$ & $0.69257$ & $0.56751$ & $0.49262$ & $0.44026$ \\
\multicolumn{9}{l}{} \\
\multirow{3}{*}{1.99}
& $\hat{v}_G^2$ & $1.67377$ & $1.14094$ & $0.83413$ & $0.59068$ & $0.48377$ & $0.41909$ & $0.37493$ \\
& CR & $1.90337$ & $1.33455$ & $0.98635$ & $0.70162$ & $0.57506$ & $0.49879$ & $0.44584$ \\
& PCVE & $1.90395$ & $1.33471$ & $0.98606$ & $0.70146$ & $0.57506$ & $0.49874$ & $0.44586$ \\
\multicolumn{9}{l}{} \\
\multirow{3}{*}{3.31}
& $\hat{v}_G^2$ & $1.69386$ & $1.16940$ & $0.85636$ & $0.61062$ & $0.49954$ & $0.43424$ & $0.38770$ \\
& CR & $1.91395$ & $1.35515$ & $1.00133$ & $0.71846$ & $0.58755$ & $0.51145$ & $0.45702$ \\
& PCVE & $1.91241$ & $1.35461$ & $1.00137$ & $0.71861$ & $0.58755$ & $0.51149$ & $0.45699$ \\
\multicolumn{9}{l}{} \\
\multirow{3}{*}{9.80}
& $\hat{v}_G^2$ & $1.74999$ & $1.23124$ & $0.90607$ & $0.64424$ & $0.52971$ & $0.45990$ & $0.41091$ \\
& CR & $1.95803$ & $1.40591$ & $1.04446$ & $0.74668$ & $0.61421$ & $0.53318$ & $0.47665$ \\
& PCVE & $1.95767$ & $1.40633$ & $1.04420$ & $0.74671$ & $0.61422$ & $0.53315$ & $0.47665$ \\
\hline\hline
\end{tabular}
\label{tab:model4X}
\begin{tablenotes}
\item [*] Number of clusters $=2G$ with $G=12, 26, 50, 100, 150, 200, 250$. Number of replications for each $G$ is $2000$. $N_{max}=500$.
\end{tablenotes}
\end{threeparttable}
\end{table}

\begin{table}[htbp]\centering
\begin{threeparttable}[b]
\def\sym#1{\ifmmode^{#1}\else\(^{#1}\)\fi}
\caption{Model 2 - Matching on $X_g$ and $N_g$\tnote{*}}
\begin{tabular}{c*{9}{c}}
\hline\hline
    \multicolumn{1}{c}{$N_{max}/N_{min}$}&\multicolumn{1}{c}{}
    &\multicolumn{1}{c}{$G=12$}&\multicolumn{1}{c}{$G=26$}&\multicolumn{1}{c}{$G=50$}
    &\multicolumn{1}{c}{$G=100$}&\multicolumn{1}{c}{$G=150$}
    &\multicolumn{1}{c}{$G=200$}&\multicolumn{1}{c}{$G=250$}\\
\hline
\multicolumn{9}{c}{} \\
\multicolumn{9}{c}{\bf Coverage} \\
\hline
\multirow{3}{*}{1.11}
& $\hat{v}_G^2$ & $0.9420$ & $0.9480$ & $0.9545$ & $0.9495$ & $0.9455$ & $0.9530$ & $0.9530$ \\
& CR & $0.9670$ & $0.9845$ & $0.9875$ & $0.9900$ & $0.9915$ & $0.9950$ & $0.9935$ \\
& PCVE & $0.9680$ & $0.9850$ & $0.9865$ & $0.9900$ & $0.9910$ & $0.9950$ & $0.9935$ \\
\multicolumn{9}{l}{} \\
\multirow{3}{*}{1.42}
& $\hat{v}_G^2$ & $0.9315$ & $0.9475$ & $0.9515$ & $0.9530$ & $0.9515$ & $0.9580$ & $0.9510$ \\
& CR & $0.9665$ & $0.9850$ & $0.9850$ & $0.9895$ & $0.9915$ & $0.9955$ & $0.9955$ \\
& PCVE & $0.9660$ & $0.9850$ & $0.9845$ & $0.9900$ & $0.9915$ & $0.9960$ & $0.9955$ \\
\multicolumn{9}{l}{} \\
\multirow{3}{*}{1.99}
& $\hat{v}_G^2$ & $0.9270$ & $0.9430$ & $0.9510$ & $0.9520$ & $0.9480$ & $0.9575$ & $0.9520$ \\
& CR & $0.9650$ & $0.9825$ & $0.9885$ & $0.9905$ & $0.9930$ & $0.9970$ & $0.9945$ \\
& PCVE & $0.9670$ & $0.9815$ & $0.9880$ & $0.9900$ & $0.9930$ & $0.9970$ & $0.9945$ \\
\multicolumn{9}{l}{} \\
\multirow{3}{*}{3.31}
& $\hat{v}_G^2$ & $0.9160$ & $0.9365$ & $0.9525$ & $0.9480$ & $0.9510$ & $0.9525$ & $0.9485$ \\
& CR & $0.9580$ & $0.9795$ & $0.9890$ & $0.9885$ & $0.9930$ & $0.9955$ & $0.9940$ \\
& PCVE & $0.9580$ & $0.9800$ & $0.9890$ & $0.9890$ & $0.9930$ & $0.9955$ & $0.9940$ \\
\multicolumn{9}{l}{} \\
\multirow{3}{*}{9.80}
& $\hat{v}_G^2$ & $0.9065$ & $0.9330$ & $0.9430$ & $0.9510$ & $0.9515$ & $0.9495$ & $0.9510$ \\
& CR & $0.9410$ & $0.9765$ & $0.9845$ & $0.9890$ & $0.9880$ & $0.9955$ & $0.9915$ \\
& PCVE & $0.9430$ & $0.9755$ & $0.9830$ & $0.9890$ & $0.9875$ & $0.9955$ & $0.9915$ \\
\hline
\multicolumn{9}{c}{} \\
\multicolumn{9}{c}{{\bf Average Length}} \\
\hline
\multirow{3}{*}{1.11}
& $\hat{v}_G^2$ & $1.57502$ & $1.02869$ & $0.73036$ & $0.51031$ & $0.41388$ & $0.35765$ & $0.31902$ \\
& CR & $1.89796$ & $1.31976$ & $0.96665$ & $0.68810$ & $0.56233$ & $0.48793$ & $0.43636$ \\
& PCVE & $1.89800$ & $1.31982$ & $0.96657$ & $0.68813$ & $0.56236$ & $0.48790$ & $0.43634$ \\
\multicolumn{9}{l}{} \\
\multirow{3}{*}{1.42}
& $\hat{v}_G^2$ & $1.58361$ & $1.03237$ & $0.73193$ & $0.50975$ & $0.41335$ & $0.35758$ & $0.31856$ \\
& CR & $1.91602$ & $1.33100$ & $0.97594$ & $0.69418$ & $0.56753$ & $0.49302$ & $0.44052$ \\
& PCVE & $1.91549$ & $1.33128$ & $0.97597$ & $0.69423$ & $0.56756$ & $0.49301$ & $0.44049$ \\
\multicolumn{9}{l}{} \\
\multirow{3}{*}{1.99}
& $\hat{v}_G^2$ & $1.61080$ & $1.04567$ & $0.74313$ & $0.51722$ & $0.41903$ & $0.36217$ & $0.32297$ \\
& CR & $1.93406$ & $1.34395$ & $0.98875$ & $0.70392$ & $0.57534$ & $0.49967$ & $0.44684$ \\
& PCVE & $1.93403$ & $1.34409$ & $0.98881$ & $0.70388$ & $0.57529$ & $0.49964$ & $0.44680$ \\
\multicolumn{9}{l}{} \\
\multirow{3}{*}{3.31}
& $\hat{v}_G^2$ & $1.63660$ & $1.07550$ & $0.76774$ & $0.53170$ & $0.43114$ & $0.37227$ & $0.33175$ \\
& CR & $1.94629$ & $1.37114$ & $1.01341$ & $0.72038$ & $0.58976$ & $0.51183$ & $0.45771$ \\
& PCVE & $1.94802$ & $1.37098$ & $1.01337$ & $0.72047$ & $0.58984$ & $0.51198$ & $0.45771$ \\
\multicolumn{9}{l}{} \\
\multirow{3}{*}{9.80}
& $\hat{v}_G^2$ & $1.70687$ & $1.13039$ & $0.80947$ & $0.55966$ & $0.45337$ & $0.39151$ & $0.34801$ \\
& CR & $1.98400$ & $1.41410$ & $1.05392$ & $0.75111$ & $0.61528$ & $0.53484$ & $0.47768$ \\
& PCVE & $1.98403$ & $1.41488$ & $1.05356$ & $0.75103$ & $0.61532$ & $0.53482$ & $0.47769$ \\
\hline\hline
\end{tabular}
\label{tab:model4XN}
\begin{tablenotes}
\item [*] Number of clusters $=2G$ with $G=12, 26, 50, 100, 150, 200, 250$. Number of replications for each $G$ is $2000$. $N_{max}=500$.
\end{tablenotes}
\end{threeparttable}
\end{table}

\begin{table}[htbp]\centering
\begin{threeparttable}[b]
\def\sym#1{\ifmmode^{#1}\else\(^{#1}\)\fi}
\caption{Model 1 - Randomization Test (RT) vs. CR/PCVE \tnote{*}}\label{tab:RT_m1}
\begin{tabular}{ccccc|ccc}
\hline\hline
    &&\multicolumn{3}{c|}{Size under $H_0$}&\multicolumn{3}{c}{Power under $H_1: \Delta_0 + 1/4$} \\
    \multicolumn{1}{c}{$N_{max}/N_{min}$}&\multicolumn{1}{c}{}
    &\multicolumn{1}{c}{$G=12$}&\multicolumn{1}{c}{$G=26$}&\multicolumn{1}{c|}{$G=50$}
    &\multicolumn{1}{c}{$G=12$}&\multicolumn{1}{c}{$G=26$}&\multicolumn{1}{c}{$G=50$}\\
\hline
\multicolumn{8}{c}{} \\
\multicolumn{8}{c}{\bf Matching on $X_g$} \\
\hline
\multirow{3}{*}{1.11}
& RT & $0.0395$ & $0.0560$ & $0.0505$ & $0.0755$ & $0.1220$ & $0.2030$  \\
& CR & $0.0015$ & $0.0010$ & $0.0005$ & $0.0095$ & $0.0105$ & $0.0160$   \\
& PCVE & $0.0770$ & $0.0690$ & $0.0615$ & $0.1195$ & $0.1410$ & $0.1995$  \\
\multicolumn{5}{l|}{}&\multicolumn{3}{l}{} \\
\multirow{3}{*}{1.42}
& RT & $0.0610$ & $0.0445$ & $0.0540$ & $0.0935$ & $0.1055$ & $0.1970$  \\
& CR & $0.0020$ & $0.0005$ & $0.0015$ & $0.0105$ & $0.0105$ & $0.0210$  \\
& PCVE & $0.0965$ & $0.0620$ & $0.0625$ & $0.1365$ & $0.1220$ & $0.1955$  \\
\multicolumn{5}{l|}{}&\multicolumn{3}{l}{} \\
\multirow{3}{*}{1.99}
& RT & $0.0505$ & $0.0505$ & $0.0505$ & $0.0770$ & $0.1130$ & $0.1820$ \\
& CR & $0.0015$ & $0.0015$ & $0$ & $0.0130$ & $0.0100$ & $0.0195$  \\
& PCVE & $0.0905$ & $0.0770$ & $0.0580$ & $0.1195$ & $0.1260$ & $0.1825$  \\
\multicolumn{5}{l|}{}&\multicolumn{3}{l}{} \\
\multirow{3}{*}{3.31}
& RT & $0.0570$ & $0.0595$ & $0.0555$ & $0.0745$ & $0.1130$ & $0.1670$  \\
& CR & $0.0050$ & $0.0020$ & $0.0020$ & $0.0145$ & $0.0190$ & $0.0270$  \\
& PCVE & $0.1020$ & $0.0845$ & $0.0670$ & $0.1220$ & $0.1340$ & $0.1760$  \\
\multicolumn{5}{l|}{}&\multicolumn{3}{l}{} \\
\multirow{3}{*}{9.80}
& RT & $0.0455$ & $0.0500$ & $0.0475$ & $0.0715$ & $0.1105$ & $0.1410$  \\
& CR & $0.0075$ & $0.0060$ & $0.0030$ & $0.0280$ & $0.0230$ & $0.0305$  \\
& PCVE & $0.1075$ & $0.0900$ & $0.0635$ & $0.1335$ & $0.1380$ & $0.1605$  \\
\hline
\multicolumn{8}{c}{} \\
\multicolumn{8}{c}{\bf Matching on $X_g$ and $N_g$} \\
\hline
\multirow{3}{*}{1.11}
& RT & $0.0490$ & $0.0535$ & $0.0585$ & $0.1165$ & $0.3050$ & $0.6760$  \\
& CR & $0$ & $0$ & $0$ & $0$ & $0$ & $0$  \\
& PCVE & $0.0900$ & $0.0740$ & $0.0640$ & $0.1540$ & $0.2395$ & $0.5015$  \\
\multicolumn{5}{l|}{}&\multicolumn{3}{l}{} \\
\multirow{3}{*}{1.42}
& RT & $0.0440$ & $0.0475$ & $0.0480$ & $0.1290$ & $0.3595$ & $0.7820$  \\
& CR & $0$ & $0$ & $0$ & $0$ & $0$ & $0$  \\
& PCVE & $0.0785$ & $0.0595$ & $0.0575$ & $0.1635$ & $0.2810$ & $0.5705$ \\
\multicolumn{5}{l|}{}&\multicolumn{3}{l}{} \\
\multirow{3}{*}{1.99}
& RT & $0.0510$ & $0.0400$ & $0.0480$ & $0.1255$ & $0.3380$ & $0.7795$ \\
& CR & $0$ & $0$ & $0$ & $0$ & $0$ & $0$  \\
& PCVE & $0.0890$ & $0.0560$ & $0.0605$ & $0.1580$ & $0.2630$ & $0.5785$  \\
\multicolumn{5}{l|}{}&\multicolumn{3}{l}{} \\
\multirow{3}{*}{3.31}
& RT & $0.0440$ & $0.0500$ & $0.0555$ & $0.1185$ & $0.3370$ & $0.7075$  \\
& CR & $0$ & $0$ & $0$ & $0$ & $0$ & $0$  \\
& PCVE & $0.0850$ & $0.0710$ & $0.0675$ & $0.1590$ & $0.2825$ & $0.5220$  \\
\multicolumn{5}{l|}{}&\multicolumn{3}{l}{} \\
\multirow{3}{*}{9.80}
& RT & $0.0525$ & $0.0550$ & $0.0500$ & $0.1180$ & $0.2780$ & $0.5965$  \\
& CR & $0$ & $0$ & $0$ & $0.0005$ & $0$ & $0$  \\
& PCVE & $0.1140$ & $0.0775$ & $0.0620$ & $0.1750$ & $0.2540$ & $0.4625$ \\
\hline\hline
\end{tabular}
\begin{tablenotes}
\item [*] Number of clusters $=2G$ with $G=12, 26, 50$. Number of replications for each $G$ is $2000$. $N_{max}=500$.
\end{tablenotes}
\end{threeparttable}
\end{table}

\begin{table}[htbp]\centering
\begin{threeparttable}[b]
\def\sym#1{\ifmmode^{#1}\else\(^{#1}\)\fi}
\caption{Model 2 - Randomization Test (RT) vs. CR/PCVE\tnote{*}}\label{tab:RT_m4}
\begin{tabular}{ccccc|ccc}
\hline\hline
    &&\multicolumn{3}{c|}{Size under $H_0$}&\multicolumn{3}{c}{Power under $H_1: \Delta_0 + 1/4$} \\
    \multicolumn{1}{c}{$N_{max}/N_{min}$}&\multicolumn{1}{c}{}
    &\multicolumn{1}{c}{$G=12$}&\multicolumn{1}{c}{$G=26$}&\multicolumn{1}{c|}{$G=50$}
    &\multicolumn{1}{c}{$G=12$}&\multicolumn{1}{c}{$G=26$}&\multicolumn{1}{c}{$G=50$}\\
\hline
\multicolumn{8}{c}{} \\
\multicolumn{8}{c}{\bf Matching on $X_g$} \\
\hline
\multirow{3}{*}{1.11}
& RT & $0.0345$ & $0.0425$ & $0.0480$ & $0.0305$ & $0.0790$ & $0.1650$  \\
& CR & $0.0430$ & $0.0365$ & $0.0245$ & $0.0540$ & $0.0645$ & $0.1120$  \\
& PCVE & $0.0440$ & $0.0355$ & $0.0250$ & $0.0550$ & $0.0655$ & $0.1115$  \\
\multicolumn{5}{l|}{}&\multicolumn{3}{l}{} \\
\multirow{3}{*}{1.42}
& RT & $0.0370$ & $0.0365$ & $0.0445$ & $0.0370$ & $0.0675$ & $0.1685$ \\
& CR & $0.0475$ & $0.0295$ & $0.0295$ & $0.0575$ & $0.0560$ & $0.1125$  \\
& PCVE & $0.0465$ & $0.0290$ & $0.0295$ & $0.0560$ & $0.0540$ & $0.1145$  \\
\multicolumn{5}{l|}{}&\multicolumn{3}{l}{} \\
\multirow{3}{*}{1.99}
& RT & $0.0465$ & $0.0445$ & $0.0490$ & $0.0385$ & $0.0785$ & $0.1485$  \\
& CR & $0.0585$ & $0.0405$ & $0.0320$ & $0.0620$ & $0.0675$ & $0.1005$  \\
& PCVE & $0.0585$ & $0.0395$ & $0.0325$ & $0.0615$ & $0.0675$ & $0.1005$  \\
\multicolumn{5}{l|}{}&\multicolumn{3}{l}{} \\
\multirow{3}{*}{3.31}
& RT & $0.0565$ & $0.0495$ & $0.0520$ & $0.0390$ & $0.0660$ & $0.1360$  \\
& CR & $0.0675$ & $0.0385$ & $0.0300$ & $0.0610$ & $0.0620$ & $0.1010$  \\
& PCVE & $0.0685$ & $0.0385$ & $0.0315$ & $0.0595$ & $0.0625$ & $0.1025$  \\
\multicolumn{5}{l|}{}&\multicolumn{3}{l}{} \\
\multirow{3}{*}{9.80}
& RT & $0.0700$ & $0.0660$ & $0.0600$ & $0.0405$ & $0.0550$ & $0.1140$  \\
& CR & $0.0845$ & $0.0540$ & $0.0360$ & $0.0585$ & $0.0600$ & $0.0895$  \\
& PCVE & $0.0825$ & $0.0550$ & $0.0365$ & $0.0595$ & $0.0580$ & $0.0895$  \\
\hline
\multicolumn{8}{c}{} \\
\multicolumn{8}{c}{\bf Matching on $X_g$ and $N_g$} \\
\hline
\multirow{3}{*}{1.11}
& RT & $0.0250$ & $0.0310$ & $0.0370$ & $0.0195$ & $0.0735$ & $0.1800$  \\
& CR & $0.0330$ & $0.0155$ & $0.0125$ & $0.0240$ & $0.0365$ & $0.0765$  \\
& PCVE & $0.0320$ & $0.0150$ & $0.0135$ & $0.0235$ & $0.0360$ & $0.0790$  \\
\multicolumn{5}{l|}{}&\multicolumn{3}{l}{} \\
\multirow{3}{*}{1.42}
& RT & $0.0295$ & $0.0290$ & $0.0345$ & $0.0205$ & $0.0730$ & $0.1740$  \\
& CR & $0.0335$ & $0.0150$ & $0.0150$ & $0.0245$ & $0.0385$ & $0.0640$  \\
& PCVE & $0.0340$ & $0.0150$ & $0.0155$ & $0.0250$ & $0.0365$ & $0.0675$  \\
\multicolumn{5}{l|}{}&\multicolumn{3}{l}{} \\
\multirow{3}{*}{1.99}
& RT & $0.0345$ & $0.0325$ & $0.0415$ & $0.0200$ & $0.0665$ & $0.1655$  \\
& CR & $0.0350$ & $0.0175$ & $0.0115$ & $0.0225$ & $0.0310$ & $0.0600$  \\
& PCVE & $0.0330$ & $0.0185$ & $0.0120$ & $0.0230$ & $0.0320$ & $0.0610$  \\
\multicolumn{5}{l|}{}&\multicolumn{3}{l}{} \\
\multirow{3}{*}{3.31}
& RT & $0.0390$ & $0.0390$ & $0.0340$ & $0.0150$ & $0.0590$ & $0.1415$  \\
& CR & $0.0420$ & $0.0205$ & $0.0110$ & $0.0220$ & $0.0295$ & $0.0610$  \\
& PCVE & $0.0420$ & $0.0200$ & $0.0110$ & $0.0210$ & $0.0310$ & $0.0595$  \\
\multicolumn{5}{l|}{}&\multicolumn{3}{l}{} \\
\multirow{3}{*}{9.80}
& RT & $0.0555$ & $0.0445$ & $0.0415$ & $0.0260$ & $0.0405$ & $0.1180$  \\
& CR & $0.0590$ & $0.0235$ & $0.0155$ & $0.0295$ & $0.0270$ & $0.0505$  \\
& PCVE & $0.0570$ & $0.0245$ & $0.0170$ & $0.0295$ & $0.0265$ & $0.0510$  \\
\hline\hline
\end{tabular}
\begin{tablenotes}
\item [*] Number of clusters $=2G$ with $G=12, 26, 50$. Number of replications for each $G$ is $2000$. $N_{max}=500$.
\end{tablenotes}
\end{threeparttable}
\end{table} 

\begin{table}[htbp]\centering
\begin{threeparttable}[b]
\def\sym#1{\ifmmode^{#1}\else\(^{#1}\)\fi}
\caption{Covariate Adjustment - Matching on $X_g$\tnote{*}}
\begin{tabular}{c*{9}{c}}
\hline\hline
    \multicolumn{1}{c}{$N_{max}/N_{min}$}&\multicolumn{1}{c}{$\psi_g$}
    &\multicolumn{1}{c}{$G=12$}&\multicolumn{1}{c}{$G=26$}&\multicolumn{1}{c}{$G=50$}
    &\multicolumn{1}{c}{$G=100$}&\multicolumn{1}{c}{$G=150$}
    &\multicolumn{1}{c}{$G=200$}&\multicolumn{1}{c}{$G=250$}\\
\hline
\multicolumn{9}{c}{} \\
\multicolumn{9}{c}{\bf Coverage} \\
\hline
1.11
& -                 & 0.9015 & 0.9235 & 0.9435 & 0.9395 & 0.9365 & 0.9445 & 0.9485 \\
& $C_g$ & 0.8305 & 0.9025 & 0.9240 & 0.9410 & 0.9435 & 0.9455 & 0.9430 \\
\multicolumn{9}{l}{} \\
1.42
& -                 & 0.9070 & 0.9315 & 0.9365 & 0.9405 & 0.9455 & 0.9490 & 0.9525 \\
& $C_g$ & 0.8415 & 0.9060 & 0.9280 & 0.9430 & 0.9450 & 0.9455 & 0.9515 \\
\multicolumn{9}{l}{} \\
1.99
& -                 & 0.9050 & 0.9310 & 0.9450 & 0.9450 & 0.9480 & 0.9530 & 0.9465 \\
& $C_g$ & 0.8380 & 0.9025 & 0.9310 & 0.9395 & 0.9450 & 0.9480 & 0.9495 \\
\multicolumn{9}{l}{} \\
3.31
& -                 & 0.9100 & 0.9340 & 0.9410 & 0.9535 & 0.9520 & 0.9490 & 0.9485 \\
& $C_g$ & 0.8475 & 0.9065 & 0.9335 & 0.9400 & 0.9450 & 0.9450 & 0.9465 \\
\multicolumn{9}{l}{} \\
9.80
& -                 & 0.8975 & 0.9305 & 0.9410 & 0.9435 & 0.9420 & 0.9430 & 0.9545 \\
& $C_g$ & 0.8290 & 0.8885 & 0.9365 & 0.9405 & 0.9415 & 0.9430 & 0.9475 \\
\hline
\multicolumn{9}{c}{} \\
\multicolumn{9}{c}{{\bf Average Length}} \\
\hline
1.11
& -                 & 1.86744 & 1.31289 & 0.95830 & 0.68388 & 0.55761 & 0.48368 & 0.43289 \\
& $C_g$ & 1.24948 & 0.91803 & 0.68139 & 0.49245 & 0.40117 & 0.34947 & 0.31297 \\
\multicolumn{9}{l}{} \\
1.42
& -                 & 1.86822 & 1.30105 & 0.95121 & 0.67677 & 0.55462 & 0.48111 & 0.43046 \\
& $C_g$ & 1.27135 & 0.91549 & 0.67994 & 0.48916 & 0.40149 & 0.34852 & 0.31232 \\
\multicolumn{9}{l}{} \\
1.99
& -                 & 1.85639 & 1.29289 & 0.94626 & 0.67421 & 0.55160 & 0.47822 & 0.42849 \\
& $C_g$ & 1.26315 & 0.91509 & 0.68035 & 0.48902 & 0.40081 & 0.34844 & 0.31184 \\
\multicolumn{9}{l}{} \\
3.31
& -                 & 1.83716 & 1.29155 & 0.94173 & 0.67099 & 0.54871 & 0.47588 & 0.42645 \\
& $C_g$ & 1.24978 & 0.92179 & 0.68201 & 0.48944 & 0.40179 & 0.34984 & 0.31320 \\
\multicolumn{9}{l}{} \\
9.80
& -                 & 1.83555 & 1.28894 & 0.93697 & 0.66756 & 0.54602 & 0.47402 & 0.42411 \\
& $C_g$ & 1.27637 & 0.92561 & 0.68705 & 0.49519 & 0.40581 & 0.35303 & 0.31622 \\
\hline\hline
\end{tabular}
\label{tab:modelH}
\begin{tablenotes}
\item [*] Number of clusters $=2G$ with $G=12, 26, 50, 100, 150, 200, 250$. Number of replications for each $G$ is $2000$. $N_{max}=500$.
\end{tablenotes}
\end{threeparttable}
\end{table}

\begin{table}[htbp]\centering
\begin{threeparttable}[b]
\def\sym#1{\ifmmode^{#1}\else\(^{#1}\)\fi}
\caption{Covariate Adjustment - Matching on $X_g$ and $N_g$\tnote{*}}
\begin{tabular}{c*{9}{c}}
\hline\hline
    \multicolumn{1}{c}{$N_{max}/N_{min}$}&\multicolumn{1}{c}{$\psi_g$}
    &\multicolumn{1}{c}{$G=12$}&\multicolumn{1}{c}{$G=26$}&\multicolumn{1}{c}{$G=50$}
    &\multicolumn{1}{c}{$G=100$}&\multicolumn{1}{c}{$G=150$}
    &\multicolumn{1}{c}{$G=200$}&\multicolumn{1}{c}{$G=250$}\\
\hline
\multicolumn{9}{c}{} \\
\multicolumn{9}{c}{\bf Coverage} \\
\hline
1.11
& -                 & 0.9120 & 0.9275 & 0.9475 & 0.9395 & 0.9425 & 0.9510 & 0.9425 \\
& $C_g$ & 0.8385 & 0.8920 & 0.9335 & 0.9400 & 0.9465 & 0.9475 & 0.9495 \\
\multicolumn{9}{l}{} \\
1.42
& -                 & 0.9135 & 0.9245 & 0.9415 & 0.9445 & 0.9495 & 0.9425 & 0.9425 \\
& $C_g$ & 0.8485 & 0.9000 & 0.9285 & 0.9435 & 0.9470 & 0.9490 & 0.9475 \\
\multicolumn{9}{l}{} \\
1.99
& -                 & 0.9085 & 0.9250 & 0.9420 & 0.9470 & 0.9455 & 0.9545 & 0.9520 \\
& $C_g$ & 0.8425 & 0.9035 & 0.9345 & 0.9410 & 0.9505 & 0.9460 & 0.9470 \\
\multicolumn{9}{l}{} \\
3.31
& -                 & 0.9090 & 0.9265 & 0.9340 & 0.9515 & 0.9465 & 0.9465 & 0.9535 \\
& $C_g$ & 0.8410 & 0.9075 & 0.9365 & 0.9390 & 0.9435 & 0.9490 & 0.9500 \\
\multicolumn{9}{l}{} \\
9.80
& -                 & 0.9070 & 0.9245 & 0.9330 & 0.9375 & 0.9510 & 0.9455 & 0.9440 \\
& $C_g$ & 0.8440 & 0.9015 & 0.9275 & 0.9415 & 0.9510 & 0.9400 & 0.9475 \\
\hline
\multicolumn{9}{c}{} \\
\multicolumn{9}{c}{{\bf Average Length}} \\
\hline
1.11
& -                 & 1.77556 & 1.21499 & 0.88201 & 0.62584 & 0.51123 & 0.44346 & 0.39699 \\
& $C_g$ & 1.31267 & 0.93535 & 0.68999 & 0.49308 & 0.40413 & 0.35129 & 0.31419 \\
\multicolumn{9}{l}{} \\
1.42
& -                 & 1.74117 & 1.20501 & 0.87067 & 0.62002 & 0.50712 & 0.43888 & 0.39274 \\
& $C_g$ & 1.31317 & 0.92993 & 0.68771 & 0.49157 & 0.40238 & 0.34915 & 0.31221 \\
\multicolumn{9}{l}{} \\
1.99
& -                 & 1.72916 & 1.19588 & 0.86887 & 0.61669 & 0.50509 & 0.43677 & 0.39112 \\
& $C_g$ & 1.30301 & 0.93106 & 0.68850 & 0.49048 & 0.40134 & 0.34801 & 0.31173 \\
\multicolumn{9}{l}{} \\
3.31
& -                 & 1.71004 & 1.19463 & 0.86708 & 0.61577 & 0.50301 & 0.43573 & 0.39127 \\
& $C_g$ & 1.30080 & 0.93384 & 0.68661 & 0.48951 & 0.40075 & 0.34720 & 0.31157 \\
\multicolumn{9}{l}{} \\
9.80
& -                 & 1.72505 & 1.19952 & 0.86484 & 0.61768 & 0.50429 & 0.43672 & 0.39197 \\
& $C_g$ & 1.31500 & 0.93975 & 0.68887 & 0.49150 & 0.40285 & 0.34975 & 0.31339 \\
\hline\hline
\end{tabular}
\label{tab:modelHN}
\begin{tablenotes}
\item [*] Number of clusters $=2G$ with $G=12, 26, 50, 100, 150, 200, 250$. Number of replications for each $G$ is $2000$. $N_{max}=500$.
\end{tablenotes}
\end{threeparttable}
\end{table}

\clearpage
\bibliography{mp-clusters}

\clearpage
\newpage

\appendix

\begin{center}
    \Large Supplemental Appendix: For Online Publication
\end{center}

\section{Sufficient Conditions for Assumptions \ref{ass:pair_formX} and \ref{ass:pair_form}}\label{sec:match_sufficient}
We only lay out the argument for Assumption \ref{ass:pair_formX} and an identical argument applies to Assumption \ref{ass:pair_form}. Let $k_x = \mathrm{dim}(X_g)$. Note
\begin{equation}\label{eq:nbp}
    \frac{1}{G}\sum_{1 \le j \le G}\|X_{\pi(2j)} - X_{\pi(2j-1)}\|^r \leq \Big ( 1 \vee \max_{1 \leq g \leq 2G} \|X_g\|^r \Big ) \frac{1}{G}\sum_{1 \le j \le G} \bigg \| \frac{X_{\pi(2j)} - X_{\pi(2j-1)}}{1 \vee \max_{1 \leq g \leq 2G} \|X_g\|} \bigg \|^r~.
\end{equation}
Consider a non-bipartite matching algorithm that minimizes the left-hand side of \eqref{eq:nbp} for $ r = 2$ for Assumption \ref{ass:pair_formX} (or $r = 4$ for Assumption \ref{ass:pair_form}). Because
\[ X_g / \max_{1 \leq g \leq 2G} \|X_g\| \in [0, 1]^{k_x}~, \]
to study
\begin{equation} \label{eq:normalized-dist}
    \frac{1}{G}\sum_{1 \le j \le G} \bigg \| \frac{X_{\pi(2j)} - X_{\pi(2j-1)}}{1 \vee \max_{1 \leq g \leq 2G} \|X_g\|} \bigg \|^r~,
\end{equation}
we can assume without loss of generality that $X_g \in [0, 1]^{k_x}$ for $1 \leq g \leq 2G$. Consider as an auxiliary proof device the block-path algorithm in the proof of Theorem 4.2 in \cite{bai2022inference} with blocks of side lengths $1/m$. Using the inequality $c^r \leq c$ if $r \geq 1$ and $c \in [0, 1]$, note if $x_1, x_2 \in [0, 1]^{k_x}$, then
\[ \|x_1 - x_2\|^r = k_x^{2/r} (\|x_1 - x_2\| / \sqrt{k_x})^r \leq k_x^{2/r} \|x_1 - x_2\| / \sqrt{k_x} = k_x^{2/r - 1/2} \|x_1 - x_2\|~. \]
Therefore, following the proof of Theorem 4.2 in \cite{bai2022inference} or Lemma A.1 in \cite{cytrynbaum2021designing},
\[ \frac{1}{G} \sum_{1 \leq j \leq G} \|X_{\pi(2j)} - X_{\pi(2j - 1)}\|^r \leq \bigg ( \frac{\sqrt k_x}{m} \bigg )^r + \frac{2}{G} k_x^{2/r} m^{k_x - 1}~. \]
Taking $m \asymp G^{1/ (r + k_x - 1)}$, \eqref{eq:normalized-dist} is of order $G^{-r/(r + k_x - 1)}$. On the other hand, if $E[\|X_g\|^d] < \infty$, Lemma S.1.1 in \cite{bai2022inference} implies $\max_{1 \leq g \leq 2G} \|X_g\|^r = o_P(G^{r/d})$. Therefore, as long as $d \geq r + k_x - 1$, the left-hand side of \eqref{eq:nbp} converges to zero in probability.

Note further that, when verifying Assumption \ref{ass:pair_form}, if $\|W_g\|$ is bounded, then
\[ \frac{1}{G}\sum_{1 \le j \le G}\|W_{\pi(2j)} - W_{\pi(2j-1)}\|^4 \lesssim \frac{1}{G}\sum_{1 \le j \le G}\|W_{\pi(2j)} - W_{\pi(2j-1)}\|^2~, \]
and therefore any algorithm that minimizes the right-hand of the above display will satisfy Assumption \ref{ass:pair_form}.

\section{Proofs of Main Results}
Please note that in what follows we will use the notation $a \lesssim b$ to denote $a \le cb$ for some constant $c$.

\subsection{Proof of Theorem \ref{thm:normal_X}}
\begin{proof}
We have that
\[ \hat \Delta_G = \frac{\frac{1}{G}\sum_{1 \leq g \leq 2G} \bar{Y}_g(1) N_g D_g}{\frac{1}{G}\sum_{1 \leq g \leq 2G} N_g D_g} - \frac{\frac{1}{G}\sum_{1 \leq g \leq 2G} \bar{Y}_g(0) N_g (1 - D_g)}{\frac{1}{G}\sum_{1 \leq g \leq 2G} N_g (1 - D_g)}~. \]
In particular, for $h(x, y, z, w) = \frac{x}{y} - \frac{z}{w}$, observe that
\begin{align*}
\hat \Delta_G &= h \left( \frac{1}{G}\sum_{1 \leq g \leq 2G}\bar{Y}_g(1) N_g D_g, \frac{1}{G}\sum_{1 \leq g \leq 2G} N_g D_g, \frac{1}{G}\sum_{1 \leq g \leq 2G} \bar{Y}_g(0) N_g (1 - D_g), \frac{1}{G}\sum_{1 \leq g \leq 2G} N_g (1 - D_g) \right)~,
\end{align*}
and by Assumption \ref{ass:indep_pairsX},
\begin{multline*}
\Delta = h \Bigg( \frac{1}{G}\sum_{1 \leq g \leq 2G}E[\bar{Y}_g(1) N_g] D_g, \frac{1}{G}\sum_{1 \leq g \leq 2G} E[N_g] D_g, \\
\frac{1}{G}\sum_{1 \leq g \leq 2G} E[\bar{Y}_g(0) N_g] (1 - D_g), \frac{1}{G}\sum_{1 \leq g \leq 2G} E[N_g] (1 - D_g) \Bigg)~.
\end{multline*}
The Jacobian of $h(\cdot)$ is
\[ D_h(x, y, z, w) = \begin{pmatrix} \frac{1}{y} & - \frac{x}{y^2} & - \frac{1}{w} & \frac{z}{w^2} \end{pmatrix}~. \]
%By Assumption \ref{ass:indep_pairsX}, 
%\[ \sqrt G \Big ( \frac{1}{G}\sum_{1 \leq g \leq 2G} \bar{Y}_g N_g D_g - E[\bar{Y}_g(1) N_g] \Big ) = %\frac{1}{\sqrt G} \sum_{1 \leq g \leq 2G} (\bar{Y}_g(1) N_g D_g - E[\bar{Y}_g(1) N_g] D_g) \]
%and similarly for the other three terms. 
By Lemma \ref{lem:L_X} and the Delta method,
\[\sqrt{G}(\hat{\Delta}_G - \Delta) \xrightarrow{d} N(0, D_{h0}\mathbb{V}D_{h0}')~,\]
where
\[ D_{h0} = \begin{pmatrix} \frac{1}{E[N_g]} & -\frac{E[\bar{Y}_g(1)N_g]}{E[N_g]^2} & -\frac{1}{E[N_g]} & \frac{E[\bar{Y}_g(0)N_g]}{E[N_g]^2} \end{pmatrix} \]
and $\mathbb{V}$ is defined in Lemma \ref{lem:L_X}. It then follows from  Lemma \ref{lem:L_X_algebra} that
\[D_{h0}\mathbb{V}D_{h0}' = \omega^2~,\]
as desired.
\end{proof}

\subsection{Proof of Theorem \ref{thm:normal_N}}
\begin{proof}
This proof follows from an identical argument to Theorem \ref{thm:normal_X}, but this time invoking Lemmas \ref{lem:L_N} and \ref{lem:L_N_algebra}.
%We have that
%\[ \hat \Delta_G = \frac{\frac{1}{G}\sum_{1 \leq g \leq 2G} \bar{Y}_g(1) N_g D_g}{\frac{1}{G}\sum_{1 \leq g \leq 2G} N_g D_g} - \frac{\frac{1}{G}\sum_{1 \leq g \leq 2G} \bar{Y}_g(0) N_g (1 - D_g)}{\frac{1}{G}\sum_{1 \leq g \leq 2G} N_g (1 - D_g)}~. \]
%In particular, for $h(x, y, z, w) = \frac{x}{y} - \frac{z}{w}$, observe that
%\begin{align*}
%\hat \Delta_G &= h \left( \frac{1}{G}\sum_{1 \leq g \leq 2G}\bar{Y}_g(1) N_g D_g, \frac{1}{G}\sum_{1 \leq g \leq 2G} N_g D_g, \frac{1}{G}\sum_{1 \leq g \leq 2G} \bar{Y}_g(0) N_g (1 - D_g), \frac{1}{G}\sum_{1 \leq g \leq 2G} N_g (1 - D_g) \right)
%\end{align*}
%and the Jacobian is
%\[ D_h(x, y, z, w) = \Big ( \frac{1}{y}, - \frac{x}{y^2}, - \frac{1}{w}, \frac{z}{w^2} \Big )~. \]
%By Assumption \ref{ass:indep_pairs}, 
%\[ \sqrt G \Big ( \frac{1}{G}\sum_{1 \leq g \leq 2G} \bar{Y}_g N_g D_g - E[\bar{Y}_g(1) N_g] \Big ) = \frac{1}{\sqrt G} \sum_{1 \leq g \leq 2G} (\bar{Y}_g(1) N_g D_g - E[\bar{Y}_g(1) N_g] D_g) \]
%and similarly for the other three terms. The desired conclusion then follows from Lemmas \ref{lem:L_N} and \ref{lem:L_N_algebra} together with an application of the Delta method. 
\end{proof}

\subsection{Proof of Theorem \ref{thm:variance-estimator}}
The desired conclusion follows immediately from Lemmas \ref{lemma:mu_G}-\ref{lemma:lambda_G} and the continuous mapping theorem. \qed

\subsection{Proof of Theorem \ref{thm:CR_limit}}
By the first result in Theorem 3.6 in \cite{bugni2024inference}, 
\begin{equation}\label{eq:omega_CR}
\hat{\omega}^2_{\rm CR,G} = \frac{1}{2}\left(\hat{\omega}^2_{\rm CR,G}(1) +  \hat{\omega}^2_{\rm CR,G}(0)\right)~,
\end{equation}
(where we note that the factor of $1/2$ appears since we are normalizing by the number of \emph{pairs}), and
\begin{equation*}
\hat{\omega}^2_{\rm CR,G}(d) := \frac{1}{\left(\frac{1}{2G}\sum_{1\le g \le 2G}{N_g}I\{D_g = d\}\right)^2}\frac{1}{2G}\sum_{1\le g \le 2G}\left[\left(\frac{N_g}{|\mathcal{M}_g|}\right)^2I\{D_g = d\}\left(\sum_{i\in \mathcal{M}_g} \hat{\epsilon}_{i,g}(d)\right)^2\right]~,
\end{equation*}
with
\[\hat{\epsilon}_{i,g}(d) := Y_{i,g} - \frac{1}{\sum_{1\le g \le 2G}N_gI\{D_g = d\}}\sum_{1 \le g \le 2G}N_g\bar{Y}_gI\{D_g = d\}~.\]
Fix $d \in \{0, 1\}$, $r \in \{0, 1, 2\}$, $\ell \in \{1, 2\}$ arbitrarily. Then by Lemmas \ref{lem:E_bounded} and \ref{lem:wlln}, 
\[\frac{1}{2G}\sum_{1 \le g \le 2G}N^{\ell}_g\bar{Y}_g^r(d)I\{D_g = d\} \xrightarrow{P} \frac{E[N^l\bar{Y}^r_g(d)]}{2}~.\]
The result then follows from additional algebra and repeated applications of the continuous mapping theorem; an identical derivation appears as the second result in Theorem 3.6 of \cite{bugni2024inference}. \qed

\subsection{Proof of Theorem \ref{thm:PCVE_limit}}
Let ${\bf 1}_{K}$ denote a column of ones of length $K$. Then consider the following cluster-robust variance estimator where clusters are defined at the level of the \emph{pair}:
\begin{equation}\label{eq:PCVE}
\Bigg(\frac{1}{G}\sum_{1\le j \le G}\sum_{g \in \lambda_j}X_g'X_g\Bigg)^{-1} \Bigg(\frac{1}{G}\sum_{1 \le j \le G}\left(\sum_{g \in \lambda_j}X_g' \hat{\epsilon}_g\right)\left(\sum_{g \in \lambda_j}X_g' \hat{\epsilon}_g\right)' \Bigg) \Bigg(\frac{1}{G}\sum_{1 \le g \le G}\sum_{g \in \lambda_j}X_g'X_g\Bigg)^{-1},
\end{equation}
where $\lambda_j: = \{\pi(2j-1),\pi(2j)\}$, and  
\begin{align*}
X_g~&:=~\left(\begin{array}{cc}
{\bf 1}_{|\mathcal{M}_g|} \cdot \sqrt{\frac{N_g}{|\mathcal{M}_g|}}, & ~~~~{\bf 1}_{|\mathcal{M}_g|} \cdot \sqrt{\frac{N_g}{|\mathcal{M}_g|}}D_g
\end{array}%
\right)\\
\hat{\epsilon}_g ~&:=~ \sqrt{\frac{N_g}{|\mathcal{M}_g|}}\left(Y_{i,g} - (\hat{\mu}_G(1) - \hat{\mu}_G(0))D_g - \hat{\mu}_G(0)~:~ i \in \mathcal{M}_g\right)'~.
\end{align*}
Imposing the condition that $N_g = n$ are equal and fixed and $|\mathcal{M}_g| = N_g$, and then following the algebra in, for instance, the proof of Theorem 3.4 in \cite{bai2024inference}, it can be shown that 
\[\hat{\omega}^2_{\rm PCVE, G} = \frac{1}{G} \sum_{1 \leq j \leq G} \left ( \sum_{g \in \lambda_j} \bar{Y}_g I \{D_g = 1\} - \sum_{g \in \lambda_j} \bar{Y}_g I \{D_g = 0\} \right )^2 - (\hat \mu_G(1) - \hat \mu_G(0))^2~.\]
By some additional algebra and repeated applications of Lemmas \ref{lem:wlln}, \ref{lem:cross}, and the continuous mapping theorem we thus obtain that
\begin{multline*}
    \hat{\omega}^2_{\rm PCVE, G} \xrightarrow{P} E[\var[\bar{Y}_g(1)|X_g]] + E[\var[\bar{Y}_g(1)|X_g]] \\
    + E[\left((E[\bar{Y}_g(1)|X_g] - E[\bar{Y}_g(1)]) - (E[\bar{Y}_g(0)|X_g] - E[\bar{Y}_g(0)])\right)^2]~.
\end{multline*}
Simplifying using the law of total variance and the fact that $\tilde{Y}_g(d) = \bar{Y}_g(d) - E[\bar{Y}_g(d)]$ once we impose that $N_g = n$, we then obtain
\[\hat{\omega}^2_{\rm PCVE, G} \xrightarrow{P} E[\tilde{Y}^2_g(1)] + E[\tilde{Y}^2_g(0)] - \frac{1}{2}E[(E[\tilde{Y}_g(1) + \tilde{Y}_g(0)|X_g])^2] + \frac{1}{2}E\left[(E[\tilde{Y}_g(1) - \tilde{Y}_g(0)|X_g])^2\right]~.\]
The conclusion then follows. \qed

\subsection{Proof of Theorem \ref{thm:rand_finite}}
\begin{proof}
Note that the null hypothesis \eqref{eq:dist_XN} combined with Assumption \ref{ass:QG}(e) implies that
\begin{align}\label{eq:bar_null}
\bar{Y}_g(1)|(X_g, N_g) \stackrel{d}{=} \bar{Y}_g(0)|(X_g, N_g)~.
\end{align}
If the assignment mechanism satisfies Assumption \ref{ass:indep_pairs}, the result then follows by applying Theorem 3.4 in \cite{bai2022inference} to the cluster-level outcomes $\{(\bar{Y}_g,D_g,X_g,N_g):1 \le g \le 2G\}$. If instead the assignment mechanism satisfies Assumption \ref{ass:indep_pairsX}, then note that \eqref{eq:bar_null} is in fact equivalent to the statement
\begin{align}\label{eq:bar_null2}
(\bar{Y}_g(1), N_g)|X_g \stackrel{d}{=} (\bar{Y}_g(0),N_g)|X_g~.
\end{align}
The result then follows by applying Theorem 3.4 in \cite{bai2022inference} using \eqref{eq:bar_null2} as the null hypothesis. To establish this equivalence, we first begin with \eqref{eq:bar_null} and verify that for any Borel sets $A$ and $B$,
\[ P \{\bar{Y}_g(1) \in A, N_g \in B | X_g\} = P \{\bar{Y}_g(0) \in A, N_g \in B | X_g\} \text{ a.s. } \]
By the definition of a conditional expectation, note we only need to verify for all Borel sets $C$,
\[ E[P \{\bar{Y}_g(1) \in A,  N_g \in B| X_g\} I \{X_g \in C\}] = P \{\bar{Y}_g(0) \in A, N_g \in B, X_g \in C \}~. \]
We have
\begin{align*}
& E[P \{\bar{Y}_g(1) \in A,  N_g \in B| X_g\} I \{X_g \in C\}] \\
& = P \{\bar{Y}_g(1) \in A, N_g \in B, X_g \in C\} \\
& = E[P \{\bar{Y}_g(1) \in A |X_g, N_g\}I\{N_g \in B\}I \{X_g \in C\}] \\
& = E[P \{\bar{Y}_g(0) \in A |X_g, N_g\}I\{N_g \in B\}I \{X_g \in C\}] \\
& = P \{\bar{Y}_g(0) \in A, N_g \in B, X_g \in C \}~,
\end{align*}
where the first and second equalities follow from the definition of conditional expectations, the the third follows from \eqref{eq:bar_null}, and the last follows again from the definition of a conditional expectation. The opposite implication follows from a similar argument and is thus omitted.
\end{proof}

\subsection{Proof of Theorem \ref{thm:rand_large}}
Note that
\begin{align*}
\sqrt{G}\hat{\Delta}_G =&\sqrt{G}\left(\frac{1}{N(1)}\sum_{1\le g \le 2G}D_gN_g\bar{Y}_g - \frac{1}{N(0)}\sum_{1 \le g \le 2G}(1 - D_g)N_g\bar{Y}_g\right) \\
&= \frac{1}{N(1)}\sqrt{G}\sum_{1 \le g \le 2G}\left(D_gN_g\bar{Y}_g - (1 - D_g)N_g\bar{Y}_g\right) + \left(\frac{1}{N(1)} - \frac{1}{N(0)}\right)\sqrt{G}\sum_{1 \le g \le 2G}(1 - D_g)N_g\bar{Y}_g \\
&= \frac{1}{N(1)/G}\frac{1}{\sqrt{G}}\sum_{1 \le j \le G}\left(N_{\pi(2j)}\bar{Y}_{\pi(2j)} - N_{\pi(2j-1)}\bar{Y}_{\pi(2j-1)}\right)(D_{\pi(2j)} - D_{\pi(2j-1)}) \\
& \hspace{3cm} + \frac{\frac{1}{\sqrt{G}}(N(0) - N(1))}{\frac{N(1)}{G}\frac{N(0)}{G}}\frac{1}{G}\sum_{1 \le g \le 2G}(1 - D_g)N_g\bar{Y}_g \\
&= \frac{1}{N(1)/G}\frac{1}{\sqrt{G}}\sum_{1 \le j \le G}\left(N_{\pi(2j)}\bar{Y}_{\pi(2j)} - N_{\pi(2j-1)}\bar{Y}_{\pi(2j-1)}\right)(D_{\pi(2j)} - D_{\pi(2j-1)}) \\ & \hspace{2cm}- \frac{\frac{1}{\sqrt{G}}\sum_{1 \le j \le G}(N_{\pi(2j)} - N_{\pi(2j-1)})(D_{\pi(2j)} - D_{\pi(2j-1)})}{\frac{N(1)}{G}\frac{N(0)}{G}}\frac{1}{G}\sum_{1 \le g \le 2G}(1 - D_g)N_g\bar{Y}_g~.
\end{align*}
Hence the randomization distribution of $\sqrt{G}\hat{\Delta}_G$ is given by
\begin{equation}\label{eq:rand_num}
\tilde{R}_G(t) :=P\Bigg\{\sqrt{G}\check\Delta(\epsilon_1, \ldots, \epsilon_G) \le t\Bigg|Z^{(G)}\Bigg\}~,
\end{equation}
where 
\begin{multline*}\sqrt{G}\check\Delta(\epsilon_1, \ldots, \epsilon_G) = \frac{1}{\tilde{N}(1)/G}\frac{1}{\sqrt{G}}\sum_{1 \le j \le G}\epsilon_j\left(N_{\pi(2j)}\bar{Y}_{\pi(2j)} - N_{\pi(2j-1)}\bar{Y}_{\pi(2j-1)}\right)(D_{\pi(2j)} - D_{\pi(2j-1)}) \\ 
- \frac{\frac{1}{\sqrt{G}}\sum_{1 \le j \le G}\epsilon_j(N_{\pi(2j)} - N_{\pi(2j-1)})(D_{\pi(2j)} - D_{\pi(2j-1)})}{\frac{\tilde{N}(1)}{G}\frac{\tilde{N}(0)}{G}}\frac{1}{G}\sum_{1 \le g \le 2G}(1 - \tilde{D}_g)N_g\bar{Y}_g~,
\end{multline*}
$\epsilon_j$, $j = 1,\ldots,G$ are i.i.d.\ Rademacher random variables generated independently of $Z^{(G)}$, $\{\tilde{D}_g: 1 \le g \le 2G\}$ denotes the assignment of cluster $g$ after applying the transformation implied by $\{\epsilon_j: 1 \le j \le G\}$, and 
\[\tilde{N}(d) = \sum_{1 \le g \le 2G}N_gI\{\tilde{D}_g = d\}~.\]
By construction, $\hat{v}_G^2$ evaluated at the transformation of the data implied by $\{\epsilon_j: 1 \le j \le G\}$ is given by 
\begin{equation}\label{eq:check_v}
    \check{v}_G^2(\epsilon_1,\dots, \epsilon_G) = \hat{\tau}_G^2-\frac{1}{2}\check{\lambda}_G^2\left(\epsilon_1, \ldots, \epsilon_G\right)
\end{equation}
where $\hat{\tau}_G^2$ is defined in (\ref{eqn:define-variance-estimator}), and
\begin{multline*}
    \check{\lambda}_G^2\left(\epsilon_1, \ldots, \epsilon_G\right) = \frac{2}{G} \sum_{1 \leq j \leq\left\lfloor G/2\right\rfloor}\epsilon_{2j-1}\epsilon_{2j}\left(\hat Y_{\pi(4 j-3)}-\hat Y_{\pi(4 j-2)}\right)\left(\hat Y_{\pi(4 j-1)}- \hat Y_{\pi(4 j)}\right) \\
    \times \left(D_{\pi(4 j-3)}-D_{\pi(4 j-2)}\right)\left(D_{\pi(4 j-1)}-D_{\pi(4 j)}\right)~.
\end{multline*}
The desired conclusion then follows from Lemmas \ref{lem:rand_numerator} and \ref{lem:rand_denom}, along with Theorem 5.2 in \cite{chung2013exact}. \qed

\subsection{Proof of Theorem \ref{thm:adj}}
\noindent \underline{Step 1: Limit of $\hat\beta_G$}

We first establish that $\hat \beta_G \xrightarrow{P} \beta^\ast$ for $\beta^\ast$ in \eqref{eq:betastar}. Recall that $\hat \beta_G$ is the OLS estimator of the slope coefficient in the linear regression of $(\hat Y_{\pi(2g - 1)}  \bar N_G - \hat Y_{\pi(2g)} \bar N_G )(D_{\pi(2g - 1)} - D_{\pi(2g)})$ on a constant and $(\psi_{\pi(2g - 1)} - \psi_{\pi(2g)})(D_{\pi(2g - 1)} - D_{\pi(2g)})$, where $\bar N_G = \frac{1}{2 G} \sum_{1\leq g \leq 2 G} N_g$. Equivalently, we have $\hat \beta_G$ as the OLS estimator of the slope coefficient in the linear regression of $\hat\mu_{1,j} - \hat\mu_{0,j}$ on a constant and $\hat\psi_{1,j} - \hat\psi_{0,j}$, where
\begin{align*}
    \hat\mu_{1,j} &= \left(\bar Y_{\pi(2j-1)}(1) -  \frac{\frac{1}{G}\sum_{1\leq g \leq 2G} \bar{Y}_g D_g  N_g}{\frac{1}{G}\sum_{1\leq g \leq 2G} D_g  N_g}\right) N_{\pi(2j-1)} D_{\pi(2j-1)} \\
    &\hspace{3em} + \left(\bar Y_{\pi(2j)}(1) -  \frac{\frac{1}{G}\sum_{1\leq g \leq 2G} \bar{Y}_g D_g  N_g}{\frac{1}{G}\sum_{1\leq g \leq 2G} D_g  N_g}\right) N_{\pi(2j)} D_{\pi(2j)} \\
    \hat\mu_{0,j} &= \left(\bar Y_{\pi(2j-1)}(0) - \frac{\frac{1}{G}\sum_{1\leq g \leq 2G} \bar{Y}_g (1-D_g)  N_g}{\frac{1}{G}\sum_{1\leq g \leq 2G} (1-D_g)  N_g} \right) N_{\pi(2j-1)} (1-D_{\pi(2j-1)}) \\
    &\hspace{3em}+ \left(\bar Y_{\pi(2j)}(0) - \frac{\frac{1}{G}\sum_{1\leq g \leq 2G} \bar{Y}_g (1-D_g)  N_g}{\frac{1}{G}\sum_{1\leq g \leq 2G} (1-D_g)  N_g}\right)  N_{\pi(2j)} (1-D_{\pi(2j)}) ~. \\
    \hat\psi_{1,j} &= \psi_{\pi(2j-1)}  D_{\pi(2j-1)} + \psi_{\pi(2j)}  D_{\pi(2j)}\\
    \hat\psi_{0,j} &= \psi_{\pi(2j-1)} (1- D_{\pi(2j-1)} )+ \psi_{\pi(2j)}  (1-D_{\pi(2j)}) ~.
\end{align*}
We start by studying an infeasible version of $\hat \beta_G$. Let $\tilde \beta_G$ denote the OLS estimator of the slope coefficient in the linear regression of $\tilde\mu_{1,j} - \tilde\mu_{0,j}$ on a constant and $\hat\psi_{1,j} - \hat\psi_{0,j}$ with $j$ denoting the pair, where
\begin{align*}
    \tilde\mu_{1,j} &= \left(\bar Y_{\pi(2j-1)}(1) -  \frac{E[\bar Y_{g}(1) N_g]}{E[N_g]}\right) N_{\pi(2j-1)} D_{\pi(2j-1)} \\
    & \hspace{3em} + \left(\bar Y_{\pi(2j)}(1) -  \frac{E[\bar Y_{g}(1) N_g]}{E[N_g]}\right) N_{\pi(2j)} D_{\pi(2j)} \\
    \tilde\mu_{0,j} &= \left(\bar Y_{\pi(2j-1)}(0) -  \frac{E[\bar Y_{g}(0) N_g]}{E[N_g]}\right) N_{\pi(2j-1)} (1-D_{\pi(2j-1)}) \\
    & \hspace{3em} + \left(\bar Y_{\pi(2j)}(0) -  \frac{E[\bar Y_{g}(0) N_g]}{E[N_g]}\right)  N_{\pi(2j)} (1-D_{\pi(2j)})~.
\end{align*}
Lemma \ref{lem:betatilde} then implies $\tilde\beta_G \xrightarrow{P} \beta^\ast$ for $\beta^\ast$ in \eqref{eq:betastar}. Lemma \ref{lem:tilde-hat} shows $\tilde \beta_G - \hat \beta_G \xrightarrow{P} 0$. Therefore, $\hat \beta_G \xrightarrow{P} \beta^\ast$.

\noindent \underline{Step 2: Improvement in Efficiency}

We first establish the limiting distribution of $\hat \Delta_G^{\rm adj}$. Define
\[ \bar \psi_{d, G} = \frac{1}{G} \sum_{1 \leq g \leq 2G} \psi_g I \{D_g = d\} \]
for $d \in \{0, 1\}$. Note that
\begin{align*}
    & \frac{1}{G} \sum_{1 \leq g \leq 2G} (\bar Y_g(1) N_g - (\psi_g - \bar \psi_G)' \hat \beta_G) D_g \\
    & = \frac{1}{G} \sum_{1 \leq g \leq 2G} (\bar Y_g(1) N_g - (\psi_g - \bar \psi_G)' \beta^\ast) D_g - \frac{1}{G} \sum_{1 \leq g \leq 2G} (\psi_g - \bar \psi_{1, G})' (\hat \beta_G - \beta^\ast) D_g - (\bar \psi_{1, G} - \bar \psi_G)' (\hat \beta_G - \beta^\ast) \\
    & = \frac{1}{G} \sum_{1 \leq g \leq 2G} (\bar Y_g(1) N_g - (\psi_g - \bar \psi_G)' \beta^\ast) D_g - O_P(G^{-1/2}) o_P(1) \\
    & = \frac{1}{G} \sum_{1 \leq g \leq 2G} (\bar Y_g(1) N_g - (\psi_g - \bar \psi_G)' \beta^\ast) D_g + o_P(G^{-1/2}) \\
    & = \frac{1}{G} \sum_{1 \leq g \leq 2G} (\bar Y_g(1) N_g - (\psi_g - E[\psi_g])' \beta^\ast) D_g - (\bar \psi_G - E[\psi_g])' \beta^\ast + o_P(G^{-1/2})~.
\end{align*}
where the second equality follows because $\hat \beta_G - \beta^\ast = o_P(1)$,
\[ \frac{1}{G} \sum_{1 \leq g \leq 2G} (\psi_g - \bar \psi_{1, G})D_g = 0~, \]
and
\[ \sqrt G(\bar \psi_{1, G} - \bar \psi_G) = O_P(1)~. \]
The last equality follows from the arguments that establish (A.24) in \cite{bai2024covariate}. Define
\begin{align*}
    \tilde \Delta_G^{\rm adj} & = \frac{\frac{1}{G} \sum_{1 \leq g \leq 2G} (\bar Y_g(1) N_g - (\psi_g - E[\psi_g])' \beta^\ast) D_g}{\frac{1}{G} \sum_{1 \leq g \leq 2G} N_g D_g} \\
    & \hspace{3em} - \frac{\frac{1}{G} \sum_{1 \leq g \leq 2G} (\bar Y_g(0) N_g - (\psi_g - E[\psi_g])' \beta^\ast) (1 - D_g)}{\frac{1}{G} \sum_{1 \leq g \leq 2G} N_g (1 - D_g)}~.
\end{align*}
It follows from previous arguments that
\begin{align*}
    & \sqrt G(\hat \Delta_G^{\rm adj} - \Delta) - \sqrt G(\tilde \Delta_G^{\rm adj} - \Delta) \\
    & = \sqrt G (\bar \psi_G - E[\psi_g])' \beta^\ast \left ( \frac{1}{\frac{1}{G} \sum_{1 \leq g \leq 2G} N_g D_g} - \frac{1}{\frac{1}{G} \sum_{1 \leq g \leq 2G} N_g (1 - D_g)} \right ) + o_P(1) \\
    & = o_P(1)~.
\end{align*}
It follows from the proof of Theorem \ref{thm:normal_N} applied to $\bar Y_g(d) N_g - (\psi_g - E[\psi_g])' \beta^\ast$ instead of $\bar Y_g(d) N_g$ and Assumptions \ref{ass:QG}, \ref{ass:pair_form}, \ref{ass:DGP}, \ref{ass:indep_xz}, and \ref{ass:psi} that $\sqrt G(\tilde \Delta_G^{\rm adj} - \Delta) \stackrel{d}{\to} N(0, \varsigma^2)$ for $\varsigma^2$ in \eqref{eq:varsigma}.

Finally, we show that $\varsigma^2 \leq \nu^2$. First note that by definition it follows immediately that 
\[ E[(E[Y_g^\ast(1) - Y_g^\ast(0) | W_g] - \Delta)^2] = E[(E[\tilde Y_g(1) - \tilde Y_g(0) | W_g] - \Delta)^2] ~.\]
It thus remains to show that 
\[E[\var[Y_g^\ast(1) | W_g]] + E[\var[Y_g^\ast(0) | W_g]] \le E[\var[\tilde Y_g(1) | W_g]] + E[\var[\tilde Y_g(0) | W_g]]~.\]
To that end,
\begin{align*}
   & E[\var[Y_g^\ast(1) | W_g]] + E[\var[Y_g^\ast(0) | W_g]] \\
   & = E\left[\var\left[\tilde{Y}_g(1) - \frac{(\psi_g - E[\psi_g])' \beta^\ast}{E[N_g]}\Big|W_g\right]\right] + E\left[\var\left[\tilde{Y}_g(0) - \frac{(\psi_g - E[\psi_g])' \beta^\ast}{E[N_g]}\Big|W_g\right]\right] \\
   & = E[\var[\tilde{Y}_g(1) | W_g]] + E[\var[\tilde{Y}_g(0) | W_g]] + 2 E\left[\var\left[\frac{(\psi_g - E[\psi_g])' \beta^\ast}{E[N_g]}\Big|W_g\right]\right] \\
   &\hspace{3em} - 2 E\left[\cov\left[\tilde{Y}_g(1) + \tilde{Y}_g(0), \frac{(\psi_g - E[\psi_g])' \beta^\ast}{E[N_g]}\Big|W_g\right]\right]\\
   &= E[\var[\tilde{Y}_g(1) | W_g]] + E[\var[\tilde{Y}_g(0) | W_g]] + \frac{2}{E[N_g]^2} E\left[\var\left[\psi_g^\prime \beta^\ast |W_g\right]\right]\\
   &\hspace{3em} - \frac{2}{E[N_g]} E\left[\cov\left[\tilde{Y}_g(1) + \tilde{Y}_g(0), \psi_g' \beta^\ast \Big|W_g\right]\right] \\
   &= E[\var[\tilde{Y}_g(1) | W_g]] + E[\var[\tilde{Y}_g(0) | W_g]] - \frac{2}{E[N_g]^2} E\left[\var\left[\psi_g^\prime \beta^\ast |W_g\right]\right]
\end{align*}
where the first equality follows by definition, the last  equality by noting that $\beta^\ast$ is the projection coefficient of $\frac{E[N_g]}{2}(\tilde Y_g(1)  + \tilde Y_g(0)  - E[\tilde Y_g(1)  + \tilde Y_g(0) | W_g])$ on $\psi_g - E[\psi_g | W_g]$, 
\[ E[N_g]E[(\tilde Y_g(1)  + \tilde Y_g(0)  - E[\tilde Y_g(1)  + \tilde Y_g(0) | W_g]) (\psi_g - E[\psi_g | W_g])' \beta^\ast] = 2 E[((\psi_g - E[\psi_g | W_g])' \beta^\ast)^2]~, \]
or equivalently,
\begin{equation} \label{eq:projection}
   E[N_g] E[\cov[\tilde Y_g(1) + \tilde Y_g(0) , \psi_g' \beta^\ast | W_g]] = 2 E[\var[\psi_g' \beta^\ast | W_g]]~.
\end{equation}
We thus obtain
\[\varsigma^2 = \nu^2 - \kappa^2~,\]where
\begin{equation*}
    \kappa^2 = \frac{2}{E[N_g]^2} E\left[\var\left[\psi_g^\prime \beta^\ast |W_g\right]\right]~,
\end{equation*}
and the desired result follows. \qed
% Finally, note that if we do not match on $N_g$, then we have that  
% \begin{align*}
%    & E[\var[Y_g^\ast(1) | X_g]] + E[\var[Y_g^\ast(0) | X_g]] \\
%    & = E[\var[\tilde Y_g(1) | X_g]] + E[\var[\tilde Y_g(0) | X_g]] -  \frac{2}{E[N_g]^2} E\left[\var\left[\psi_g^\prime \beta^\ast |X_g\right]\right]
% \end{align*}
% and we still have the efficiency improvement.

\subsection{Proof of Theorem \ref{thm:adj-var}}
The desired result follows from combining the arguments used to establish Theorem \ref{thm:variance-estimator} and those used to establish Theorem 3.2 in \cite{bai2024covariate}. \qed

\section{Auxiliary Lemmas}

\begin{lemma} \label{lem:L_X}
Suppose $Q$ satisfies Assumptions \ref{ass:QG} and \ref{ass:DGP2} and the treatment assignment mechanism satisfies Assumptions \ref{ass:indep_pairsX}--\ref{ass:pair_formX}. Define
\begin{align*}
\mathbb L_G^{\rm YN1} & = \frac{1}{\sqrt G} \sum_{1 \leq g \leq 2G} (\bar{Y}_g(1) N_g D_g - E[\bar{Y}_g(1) N_g] D_g) \\
\mathbb L_G^{\rm N1} & = \frac{1}{\sqrt G} \sum_{1 \leq g \leq 2G} (N_g D_g - E[N_g] D_g) \\
\mathbb L_G^{\rm YN0} & = \frac{1}{\sqrt G} \sum_{1 \leq g \leq 2G} (\bar{Y}_g(0) N_g (1 - D_g) - E[\bar{Y}_g(0) N_g] (1 - D_g)) \\
\mathbb L_G^{\rm N0} & = \frac{1}{\sqrt G} \sum_{1 \leq g \leq 2G} (N_g (1 - D_g) - E[N_g] (1 - D_g))~.
\end{align*}
Then, as $G \to \infty$,
\[ (\mathbb L_G^{\rm YN1}, \mathbb L_G^{\rm N1}, \mathbb L_G^{\rm YN0}, \mathbb L_G^{\rm N0})' \stackrel{d}{\to} N(0, \mathbb V)~, \] where
\[ \mathbb V = \mathbb V_1 + \mathbb V_2 \]
for
\[ \mathbb V_1 = \begin{pmatrix}
\mathbb V_1^1 & 0 \\
0 & \mathbb V_1^0
\end{pmatrix} \]
\begin{align*}
\mathbb V_1^1 & = \begin{pmatrix}
E[\var[\bar{Y}_g(1) N_g | X_g]] & E[\cov[\bar{Y}_g(1) N_g, N_g | X_g]] \\
E[\cov[\bar{Y}_g(1) N_g, N_g | X_g]] & E[\var[N_g | X_g]]
\end{pmatrix} \\
\mathbb V_1^0 & = \begin{pmatrix}
E[\var[\bar{Y}_g(0) N_g | X_g]] & E[\cov[\bar{Y}_g(0) N_g, N_g | X_g]] \\
E[\cov[\bar{Y}_g(0) N_g, N_g | X_g]] & E[\var[N_g | X_g]]
\end{pmatrix}
\end{align*}
\[ \mathbb V_2 = \frac{1}{2} \var[(E[\bar{Y}_g(1) N_g | X_g], E[N_g | X_g], E[\bar{Y}_g(0) N_g | X_g], E[N_g | X_g])']~. \]
\end{lemma}

\begin{proof}
We break the proof into the following steps:

\noindent \underline{Step 1: Decomposition by conditioning on $X^{(G)}$ and $D^{(G)}$}

Note
\[ (\mathbb L_G^{\rm YN1}, \mathbb L_G^{\rm N1}, \mathbb L_G^{\rm YN0}, \mathbb L_G^{\rm N0}) = (\mathbb L_{1, G}^{\rm YN1}, \mathbb L_{1, G}^{\rm N1}, \mathbb L_{1, G}^{\rm YN0}, \mathbb L_{1, G}^{\rm N0}) + (\mathbb L_{2, G}^{\rm YN1}, \mathbb L_{2, G}^{\rm N1}, \mathbb L_{2, G}^{\rm YN0}, \mathbb L_{2, G}^{\rm N0})~, \]
where
\begin{align*}
\mathbb L_{1, G}^{\rm YN1} & = \frac{1}{\sqrt G} \sum_{1 \leq g \leq 2G} (\bar{Y}_g(1) N_g D_g - E[\bar{Y}_g(1) N_g D_g | X^{(G)}, D^{(G)}]) \\
\mathbb L_{2, G}^{\rm YN1} & = \frac{1}{\sqrt G} \sum_{1 \leq g \leq 2G} (E[\bar{Y}_g(1) N_g D_g | X^{(G)}, D^{(G)}] - E[\bar{Y}_g(1) N_g] D_g)
\end{align*}
and similarly for the rest. Next, note $(\mathbb L_{1, G}^{\rm YN1}, \mathbb L_{1, G}^{\rm N1}, \mathbb L_{1, G}^{\rm YN0}, \mathbb L_{1, G}^{\rm N0}), G \geq 1$ is a triangular array of mean-zero random vectors. Conditional on $X^{(G)}, D^{(G)}$, $(\mathbb L_{1, G}^{\rm YN1}, \mathbb L_{1, G}^{\rm N1}) \independent (\mathbb L_{1, G}^{\rm YN0}, \mathbb L_{1, G}^{\rm N0})$. Moreover, it follows from $Q_G = Q^{2G}$ and Assumption \ref{ass:indep_pairsX} that
\begin{align*}
& \var \left [ \begin{pmatrix}
\mathbb L_{1, G}^{\rm YN1} \\
\mathbb L_{1, G}^{\rm N1}
\end{pmatrix} \Bigg |  X^{(G)}, D^{(G)} \right ] \\
& \hspace{3em} = \begin{pmatrix}
\frac{1}{G} \sum_{1 \leq g \leq 2G} \var[\bar{Y}_g(1) N_g | X_g] D_g & \frac{1}{G} \sum_{1 \leq g \leq 2G} \cov[\bar{Y}_g(1) N_g, N_g | X_g] D_g \\
\frac{1}{G} \sum_{1 \leq g \leq 2G} \cov[\bar{Y}_g(1) N_g, N_g | X_g] D_g & \frac{1}{G} \sum_{1 \leq g \leq 2G} \var[N_g | X_g] D_g
\end{pmatrix}~.   
\end{align*}

\noindent \underline{Step 2: Limits of conditional variances}

For the upper left component, we have
\begin{equation} \label{eq:condlip}
\frac{1}{G} \sum_{1 \leq g \leq 2G} \var[\bar{Y}_g(1) N_g | X_g] D_g = \frac{1}{G} \sum_{1 \leq g \leq 2G} E[\bar{Y}^2_g(1) N_g^2 | X_g] D_g - \frac{1}{G} \sum_{1 \leq g \leq 2G} E[\bar{Y}_g(1) N_g | X_g]^2 D_g~.
\end{equation}
Note
\begin{align*}
& \frac{1}{G} \sum_{1 \leq g \leq 2G} E[\bar{Y}^2_g(1) N_g^2 | X_g] D_g \\
& = \frac{1}{2G} \sum_{1 \leq g \leq 2G} E[\bar{Y}^2_g(1) N_g^2 | X_g] + \frac{1}{2} \Big ( \frac{1}{G} \sum_{1 \leq g \leq 2G: D_g = 1} E[\bar{Y}^2_g(1) N_g^2 | X_g] - \frac{1}{G} \sum_{1 \leq g \leq 2G: D_g = 0} E[\bar{Y}^2_g(1) N_g^2 | X_g] \Big )~.
\end{align*}
It follows from the weak law of large numbers, the application of which is permitted by Lemma \ref{lem:E_bounded}, that
\[ \frac{1}{2G} \sum_{1 \leq g \leq 2G} E[\bar{Y}^2_g(1) N_g^2 | X_g] \stackrel{P}{\to} E[\bar{Y}^2_g(1) N_g^2]~. \]
On the other hand, it follows from Assumptions \ref{ass:pair_formX} and \ref{ass:DGP2}(a) that
\begin{align*}
& \Big | \frac{1}{G} \sum_{1 \leq g \leq 2G: D_g = 1} E[\bar{Y}^2_g(1) N_g^2 | X_g] - \frac{1}{G} \sum_{1 \leq g \leq 2G: D_g = 0} E[\bar{Y}^2_g(1) N_g^2 | X_g] \Big | \\
& \leq \frac{1}{G} \sum_{1 \leq j \leq G} |E[\bar{Y}_{\pi(2j - 1)}^2(1) N_{\pi(2j - 1)}^2 | X_{\pi(2j - 1)}] - E[\bar{Y}_{\pi(2j)}^2(1) N_{\pi(2j)}^2 | X_{\pi(2j)}]| \\
& \lesssim \frac{1}{G} \sum_{1 \leq j \leq G} \|X_{\pi(2j - 1)} - X_{\pi(2j)}\| \stackrel{P}{\to} 0~.
\end{align*}
Therefore,
\[ \frac{1}{G} \sum_{1 \leq g \leq 2G} E[\bar{Y}^2_g(1) N_g^2 | X_g] D_g \stackrel{P}{\to} E[\bar{Y}^2_g(1) N_g^2]~. \]
Meanwhile,
\begin{align*}
& \frac{1}{G} \sum_{1 \leq g \leq 2G} E[\bar{Y}_g(1) N_g | X_g]^2 D_g \\
& = \frac{1}{2G} \sum_{1 \leq g \leq 2G} E[\bar{Y}_g(1) N_g | X_g]^2 + \frac{1}{2} \Big ( \frac{1}{G} \sum_{1 \leq g \leq 2G: D_g = 1} E[\bar{Y}_g(1) N_g | X_g]^2 - \frac{1}{G} \sum_{1 \leq g \leq 2G: D_g = 0} E[\bar{Y}_g(1) N_g | X_g]^2 \Big )~.
\end{align*}
It follows from the weak law of large numbers, the application of which is permitted by Lemma \ref{lem:E_bounded}, that
\[ \frac{1}{2G} \sum_{1 \leq g \leq 2G} E[\bar{Y}_g(1) N_g | X_g]^2 \stackrel{P}{\to} E[E[\bar{Y}_g(1) N_g | X_g]^2]~. \]
Next,
\begin{align*}
& \Big | \frac{1}{G} \sum_{1 \leq g \leq 2G: D_g = 1} E[\bar{Y}_g(1) N_g | X_g]^2 - \frac{1}{G} \sum_{1 \leq g \leq 2G: D_g = 0} E[\bar{Y}_g(1) N_g | X_g]^2 \Big | \\
& \leq \frac{1}{G} \sum_{1 \leq j \leq G} |E[\bar{Y}_{\pi(2j - 1)}(1) N_{\pi(2j - 1)} | X_{\pi(2j - 1)}] - E[\bar{Y}_{\pi(2j)}(1) N_{\pi(2j)} | X_{\pi(2j)}]| \\
& \hspace{10em} \times |E[\bar{Y}_{\pi(2j - 1)}(1) N_{\pi(2j - 1)} | X_{\pi(2j - 1)}] + E[\bar{Y}_{\pi(2j)}(1) N_{\pi(2j)} | X_{\pi(2j)}]| \\
& \lesssim \Big ( \frac{1}{G} \sum_{1 \leq j \leq G} \|X_{\pi(2j - 1)} - X_{\pi(2j)}\|^2 \Big )^{1/2} \\
& \hspace{3em} \times \Big (\frac{1}{G} \sum_{1 \leq j \leq G} (|E[\bar{Y}_{\pi(2j - 1)}(1) N_{\pi(2j - 1)} | X_{\pi(2j - 1)}] + E[\bar{Y}_{\pi(2j)}(1) N_{\pi(2j)} | X_{\pi(2j)}]|)^2 \Big )^{1/2} \\
& \lesssim \Big ( \frac{1}{G} \sum_{1 \leq j \leq G} \|X_{\pi(2j - 1)} - X_{\pi(2j)}\|^2 \Big )^{1/2} \\
& \hspace{3em} \times \Big (\frac{1}{G} \sum_{1 \leq j \leq G} (|E[\bar{Y}_{\pi(2j - 1)}(1) N_{\pi(2j - 1)} | X_{\pi(2j - 1)}]|^2 + |E[\bar{Y}_{\pi(2j)}(1) N_{\pi(2j)} | X_{\pi(2j)}]|^2) \Big )^{1/2} \\
& \leq \Big ( \frac{1}{G} \sum_{1 \leq j \leq G} \|X_{\pi(2j - 1)} - X_{\pi(2j)}\|^2 \Big )^{1/2} \Big ( \frac{1}{G} \sum_{1 \leq g \leq 2G} E[\bar{Y}_g(1) N_g | X_g]^2 \Big )^{1/2} \stackrel{P}{\to} 0~,
\end{align*}
where the first inequality follows by inspection, the second follows from Assumption \ref{ass:DGP2}(a) and the Cauchy-Schwarz inequality, the third follows from $(a + b)^2 \leq 2 a^2 + 2 b^2$, the last follows by inspection again and the convergence in probability follows from Assumption \ref{ass:pair_formX} and the law of large numbers. Therefore, 
\[\frac{1}{G} \sum_{1 \leq g \leq 2G} E[\bar{Y}_g(1) N_g | X_g]^2 D_g \stackrel{P}{\to} E\left[E[\bar{Y}_g(1)N_g|X_g]^2\right]~,\]
and hence it follows from \eqref{eq:condlip} that
\[ \frac{1}{G} \sum_{1 \leq g \leq 2G} \var[\bar{Y}_g(1) N_g | X_g] D_g \stackrel{P}{\to} E[\var[\bar{Y}_g(1) N_g | X_g]]~. \]
An identical argument establishes that 
\[ \frac{1}{G} \sum_{1 \leq g \leq 2G} \var[N_g | X_g] D_g \stackrel{P}{\to} E[\var[N_g | X_g]]~. \]
To study the off-diagonal components, note that 
\begin{multline} \label{eq:condlipcovariance}
\frac{1}{G} \sum_{1 \leq g \leq 2G} \cov[\bar{Y}_g(1)N_g, N_g | X_g] D_g \\
= \frac{1}{G} \sum_{1 \leq g \leq 2G} E[\bar{Y}_g(1) N_g^2 | X_g] D_g - \frac{1}{G} \sum_{1 \leq g \leq 2G} E[\bar{Y}_g(1) N_g | X_g]E[N_g|X_g] D_g~.
\end{multline}
By a similar argument to that used above, it can be shown that 
\[ \frac{1}{G} \sum_{1 \leq g \leq 2G} E[\bar{Y}_g(1) N_g^2 | X_g] D_g \stackrel{P}{\to} E[\bar{Y}_g(1) N_g^2]~. \]
Meanwhile,
\begin{align*}
& \frac{1}{G} \sum_{1 \leq g \leq 2G} E[\bar{Y}_g(1) N_g | X_g]E[N_g|X_g] D_g \\
& = \frac{1}{2G} \sum_{1 \leq g \leq 2G} E[\bar{Y}_g(1) N_g | X_g]E[N_g|X_g] \\
& \hspace{3em} + \frac{1}{2} \Big ( \frac{1}{G} \sum_{1 \leq g \leq 2G: D_g = 1} E[\bar{Y}_g(1) N_g | X_g]E[N_g|X_g] - \frac{1}{G} \sum_{1 \leq g \leq 2G: D_g = 0} E[\bar{Y}_g(1) N_g | X_g]E[N_g|X_g] \Big )~.
\end{align*}
Note that
\[E[E[\bar{Y}_g(1) N_g | X_g]E[N_g|X_g]] = E[[N_gE[\bar{Y}_g(1)|W_g]|X_g]E[N_g|X_g]] \lesssim E[N_g^2] < \infty~,\]
where the equality follows by the law of iterated expectations and the inequality by Lemma \ref{lem:E_bounded} and Jensen's inequality, and the law of iterated expectations. Thus by the weak law of large numbers,
\[ \frac{1}{2G} \sum_{1 \leq g \leq 2G} E[\bar{Y}_g(1) N_g | X_g]E[N_g|X_g] \stackrel{P}{\to} E[E[\bar{Y}_g(1) N_g | X_g]E[N_g|X_g]]~. \]
Next, by the triangle inequality
\begin{align*}
& \Big | \frac{1}{G} \sum_{1 \leq g \leq 2G: D_g = 1} E[\bar{Y}_g(1) N_g | X_g]E[N_g|X_g]- \frac{1}{G} \sum_{1 \leq g \leq 2G: D_g = 0} E[\bar{Y}_g(1) N_g | X_g]E[N_g|X_g] \Big | \\
& \leq \frac{1}{G}\sum_{1 \le j \le G}\left|E[\bar{Y}_{\pi(2j - 1)}(1) N_{\pi(2j - 1)} | X_{\pi(2j - 1)}] E[N_{\pi(2j-1)}|X_{\pi(2j-1)}] \right. \\
& \hspace{6em} \left. - E[\bar{Y}_{\pi(2j)}(1) N_{\pi(2j)} | X_{\pi(2j)}] E[N_{\pi(2j)}|X_{\pi(2j)}]\right|~,
\end{align*}
and for each $j$,
\begin{align*}
& \left|E[\bar{Y}_{\pi(2j - 1)}(1) N_{\pi(2j - 1)} | X_{\pi(2j - 1)}] E[N_{\pi(2j-1)}|X_{\pi(2j-1)}] - E[\bar{Y}_{\pi(2j)}(1) N_{\pi(2j)} | X_{\pi(2j)}] E[N_{\pi(2j)}|X_{\pi(2j)}]\right|\\ 
& = \Big|(E[\bar{Y}_{\pi(2j - 1)}(1) N_{\pi(2j - 1)} | X_{\pi(2j - 1)}] - E[\bar{Y}_{\pi(2j)}(1) N_{\pi(2j)} | X_{\pi(2j)}])E[N_{\pi(2j)}|X_{\pi(2j)}] \\
& \hspace{3em} + (E[N_{\pi(2j-1)}|X_{\pi(2j-1)}] - E[N_{\pi(2j)}|X_{\pi(2j)}])E[\bar{Y}_{\pi(2j - 1)}(1) N_{\pi(2j - 1)} | X_{\pi(2j - 1)}]\Big|\\
&\lesssim \left|E[\bar{Y}_{\pi(2j - 1)}(1) N_{\pi(2j - 1)} | X_{\pi(2j - 1)}] - E[\bar{Y}_{\pi(2j)}(1) N_{\pi(2j)} | X_{\pi(2j)}]\right| \\
& \hspace{3em} +  \left|E[N_{\pi(2j - 1)} | X_{\pi(2j - 1)}] - E[N_{\pi(2j)} | X_{\pi(2j)}]\right|~,
\end{align*}
where the final inequality follows from the triangle inequality, Assumption \ref{ass:DGP2}(b) and Lemma \ref{lem:E_bounded}. Therefore,
\begin{align*}
& \Big | \frac{1}{G} \sum_{1 \leq g \leq 2G: D_g = 1} E[\bar{Y}_g(1) N_g | X_g]E[N_g|X_g]- \frac{1}{G} \sum_{1 \leq g \leq 2G: D_g = 0} E[\bar{Y}_g(1) N_g | X_g]E[N_g|X_g] \Big | \\
& \lesssim \frac{1}{G}\sum_{1 \le j \le G} \big ( \left|E[\bar{Y}_{\pi(2j - 1)}(1) N_{\pi(2j - 1)} | X_{\pi(2j - 1)}] - E[\bar{Y}_{\pi(2j)}(1) N_{\pi(2j)} | X_{\pi(2j)}]\right| \\
& \hspace{6em} + \left|E[N_{\pi(2j - 1)} | X_{\pi(2j - 1)}] - E[N_{\pi(2j)} | X_{\pi(2j)}]\right| \big ) \\
& \lesssim \frac{1}{G} \sum_{1 \leq j \leq G} \|X_{\pi(2j - 1)} - X_{\pi(2j)}\| \stackrel{P}{\to} 0~,
\end{align*}
where the final inequality follows from Assumptions \ref{ass:DGP2} and the convergence in probability follows from Assumption \ref{ass:indep_pairsX}.
Proceeding as in the case of the upper left component, we obtain that
\[ \frac{1}{G} \sum_{1 \leq g \leq 2G} \cov[\bar{Y}_g(1)N_g, N_g | X_g] D_g \stackrel{P}{\to} E[ \cov[\bar{Y}_g(1)N_g, N_g | X_g]]~. \]
Thus we have established that
\[ \var \left [ \begin{pmatrix}
\mathbb L_{1, G}^{\rm YN1} \\
\mathbb L_{1, G}^{\rm N1}
\end{pmatrix} \Bigg |  X^{(G)}, D^{(G)} \right ] \stackrel{P}{\to} \mathbb V_1^1~. \]
Similarly,
\[ \var \left [ \begin{pmatrix}
\mathbb L_{1, G}^{\rm YN0} \\
\mathbb L_{1, G}^{\rm N0}
\end{pmatrix} \Bigg |  X^{(G)}, D^{(G)} \right ] \stackrel{P}{\to} \mathbb V_1^0~. \]

\noindent \underline{Step 3: Conditional CLT}

We now establish 
\begin{equation} \label{eq:cond}
\rho(\mathcal L((\mathbb L_{1, G}^{\rm YN1}, \mathbb L_{1, G}^{\rm N1}, \mathbb L_{1, G}^{\rm YN0}, \mathbb L_{1, G}^{\rm N0})' | X^{(G)}, D^{(G)}),  N(0, \mathbb V_1)) \stackrel{P}{\to} 0~,
\end{equation}
where $\mathcal L(\cdot)$ is used to denote the law of a random variable and $\rho$ is any metric that metrizes weak convergence. For that purpose, note that we only need to show that for any subsequence $\{G_k\}$ there exists a further subsequence $\{G_{k_l}\}$ along which
\begin{equation} \label{eq:cond-as}
\rho(\mathcal L((\mathbb L_{1, G_{k_l}}^{\rm YN1}, \mathbb L_{1, G_{k_l}}^{\rm N1}, \mathbb L_{1, G_{k_l}}^{\rm YN0}, L_{1, G_{k_l}}^{\rm N0}) | X^{(G_{k_l})}, D^{(G^{k_l})},  N(0, \mathbb V_1)) \rightarrow 0 \text{ with probability one}~.
\end{equation}
In order to extract such a subsequence, we verify the conditions in the Lindeberg central limit theorem in Proposition 2.27 of \cite{van_der_vaart1998asymptotic} are satisfied in probability for the original sequence, because then we can extract a subsequence along which the conditions in that proposition hold almost surely. The second condition in that proposition is satisfied because we have shown
\[ \var[ (\mathbb L_{1, G}^{\rm YN1}, \mathbb L_{1, G}^{\rm N1}, \mathbb L_{1, G}^{\rm YN0}, \mathbb L_{1, G}^{\rm N0})' | X^{(G)}, D^{(G)}] \stackrel{P}{\to} \mathbb V_1~. \]
The first condition in that proposition can be verified component wise because of the following inequality:
\begin{equation} \label{eq:indicator}
    \left | \sum_{1 \leq j \leq k} a_j \right | I \left \{ \left | \sum_{1 \leq j \leq k} a_j \right | > \epsilon \right \} \leq \sum_{1 \leq j \leq k} k |a_j| I \left \{ |a_j| > \frac{\epsilon}{k} \right \}~.
\end{equation}
Therefore, we will only verify that 
\begin{multline} \label{eq:vdv}
    \frac{1}{G} \sum_{1 \leq g \leq 2G} E[(D_g(\bar{Y}_g(1) N_g - E[\bar{Y}_g(1) N_g | X_g]))^2 \\
    \times I \{(D_g(\bar{Y}_g(1) N_g - E[\bar{Y}_g(1) N_g | X_g]))^2 > \epsilon^2 G\} | X^{(G)}, D^{(G)} ] \xrightarrow{P} 0
\end{multline}

To verify \eqref{eq:vdv}, note it follows from \eqref{eq:indicator} that
\begin{align*}
    & \frac{1}{G} \sum_{1 \leq g \leq 2G} E[(D_g(\bar{Y}_g(1) N_g - E[\bar{Y}_g(1) N_g | X_g]))^2 I \{(D_g(\bar{Y}_g(1) N_g - E[\bar{Y}_g(1) N_g | X_g]))^2 > \epsilon^2 G\} | X^{(G)}, D^{(G)} ] \\
    & \lesssim \frac{1}{G} \sum_{1 \leq g \leq 2G} E[D_g(\bar{Y}_g(1) N_g - E[\bar{Y}_g(1) N_g | X_g])^2 I \{D_g(\bar{Y}_g(1) N_g - E[\bar{Y}_g(1) N_g | X_g])^2 > \epsilon^2 G / 2\} | X^{(G)}, D^{(G)}] \\
    & \leq \frac{1}{G} \sum_{1 \leq g \leq 2G} E[(\bar{Y}_g(1) N_g - E[\bar{Y}_g(1) N_g | X_g])^2 I \{|\bar{Y}_g(1) N_g - E[\bar{Y}_g(1) N_g | X_g]| > \epsilon \sqrt{G} / \sqrt{2}\} | X_g]~.
\end{align*}
Fix any $m > 0$. For $G$ large enough, the previous line
\begin{align*}
    & \leq \frac{1}{G} \sum_{1 \leq g \leq 2G} E[(\bar{Y}_g(1) N_g - E[\bar{Y}_g(1) N_g | X_g])^2 I \{|\bar{Y}_g(1) N_g - E[\bar{Y}_g(1) N_g | X_g]| > m\} | X_g] \\
    & \stackrel{P}{\to} 2 E[(\bar{Y}_g(1) N_g - E[\bar{Y}_g(1) N_g | X_g])^2 I \{|\bar{Y}_g(1) N_g - E[\bar{Y}_g(1) N_g | X_g]| > m\}]
\end{align*}
because $E[(\bar{Y}_g(1) N_g - E[\bar{Y}_g(1) N_g | X_g])^2] < \infty$. As $m \to \infty$, the last expression goes to $0$. Therefore, it follows from a similar diagonalization argument to that in the proof of Lemma B.3 of \cite{bai2022optimality} that both conditions in Proposition 2.27 of \cite{van_der_vaart1998asymptotic} hold in probability, and therefore there must be a subsequence along which they hold almost surely, so \eqref{eq:cond-as} and hence \eqref{eq:cond} holds.
%If $E[\var[\bar{Y}_g(1) N_g | X_g]] = E[\var[N_g | X_g]] = E[\var[\bar{Y}_g(0) N_g | X_g]] = 0$, then it follows from Markov's inequality conditional on $X^{(G)}$ and $D^{(G)}$, and the fact that probabilities are bounded and hence uniformly integrable, that $(\mathbb L_{1, G}^{\rm YN1}, \mathbb L_{1, G}^{\rm N1}, \mathbb L_{1, G}^{\rm YN0}, \mathbb L_{1, G}^{\rm N0}) \stackrel{P}{\to} 0$. Otherwise,

\noindent \underline{Step 4: Unconditional components}

Next, we study $(\mathbb L_{2, G}^{\rm YN1}, \mathbb L_{2, G}^{\rm N1}, \mathbb L_{2, G}^{\rm YN0}, \mathbb L_{2, G}^{\rm N0})$. It follows from $Q_G = Q^{2G}$ and Assumption \ref{ass:indep_pairsX} that
\[ \begin{pmatrix}
\mathbb L_{2, G}^{\rm YN1} \\
\mathbb L_{2, G}^{\rm N1} \\
\mathbb L_{2, G}^{\rm YN0} \\
\mathbb L_{2, G}^{\rm N0}
\end{pmatrix} = \begin{pmatrix}
\frac{1}{\sqrt G} \sum_{1 \leq g \leq 2G} D_g (E[\bar{Y}_g(1) N_g | X_g] - E[\bar{Y}_g(1) N_g]) \\
\frac{1}{\sqrt G} \sum_{1 \leq g \leq 2G} D_g (E[N_g | X_g] - E[N_g]) \\
\frac{1}{\sqrt G} \sum_{1 \leq g \leq 2G} (1 - D_g) (E[\bar{Y}_g(0) N_g | X_g] - E[\bar{Y}_g(0) N_g]) \\
\frac{1}{\sqrt G} \sum_{1 \leq g \leq 2G} (1 - D_g) (E[N_g | X_g] - E[N_g])
\end{pmatrix}~. \]
For $\mathbb L_{2, G}^{\rm YN1}$, note it follows from Assumption \ref{ass:indep_pairsX} that
\begin{align*}
\var[\mathbb L_{2, G}^{\rm YN1} | X^{(G)}] & = \frac{1}{4G} \sum_{1 \leq j \leq G} (E[\bar{Y}_{\pi(2j - 1)}(1) N_{\pi(2j - 1)} | X_{\pi(2j - 1)}] - E[\bar{Y}_{\pi(2j)}(1) N_{\pi(2j)} | X_{\pi(2j)}])^2 \\
& \lesssim \frac{1}{G} \sum_{1 \leq j \leq G} \|X_{\pi(2j - 1)} - X_{\pi(2j)}\|^2 \stackrel{P}{\to} 0~.
\end{align*}
Therefore, it follows from Markov's inequality conditional on $X^{(G)}$ and $D^{(G)}$, and the fact that probabilities are bounded and hence uniformly integrable, that
\[ \mathbb L_{2, G}^{\rm YN1} = E[\mathbb L_{2, G}^{\rm YN1} | X^{(G)}] + o_P(1)~. \]
Applying a similar argument to each of $L_{2, G}^{\rm N1}$, $L_{2, G}^{\rm YN0}$, $L_{2, G}^{\rm N0}$ allows us to conclude that
\[ \begin{pmatrix}
\mathbb L_{2, G}^{\rm YN1} \\
\mathbb L_{2, G}^{\rm N1} \\
\mathbb L_{2, G}^{\rm YN0} \\
\mathbb L_{2, G}^{\rm N0}
\end{pmatrix} = \begin{pmatrix}
\frac{1}{2\sqrt G} \sum_{1 \leq g \leq 2G} (E[\bar{Y}_g(1) N_g | X_g] - E[\bar{Y}_g(1) N_g]) \\
\frac{1}{2\sqrt G} \sum_{1 \leq g \leq 2G} (E[N_g | X_g] - E[N_g]) \\
\frac{1}{2\sqrt G} \sum_{1 \leq g \leq 2G} (E[\bar{Y}_g(0) N_g | X_g] - E[\bar{Y}_g(0) N_g]) \\
\frac{1}{2\sqrt G} \sum_{1 \leq g \leq 2G} (E[N_g | X_g] - E[N_g])
\end{pmatrix} + o_P(1)~. \]
It thus follows from the central limit theorem, the application of which is justified by Jensen's inequality combined with Assumption \ref{ass:QG}(b) and Lemma \ref{lem:E_bounded}, that
\[ (\mathbb L_{2, G}^{\rm YN1}, \mathbb L_{2, G}^{\rm N1}, \mathbb L_{2, G}^{\rm YN0}, \mathbb L_{2, G}^{\rm N0})' \stackrel{d}{\to} N(0, \mathbb V_2)~. \]

\noindent \underline{Step 5: Combining unconditional and conditional components}

Because \eqref{eq:cond} holds and $(\mathbb L_{2, G}^{\rm YN1}, \mathbb L_{2, G}^{\rm N1}, \mathbb L_{2, G}^{\rm YN0}, \mathbb L_{2, G}^{\rm N0})$ is deterministic conditional on $X^{(G)}, D^{(G)}$, the conclusion of the theorem follows from Lemma S.1.3 in \cite{bai2022inference}.
\end{proof}

\begin{lemma}\label{lem:L_X_algebra}
Let $\mathbb{V}$ be defined as in Lemma \ref{lem:L_X}, and $D_{h0}$ be defined as in the proof of Theorem \ref{thm:normal_X}, then
\[D_{h0}\mathbb{V}D_{h0}' = \omega^2~,\]
where
\[\omega^2 = E[\tilde Y_g^2(1)] + E[\tilde Y_g^2(0)] - \frac{1}{2} E[(E[\tilde{Y}_g(1) + \tilde{Y}_g(0) | X_g])^2]~.\]
\end{lemma}
\begin{proof}
To see this, note by the laws of total variance and total covariance that $\mathbb V$ in Lemma \ref{lem:L_X} is symmetric with entries
\begin{align*}
\mathbb V_{11} & = \var[\bar{Y}_g(1) N_g] - \frac{1}{2} \var[E[\bar{Y}_g(1) N_g | X_g]] \\
\mathbb V_{12} & = \cov[\bar{Y}_g(1) N_g, N_g] - \frac{1}{2} \cov[E[\bar{Y}_g(1) N_g | X_g], E[N_g | X_g]] \\
\mathbb V_{13} & = \frac{1}{2} \cov[E[\bar{Y}_g(1) N_g | X_g], E[\bar{Y}_g(0) N_g | X_g]] \\
\mathbb V_{14} & = \frac{1}{2} \cov[E[\bar{Y}_g(1) N_g | X_g], E[N_g | X_g]] \\
\mathbb V_{22} & = \var[N_g] - \frac{1}{2} \var[E[N_g | X_g]] \\
\mathbb V_{23} & = \frac{1}{2} \cov[E[N_g | X_g], E[\bar{Y}_g(0) N_g | X_g]] \\
\mathbb V_{24} & = \frac{1}{2} \cov[E[N_g | X_g], E[N_g | X_g]] \\
\mathbb V_{33} & = \var[\bar{Y}_g(0) N_g] - \frac{1}{2} \var[E[\bar{Y}_g(0) N_g | X_g]] \\
\mathbb V_{34} & = \cov[\bar{Y}_g(0) N_g, N_g] - \frac{1}{2} \cov[E[\bar{Y}_g(0) N_g | X_g], E[N_g | X_g]] \\
\mathbb V_{44} & = \var[N_g] - \frac{1}{2} \var[E[N_g | X_g]]~.
\end{align*}
We separately calculate the variance terms involving conditional expectations and those that don't. The terms not involving conditional expectations are
\begin{align*}
& \frac{\var[\bar{Y}_g(1) N_g]}{E[N_g]^2} + \frac{\var[N_g] E[\bar{Y}_g(1) N_g]^2}{E[N_g]^4} + \frac{\var[\bar{Y}_g(0) N_g]}{E[N_g]^2} + \frac{\var[N_g] E[\bar{Y}_g(0) N_g]^2}{E[N_g]^4} \\
& \hspace{3em} - \frac{2 \cov[\bar{Y}_g(1) N_g, N_g] E[\bar{Y}_g(1) N_g]}{E[N_g]^3} - \frac{2 \cov[\bar{Y}_g(0) N_g, N_g] E[\bar{Y}_g(0) N_g]}{E[N_g]^3} \\
& = \frac{E[\bar{Y}_g^2(1) N_g^2] - E[\bar{Y}_g(1) N_g]^2}{E[N_g]^2} + \frac{E[N_g^2] E[\bar{Y}_g(1) N_g]^2 - E[N_g]^2 E[\bar{Y}_g(1) N_g]^2}{E[N_g]^4} \\
& \hspace{3em} + \frac{E[\bar{Y}_g^2(0) N_g^2] - E[\bar{Y}_g(0) N_g]^2}{E[N_g]^2} + \frac{E[N_g^2] E[\bar{Y}_g(0) N_g]^2 - E[N_g]^2 E[\bar{Y}_g(0) N_g]^2}{E[N_g]^4} \\
& \hspace{3em} - \frac{2 E[\bar{Y}_g(1) N_g^2] E[\bar{Y}_g(1) N_g]}{E[N_g]^3} + \frac{2 E[\bar{Y}_g(1) N_g] E[N_g] E[\bar{Y}_g(1) N_g]}{E[N_g]^3} \\
& \hspace{3em} - \frac{2 E[\bar{Y}_g(0) N_g^2] E[\bar{Y}_g(0) N_g]}{E[N_g]^3} + \frac{2 E[\bar{Y}_g(0) N_g] E[N_g] E[\bar{Y}_g(0) N_g]}{E[N_g]^3} \\
& = \frac{E[\bar{Y}_g^2(1) N_g^2]}{E[N_g]^2} + \frac{E[\bar{Y}_g^2(0) N_g^2]}{E[N_g]^2} + \frac{E[N_g^2] E[\bar{Y}_g(1) N_g]^2}{E[N_g]^4} + \frac{E[N_g^2] E[\bar{Y}_g(0) N_g]^2}{E[N_g]^4} \\
& \hspace{3em} - \frac{2 E[\bar{Y}_g(1) N_g^2] E[\bar{Y}_g(1) N_g]}{E[N_g]^3} - \frac{2 E[\bar{Y}_g(0) N_g^2] E[\bar{Y}_g(0) N_g]}{E[N_g]^3} \\
& = E[\tilde Y_g^2(1)] + E[\tilde Y_g^2(0)]~,
\end{align*}
where 
\begin{equation*}
    \tilde Y_g(d) = \frac{N_g}{E[N_g]} \left ( \bar{Y}_g(d) - \frac{E[\bar{Y}_g(d) N_g]}{E[N_g]} \right )
\end{equation*}
for $d \in \{0, 1\}$.

Next, the terms involving conditional expectations are
\begin{align*}
& - \frac{\var[E[\bar{Y}_g(1) N_g | X_g]]}{2E[N_g]^2} - \frac{\var[E[N_g | X_g]] E[\bar{Y}_g(1) N_g]^2}{2 E[N_g]^4} \\
& \hspace{3em} - \frac{\var[E[\bar{Y}_g(0) N_g | X_g]]}{2E[N_g]^2} - \frac{\var[E[N_g | X_g]] E[\bar{Y}_g(0) N_g]^2}{2 E[N_g]^4} \\
& \hspace{3em} + \frac{\cov[E[\bar{Y}_g(1) N_g | X_g], E[N_g | X_g]] E[\bar{Y}_g(1) N_g]}{E[N_g]^3} + \frac{\cov[E[\bar{Y}_g(0) N_g | X_g], E[N_g | X_g]] E[\bar{Y}_g(0) N_g]}{E[N_g]^3} \\
& \hspace{3em}- \frac{\cov[E[\bar{Y}_g(1) N_g | X_g], E[\bar{Y}_g(0) N_g | X_g]]}{E[N_g]^2}  + \frac{\cov[E[\bar{Y}_g(1) N_g | X_g], E[N_g | X_g]] E[\bar{Y}_g(0) N_g]}{E[N_g] E[N_g]^2} \\
& \hspace{3em} + \frac{\cov[E[N_g | X_g], E[\bar{Y}_g(0) N_g | X_g]] E[\bar{Y}_g(1) N_g]}{E[N_g]^2 E[N_g]} - \frac{\cov[E[N_g | X_g], E[N_g | X_g]] E[\bar{Y}_g(1) N_g] E[\bar{Y}_g(0) N_g]}{E[N_g]^2 E[N_g]^2} \\
& = - \frac{E[E[\bar{Y}_g(1) N_g | X_g]^2] - E[\bar{Y}_g(1) N_g]^2}{2E[N_g]^2} - \frac{(E[E[N_g | X_g]^2] - E[N_g]^2) E[\bar{Y}_g(1) N_g]^2}{2 E[N_g]^4} \\
& \hspace{3em} - \frac{E[E[\bar{Y}_g(0) N_g | X_g]^2] - E[\bar{Y}_g(0) N_g]^2}{2E[N_g]^2} - \frac{(E[E[N_g | X_g]^2] - E[N_g]^2) E[\bar{Y}_g(0) N_g]^2}{2 E[N_g]^4} \\
& \hspace{3em} + \frac{(E[E[\bar{Y}_g(1) N_g | X_g] E[N_g | X_g]] - E[\bar{Y}_g(1) N_g] E[N_g]) E[\bar{Y}_g(1) N_g]}{E[N_g]^3} \\
& \hspace{3em} + \frac{(E[E[\bar{Y}_g(0) N_g | X_g] E[N_g | X_g]] - E[\bar{Y}_g(0) N_g] E[N_g]) E[\bar{Y}_g(0) N_g]}{E[N_g]^3} \\
& \hspace{3em} - \frac{E[E[\bar{Y}_g(1) N_g | X_g] E[\bar{Y}_g(0) N_g | X_g]] - E[\bar{Y}_g(1) N_g] E[\bar{Y}_g(0) N_g]}{E[N_g] E[N_g]} \\
& \hspace{3em} + \frac{(E[E[\bar{Y}_g(1) N_g | X_g] E[N_g | X_g]] - E[\bar{Y}_g(1) N_g] E[N_g]) E[\bar{Y}_g(0) N_g]}{E[N_g] E[N_g]^2} \\
& \hspace{3em} + \frac{(E[E[\bar{Y}_g(0) N_g | X_g] E[N_g | X_g]] - E[\bar{Y}_g(0) N_g] E[N_g]) E[\bar{Y}_g(1) N_g]}{E[N_g]^2 E[N_g]} \\
& \hspace{3em} - \frac{(E[E[N_g | X_g] E[N_g | X_g]] - E[N_g] E[N_g]) E[\bar{Y}_g(1) N_g] E[\bar{Y}_g(0) N_g]}{E[N_g]^2 E[N_g]^2} \\
& = - \frac{E[E[\bar{Y}_g(1) N_g | X_g]^2] }{2E[N_g]^2} - \frac{E[E[N_g | X_g]^2]E[\bar{Y}_g(1)N_g]^2}{2 E[N_g]^4} - \frac{E[E[\bar{Y}_g(0) N_g | X_g]^2] }{2E[N_g]^2} - \frac{E[E[N_g | X_g]^2]E[\bar{Y}_g(0)N_g]^2}{2 E[N_g]^4} \\
& \hspace{3em} + \frac{E[E[\bar{Y}_g(1) N_g | X_g] E[N_g | X_g]] E[\bar{Y}_g(1) N_g]}{E[N_g]^3} + \frac{E[E[\bar{Y}_g(0) N_g | X_g] E[N_g | X_g]] E[\bar{Y}_g(0) N_g]}{E[N_g]^3} \\
& \hspace{3em} - \frac{E[E[\bar{Y}_g(1) N_g | X_g] E[\bar{Y}_g(0) N_g | X_g]]}{E[N_g]^2} + \frac{E[E[\bar{Y}_g(1) N_g | X_g] E[N_g | X_g]] E[\bar{Y}_g(0) N_g]}{E[N_g]^3} \\
& \hspace{3em} + \frac{E[E[\bar{Y}_g(0) N_g | X_g] E[N_g | X_g]] E[\bar{Y}_g(1) N_g]}{E[N_g]^3} - \frac{E[E[N_g | X_g]^2]E[\bar{Y}_g(1)N_g]E[\bar{Y}_g(0)N_g]}{E[N_g]^4} \\
& = - \frac{1}{2} E[E[\tilde{Y}_g(1) | X_g]^2] - \frac{1}{2} E[E[\tilde{Y}_g(0) | X_g]^2] - E[E[\tilde{Y}_g(1) | X_g] E[\tilde{Y}_g(0) | X_g]]\\
& = - \frac{1}{2} E[(E[\tilde{Y}_g(1) + \tilde{Y}_g(0) | X_g])^2]~,
\end{align*}
as desired.
\end{proof}

\begin{lemma} \label{lem:L_N}
Suppose $Q$ satisfies Assumptions \ref{ass:QG} and \ref{ass:DGP} and the treatment assignment mechanism satisfies Assumptions \ref{ass:indep_pairs}--\ref{ass:pair_form}. Define
\begin{align*}
\mathbb L_G^{\rm YN1} & = \frac{1}{\sqrt G} \sum_{1 \leq g \leq 2G} (\bar{Y}_g(1) N_g D_g - E[\bar{Y}_g(1) N_g] D_g) \\
\mathbb L_G^{\rm N1} & = \frac{1}{\sqrt G} \sum_{1 \leq g \leq 2G} (N_g D_g - E[N_g] D_g) \\
\mathbb L_G^{\rm YN0} & = \frac{1}{\sqrt G} \sum_{1 \leq g \leq 2G} (\bar{Y}_g(0) N_g (1 - D_g) - E[\bar{Y}_g(0) N_g] (1 - D_g)) \\
\mathbb L_G^{\rm N0} & = \frac{1}{\sqrt G} \sum_{1 \leq g \leq 2G} (N_g (1 - D_g) - E[N_g] (1 - D_g))~.
\end{align*}
Then, as $G \to \infty$,
\[ (\mathbb L_G^{\rm YN1}, \mathbb L_G^{\rm N1}, \mathbb L_G^{\rm YN0}, \mathbb L_G^{\rm N0})' \stackrel{d}{\to} N(0, \mathbb V)~, \] where
\[ \mathbb V = \mathbb V_1 + \mathbb V_2 \]
for
\[ \mathbb V_1 = \begin{pmatrix}
\mathbb V_1^1 & 0 \\
0 & \mathbb V_1^0
\end{pmatrix} \]
\begin{align*}
\mathbb V_1^1 & = \begin{pmatrix}
E[\var[\bar{Y}_g(1) N_g | W_g]] & 0 \\
0 &0
\end{pmatrix} \\
\mathbb V_1^0 & = \begin{pmatrix}
E[\var[\bar{Y}_g(0) N_g | W_g]] & 0 \\
0 & 0
\end{pmatrix}
\end{align*}
\[ \mathbb V_2 = \frac{1}{2} \var[(E[\bar{Y}_g(1) N_g | W_g], N_g , E[\bar{Y}_g(0) N_g | W_g], N_g)']~. \]
\end{lemma}

\begin{proof}
We will only verify Steps 1 and 2 in the proof of Lemma \ref{lem:L_X} because Steps 3--5 are identical. Note
\[ (\mathbb L_G^{\rm YN1}, \mathbb L_G^{\rm N1}, \mathbb L_G^{\rm YN0}, \mathbb L_G^{\rm N0}) = (\mathbb L_{1, G}^{\rm YN1}, 0, \mathbb L_{1, G}^{\rm YN0}, 0) + (\mathbb L_{2, G}^{\rm YN1}, \mathbb L_{G}^{\rm N1}, \mathbb L_{2, G}^{\rm YN0}, \mathbb L_{G}^{\rm N0})~, \]
where
\begin{align*}
\mathbb L_{1, G}^{\rm YN1} & = \frac{1}{\sqrt G} \sum_{1 \leq g \leq 2G} (\bar{Y}_g(1) N_g D_g - E[\bar{Y}_g(1) N_g D_g | N^{(G)}, X^{(G)}, D^{(G)}]) \\
\mathbb L_{2, G}^{\rm YN1} & = \frac{1}{\sqrt G} \sum_{1 \leq g \leq 2G} (E[\bar{Y}_g(1) N_g D_g |N^{(G)}, X^{(G)}, D^{(G)}] - E[\bar{Y}_g(1) N_g] D_g)
\end{align*}
and similarly for $\mathbb L^{\rm YN0}_G$. Next, note $(\mathbb L_{1, G}^{\rm YN1},0, \mathbb L_{1, G}^{\rm YN0}, 0), G \geq 1$ is a triangular array of normalized sums of random vectors. Conditional on $N^{(G)}, X^{(G)}, D^{(G)}$, $\mathbb L_{1, G}^{\rm YN1} \independent \mathbb L_{1, G}^{\rm YN0}$. Moreover, it follows from $Q_G = Q^{2G}$ and Assumption \ref{ass:indep_pairs} that
\[ \var \left [ \mathbb L_{1, G}^{\rm YN1} \Bigg | N^{(G)}, X^{(G)}, D^{(G)} \right ] = \var[\bar{Y}_g(1) N_g | W_g] D_g ~. \]

We have
\begin{equation} \label{eq:condlipN}
\frac{1}{G} \sum_{1 \leq g \leq 2G} \var[\bar{Y}_g(1) N_g | W_g] D_g = \frac{1}{G} \sum_{1 \leq g \leq 2G} E[\bar{Y}^2_g(1) N_g^2 | W_g] D_g - \frac{1}{G} \sum_{1 \leq g \leq 2G} E[\bar{Y}_g(1) N_g | W_g]^2 D_g~.
\end{equation}
Note
\begin{align*}
& \frac{1}{G} \sum_{1 \leq g \leq 2G} E[\bar{Y}^2_g(1) N_g^2 | W_g] D_g \\
& = \frac{1}{2G} \sum_{1 \leq g \leq 2G} E[\bar{Y}^2_g(1) N_g^2 | W_g] + \frac{1}{2} \Big ( \frac{1}{G} \sum_{1 \leq g \leq 2G: D_g = 1} E[\bar{Y}^2_g(1) N_g^2 | W_g] - \frac{1}{G} \sum_{1 \leq g \leq 2G: D_g = 0} E[\bar{Y}^2_g(1) N_g^2 | W_g] \Big )~.
\end{align*}
It follows from the weak law of large numbers, the application of which is permitted by Lemma \ref{lem:E_bounded},
\[ \frac{1}{2G} \sum_{1 \leq g \leq 2G} E[\bar{Y}^2_g(1) N_g^2 | W_g] \stackrel{P}{\to} E[\bar{Y}^2_g(1) N_g^2]~. \]
On the other hand, 
\begin{align*}
& \Big | \frac{1}{G} \sum_{1 \leq g \leq 2G: D_g = 1} E[\bar{Y}^2_g(1) N_g^2 | W_g] - \frac{1}{G} \sum_{1 \leq g \leq 2G: D_g = 0} E[\bar{Y}^2_g(1) N_g^2 | W_g] \Big | \\
& \leq \frac{1}{G} \sum_{1 \leq j \leq G} |N_{\pi(2j - 1)}^2E[\bar{Y}_{\pi(2j - 1)}^2(1)  | W_{\pi(2j - 1)}] - N_{\pi(2j)}^2 E[\bar{Y}_{\pi(2j)}^2(1) | W_{\pi(2j)}]| \\
&\leq \frac{1}{G}\sum_{1 \le j \le G}N^2_{\pi(2j)}|E[\bar{Y}_{\pi(2j - 1)}^2(1)|W_{\pi(2j-1)}] - E[\bar{Y}_{\pi(2j)}^2(1) | W_{\pi(2j)}]| \\
& \hspace{3em} + \frac{1}{G}\sum_{1 \le j \le G}|N^2_{\pi(2j)} - N^2_{\pi(2j-1)}\|E[\bar{Y}_{\pi(2j - 1)}^2(1)|W_{\pi(2j-1)}]|\\
& \lesssim \frac{1}{G}\sum_{1 \le j \le G}N^2_{\pi(2j)}\|W_{\pi(2j-1)} - W_{\pi(2j)}\| + \frac{1}{G}\sum_{1 \le j \le G}|N^2_{\pi(2j)} - N^2_{\pi(2j-1)}|  \stackrel{P}{\to} 0~,
\end{align*}
where the first inequality follows from Assumption \ref{ass:indep_pairs} and the triangle inequality, the second inequality by some algebraic manipulations, the final inequality by Assumption \ref{ass:DGP} and Lemma \ref{lem:E_bounded}, and the convergence in probability follows from Assumption \ref{ass:pair_form} and Lemmas \ref{lem:pair_form} and \ref{lem:pair_squared}.
%On the other hand, it follows from Assumption \ref{ass:DGP2}(LABEL) and \ref{ass:pair_formX} that
Therefore,
\[ \frac{1}{G} \sum_{1 \leq g \leq 2G} E[\bar{Y}^2_g(1) N_g^2 | W_g] D_g \stackrel{P}{\to} E[\bar{Y}^2_g(1) N_g^2]~. \]
Meanwhile,
\begin{align*}
& \frac{1}{G} \sum_{1 \leq g \leq 2G} E[\bar{Y}_g(1) N_g | W_g]^2 D_g \\
& = \frac{1}{2G} \sum_{1 \leq g \leq 2G} E[\bar{Y}_g(1) N_g | W_g]^2 + \frac{1}{2} \Big ( \frac{1}{G} \sum_{1 \leq g \leq 2G: D_g = 1} E[\bar{Y}_g(1) N_g | W_g]^2 - \frac{1}{G} \sum_{1 \leq g \leq 2G: D_g = 0} E[\bar{Y}_g(1) N_g | W_g]^2 \Big )~.
\end{align*}
It follows from the weak law of large numbers, the application of which is permitted by Lemma \ref{lem:E_bounded} and Assumption \ref{ass:QG}(c) that
\[ \frac{1}{2G} \sum_{1 \leq g \leq 2G} E[\bar{Y}_g(1) N_g | W_g]^2 \stackrel{P}{\to} E[E[\bar{Y}_g(1) N_g | W_g]^2]~. \]
Next,
\begin{align*}
& \Big | \frac{1}{G} \sum_{1 \leq g \leq 2G: D_g = 1} E[\bar{Y}_g(1) N_g | W_g]^2 - \frac{1}{G} \sum_{1 \leq g \leq 2G: D_g = 0} E[\bar{Y}_g(1) N_g | W_g]^2 \Big | \\
& \leq \frac{1}{G} \sum_{1 \leq j \leq G} |E[\bar{Y}_{\pi(2j - 1)}(1) N_{\pi(2j - 1)} | W_{\pi(2j - 1)}] - E[\bar{Y}_{\pi(2j)}(1) N_{\pi(2j)} | W_{\pi(2j)}]| \\
& \hspace{10em} \times |E[\bar{Y}_{\pi(2j - 1)}(1) N_{\pi(2j - 1)} | W_{\pi(2j - 1)}] + E[\bar{Y}_{\pi(2j)}(1) N_{\pi(2j)} | W_{\pi(2j)}]| \\
& \leq \Big ( \frac{1}{G} \sum_{1 \leq j \leq G} |E[\bar{Y}_{\pi(2j - 1)}(1) N_{\pi(2j - 1)} | W_{\pi(2j - 1)}] - E[\bar{Y}_{\pi(2j)}(1) N_{\pi(2j)} | W_{\pi(2j)}]| ^2 \Big )^{1/2} \\
& \hspace{3em} \times \Big (\frac{1}{G} \sum_{1 \leq j \leq G} |E[\bar{Y}_{\pi(2j - 1)}(1) N_{\pi(2j - 1)} | W_{\pi(2j - 1)}] + E[\bar{Y}_{\pi(2j)}(1) N_{\pi(2j)} | W_{\pi(2j)}]|^2 \Big )^{1/2} \\
& \lesssim \Big ( \frac{1}{G} \sum_{1 \leq j \leq G} |E[\bar{Y}_{\pi(2j - 1)}(1) N_{\pi(2j - 1)} | W_{\pi(2j - 1)}] - E[\bar{Y}_{\pi(2j)}(1) N_{\pi(2j)} | W_{\pi(2j)}]|^2 \Big )^{1/2} \\
& \hspace{3em} \times \Big ( \frac{1}{G} \sum_{1 \leq g \leq 2G} E[\bar{Y}_g(1) N_g | W_g]^2 \Big )^{1/2} \stackrel{P}{\to} 0~,
\end{align*}
where the first inequality follows by inspection, the second follows from Cauchy-Schwarz, the third follows from $(a + b)^2 \leq 2 a^2 + 2 b^2$, and the convergence in probability follows from Assumptions \ref{ass:pair_form}--\ref{ass:DGP}, Lemma \ref{lem:pair_form}, and the weak law of large numbers. Therefore, 
\[\frac{1}{G} \sum_{1 \leq g \leq 2G} E[\bar{Y}_g(1) N_g | W_g]^2 D_g \stackrel{P}{\to} E\left[E[\bar{Y}_g(1)N_g|W_g]^2\right]~,\]
and hence it follows from \eqref{eq:condlipN} that
\[ \frac{1}{G} \sum_{1 \leq g \leq 2G} \var[\bar{Y}_g(1) N_g | W_g] D_g \stackrel{P}{\to} E[\var[\bar{Y}_g(1) N_g | W_g]]~. \]
Similarly,
\[ \frac{1}{G} \sum_{1 \leq g \leq 2G} \var[\bar{Y}_g(0) N_g | W_g] D_g \stackrel{P}{\to} E[\var[\bar{Y}_g(0) N_g | W_g]]~. \]
Putting these results together, we obtain
\[\var[(\mathbb L_{1, G}^{\rm YN1},0, \mathbb L_{1, G}^{\rm YN0}, 0)'|W^{(G)}, D^{(G)}] \xrightarrow{P} \mathbb{V}_1~.\]
The rest of the proof is identical to Steps 3--5 in the proof of Lemma \ref{lem:L_X} and is omitted.
\end{proof}

\begin{lemma} \label{lem:L_N_algebra}
Let $\mathbb{V}$ be defined as in Lemma \ref{lem:L_N}, and $D_{h0}$ be defined as in the proof of Theorem \ref{thm:normal_X}, then
\[D_{h0}\mathbb{V}D_{h0}' = \nu^2~,\]
where
\[\nu^2 = E[\tilde Y_g^2(1)] + E[\tilde Y_g^2(0)] - \frac{1}{2} E[(E[\tilde{Y}_g(1) + \tilde{Y}_g(0) | W_g])^2]~.\]
\end{lemma}
\begin{proof}
$\mathbb V$ in Lemma \ref{lem:L_N} is symmetric with entries
\begin{align*}
\mathbb V_{11} & = \var[\bar{Y}_g(1) N_g] - \frac{1}{2} \var[E[\bar{Y}_g(1) N_g | W_g]] \\
\mathbb V_{12} & = \cov[E[\bar{Y}_g(1)N_g|W_g],N_g] - \frac{1}{2}\cov[E[\bar{Y}_g(1)N_g|W_g],N_g] \\
\mathbb V_{13} & = \frac{1}{2} \cov[E[\bar{Y}_g(1) N_g | W_g], E[\bar{Y}_g(0) N_g | W_g]] \\
\mathbb V_{14} & = \frac{1}{2} \cov[E[\bar{Y}_g(1) N_g | W_g], N_g] \\
\mathbb V_{22} & = \var[N_g] - \frac{1}{2}\var[N_g]  \\
\mathbb V_{23} & = \frac{1}{2} \cov[N_g, E[\bar{Y}_g(0) N_g | X_g]] \\
\mathbb V_{24} & = \frac{1}{2} \var[N_g] \\
\mathbb V_{33} & = \var[\bar{Y}_g(0) N_g] - \frac{1}{2} \var[E[\bar{Y}_g(0) N_g | W_g]] \\
\mathbb V_{34} & = \cov[E[\bar{Y}_g(0) N_g | W_g], N_g] - \frac{1}{2}\cov[E[\bar{Y}_g(0) N_g | W_g], N_g] \\
\mathbb V_{44} & = \var[N_g] - \frac{1}{2}\var[N_g] ~.
\end{align*}
We proceed by mirroring the algebra in Lemma \ref{lem:L_X_algebra}. Expanding and simplifying the first half of the expression:
\begin{align*}
& \frac{\var[\bar{Y}_g(1) N_g]}{E[N_g]^2} + \frac{\var[N_g] E[\bar{Y}_g(1) N_g]^2}{E[N_g]^4} + \frac{\var[\bar{Y}_g(0) N_g]}{E[N_g]^2} + \frac{\var[N_g] E[\bar{Y}_g(0) N_g]^2}{E[N_g]^4} \\
& \hspace{3em} - \frac{2 \cov[E[\bar{Y}_g(1) N_g|W_g], N_g] E[\bar{Y}_g(1) N_g]}{E[N_g]^3} - \frac{2 \cov[E[\bar{Y}_g(0) N_g|W_g], N_g] E[\bar{Y}_g(0) N_g]}{E[N_g]^3} \\
& = \frac{E[\bar{Y}_g^2(1) N_g^2] - E[\bar{Y}_g(1) N_g]^2}{E[N_g]^2} + \frac{E[N_g^2] E[\bar{Y}_g(1) N_g]^2 - E[N_g]^2 E[\bar{Y}_g(1) N_g]^2}{E[N_g]^4} \\
& \hspace{3em} + \frac{E[\bar{Y}_g^2(0) N_g^2] - E[\bar{Y}_g(0) N_g]^2}{E[N_g]^2} + \frac{E[N_g^2] E[\bar{Y}_g(0) N_g]^2 - E[N_g]^2 E[\bar{Y}_g(0) N_g]^2}{E[N_g]^4} \\
& \hspace{3em} - \frac{2 E[\bar{Y}_g(1) N_g^2] E[\bar{Y}_g(1) N_g]}{E[N_g]^3} + \frac{2 E[\bar{Y}_g(1) N_g] E[N_g] E[\bar{Y}_g(1) N_g]}{E[N_g]^3} \\
& \hspace{3em} - \frac{2 E[\bar{Y}_g(0) N_g^2] E[\bar{Y}_g(0) N_g]}{E[N_g]^3} + \frac{2 E[\bar{Y}_g(0) N_g] E[N_g] E[\bar{Y}_g(0) N_g]}{E[N_g]^3} \\
& = \frac{E[\bar{Y}_g^2(1) N_g^2]}{E[N_g]^2} + \frac{E[\bar{Y}_g^2(0) N_g^2]}{E[N_g]^2} + \frac{E[N_g^2] E[\bar{Y}_g(1) N_g]^2}{E[N_g]^4} + \frac{E[N_g^2] E[\bar{Y}_g(0) N_g]^2}{E[N_g]^4} \\
& \hspace{3em} - \frac{2 E[\bar{Y}_g(1) N_g^2] E[\bar{Y}_g(1) N_g]}{E[N_g]^3} - \frac{2 E[\bar{Y}_g(0) N_g^2] E[\bar{Y}_g(0) N_g]}{E[N_g]^3} \\
& = E[\tilde Y_g^2(1)] + E[\tilde Y_g^2(0)]~,
\end{align*}
where 
\begin{equation*}
    \tilde Y_g(d) = \frac{N_g}{E[N_g]} \left ( \bar{Y}_g(d) - \frac{E[\bar{Y}_g(d) N_g]}{E[N_g]} \right )
\end{equation*}
for $d \in \{0, 1\}$.

Expanding the second half of the expression:
\begin{align*}
& - \frac{\var[E[\bar{Y}_g(1) N_g | W_g]]}{2E[N_g]^2} - \frac{\var[N_g] E[\bar{Y}_g(1) N_g]^2}{2 E[N_g]^4} \\
& \hspace{3em} - \frac{\var[E[\bar{Y}_g(0) N_g | W_g]]}{2E[N_g]^2} - \frac{\var[N_g] E[\bar{Y}_g(0) N_g]^2}{2 E[N_g]^4} \\
& \hspace{3em} + \frac{\cov[E[\bar{Y}_g(1) N_g | W_g], N_g] E[\bar{Y}_g(1) N_g]}{E[N_g]^3} + \frac{\cov[E[\bar{Y}_g(0) N_g | W_g], N_g] E[\bar{Y}_g(0) N_g]}{E[N_g]^3} \\
& \hspace{3em}- \frac{\cov[E[\bar{Y}_g(1) N_g | W_g], E[\bar{Y}_g(0) N_g | W_g]]}{E[N_g]^2}  + \frac{\cov[E[\bar{Y}_g(1) N_g | W_g], N_g] E[\bar{Y}_g(0) N_g]}{E[N_g] E[N_g]^2} \\
& \hspace{3em} + \frac{\cov[N_g, E[\bar{Y}_g(0) N_g | W_g]] E[\bar{Y}_g(1) N_g]}{E[N_g]^2 E[N_g]} - \frac{\cov[N_g, N_g] E[\bar{Y}_g(1) N_g] E[\bar{Y}_g(0) N_g]}{E[N_g]^2 E[N_g]^2} \\
& = - \frac{E[E[\bar{Y}_g(1) N_g | W_g]^2] - E[\bar{Y}_g(1) N_g]^2}{2E[N_g]^2} - \frac{(E[N_g^2] - E[N_g]^2) E[\bar{Y}_g(1) N_g]^2}{2 E[N_g]^4} \\
& \hspace{3em} - \frac{E[E[\bar{Y}_g(0) N_g | W_g]^2] - E[\bar{Y}_g(0) N_g]^2}{2E[N_g]^2} - \frac{(E[N_g^2] - E[N_g]^2) E[\bar{Y}_g(0) N_g]^2}{2 E[N_g]^4} \\
& \hspace{3em} + \frac{(E[E[\bar{Y}_g(1) N_g | W_g]N_g] - E[\bar{Y}_g(1) N_g] E[N_g]) E[\bar{Y}_g(1) N_g]}{E[N_g]^3} \\
& \hspace{3em} + \frac{(E[E[\bar{Y}_g(0) N_g | W_g]N_g] - E[\bar{Y}_g(0) N_g] E[N_g]) E[\bar{Y}_g(0) N_g]}{E[N_g]^3} \\
& \hspace{3em} - \frac{E[E[\bar{Y}_g(1) N_g | W_g] E[\bar{Y}_g(0) N_g | W_g]] - E[\bar{Y}_g(1) N_g] E[\bar{Y}_g(0) N_g]}{E[N_g] E[N_g]} \\
& \hspace{3em} + \frac{(E[E[\bar{Y}_g(1) N_g | W_g]N_g] - E[\bar{Y}_g(1) N_g] E[N_g]) E[\bar{Y}_g(0) N_g]}{E[N_g] E[N_g]^2} \\
& \hspace{3em} + \frac{(E[E[\bar{Y}_g(0) N_g | W_g]N_g] - E[\bar{Y}_g(0) N_g] E[N_g]) E[\bar{Y}_g(1) N_g]}{E[N_g]^2 E[N_g]} \\
& \hspace{3em} - \frac{(E[N_g^2] - E[N_g]^2) E[\bar{Y}_g(1) N_g] E[\bar{Y}_g(0) N_g]}{E[N_g]^2 E[N_g]^2} \\
& = - \frac{E[E[\bar{Y}_g(1) N_g | W_g]^2] }{2E[N_g]^2} - \frac{E[N_g^2]E[\bar{Y}_g(1)N_g]^2}{2 E[N_g]^4} - \frac{E[E[\bar{Y}_g(0) N_g | W_g]^2] }{2E[N_g]^2} - \frac{E[N_g^2]E[\bar{Y}_g(0)N_g]^2}{2 E[N_g]^4} \\
& \hspace{3em} + \frac{E[E[\bar{Y}_g(1) N_g | W_g]N_g] E[\bar{Y}_g(1) N_g]}{E[N_g]^3} + \frac{E[E[\bar{Y}_g(0) N_g | W_g]N_g] E[\bar{Y}_g(0) N_g]}{E[N_g]^3} \\
& \hspace{3em} - \frac{E[E[\bar{Y}_g(1) N_g | W_g] E[\bar{Y}_g(0) N_g | W_g]]}{E[N_g]^2} + \frac{E[E[\bar{Y}_g(1) N_g | W_g]N_g] E[\bar{Y}_g(0) N_g]}{E[N_g]^3} \\
& \hspace{3em} + \frac{E[E[\bar{Y}_g(0) N_g | W_g]N_g] E[\bar{Y}_g(1) N_g]}{E[N_g]^3} - \frac{E[N_g^2]E[\bar{Y}_g(1)N_g]E[\bar{Y}_g(0)N_g]}{E[N_g]^4} \\
& = - \frac{1}{2} E[E[\tilde{Y}_g(1) | W_g]^2] - \frac{1}{2} E[E[\tilde{Y}_g(0) | W_g]^2] - E[E[\tilde{Y}_g(1) | W_g] E[\tilde{Y}_g(0) | W_g]]\\
& = - \frac{1}{2} E[(E[\tilde{Y}_g(1) + \tilde{Y}_g(0) | W_g])^2]~,
\end{align*}
as desired.
\end{proof}

\begin{lemma}\label{lemma:mu_G}
Consider the following adjusted potential outcomes:
\begin{align*}
    \hat Y_{g}(d) = \frac{N_{g}}{ \frac{1}{2G}\sum_{1\leq j \leq 2G} N_j }\left(\bar{Y}_{g}(d)-\frac{\frac{1}{G}\sum_{1\leq j \leq 2G} \bar{Y}_{j}(d) I\{D_j = d\} N_j}{\frac{1}{G}\sum_{1\leq j \leq 2G} I\{D_j = d\} N_j}\right)~.
\end{align*}
Note the usual relationship still holds for adjusted outcomes, i.e. $\hat Y_{g} = D_g \hat Y_{g}(1) + (1-D_g)\hat Y_{g}(0)$. If Assumptions \ref{ass:QG} holds, and additionally Assumptions \ref{ass:pair_formX}--\ref{ass:DGP2} (or Assumptions \ref{ass:pair_form}--\ref{ass:DGP}) hold, then
\begin{align*}
    \hat{\mu}_{G}(d)&= \frac{1}{G} \sum_{1 \leq g \leq 2 G} \hat Y_{g}(d) I\left\{D_{g}=d\right\} \stackrel{P}{\rightarrow} 0 \\
    \hat{\sigma}_{G}^{2}(d)&=\frac{1}{G} \sum_{1<g<2 G}\left(\hat Y_{g}-\hat{\mu}_{G}(d)\right)^{2} I\left\{D_{g}=d\right\} \stackrel{P}{\rightarrow} \var \left[\tilde Y_{g}(d)\right]~.
\end{align*}
\end{lemma}
\begin{proof}
It suffices to show that
\begin{equation*}
    \frac{1}{G} \sum_{1 \leq g \leq 2 G} \hat Y_{g}^{r}(d) I\left\{D_{g}=d\right\} \stackrel{P}{\rightarrow} E\left[\tilde Y_{g}^{r}(d)\right]
\end{equation*}
for $r\in\{1,2\}$. We prove this result only for $r = 1$ and $d = 1$; the other cases can be proven similarly. To this end, write
\begin{align*}
&\frac{1}{G} \sum_{1 \leq g \leq 2 G} \hat{Y}_g(1) I\left\{D_g=1\right\} =\frac{1}{G} \sum_{1 \leq g \leq 2 G} \hat{Y}_g(1) D_g = \frac{1}{G} \sum_{1 \leq g \leq 2 G} \tilde{Y}_g(1) D_g +  \frac{1}{G} \sum_{1 \leq g \leq 2 G} \left(\hat Y_g(1) - \tilde Y_g(1)\right) D_g~.
\end{align*}
Note that
\begin{align*}
    &\frac{1}{G} \sum_{1 \leq g \leq 2 G} \left(\hat Y_g(1) - \tilde Y_g(1)\right) D_g = \left( \frac{1}{\frac{1}{2G} \sum_{1\leq g \leq 2G} N_g} - \frac{1}{E[N_g]} \right)\left( \frac{1}{G} \sum_{1 \leq g \leq 2 G} \bar Y_g(1) N_g D_g\right) \\
    &\quad - \left( \frac{\frac{1}{G}\sum_{1\leq g \leq 2G} \bar Y_g(d) I\{D_g=d\} N_g}{\left(\frac{1}{2G} \sum_{1\leq g \leq 2G} N_g \right)^2} - \frac{E[\Bar{Y}_g(d) N_g]}{E[N_g]^2} \right)\left( \frac{1}{G} \sum_{1 \leq g \leq 2 G} N_g D_g \right)
\end{align*}
By the weak law of large numbers, Lemma \ref{lem:wlln} and Slutsky's theorem, we have 
\begin{equation*}
    \frac{1}{G} \sum_{1 \leq g \leq 2 G} \left(\hat Y_g(1) - \tilde Y_g(1)\right) D_g \xrightarrow{P} 0~.
\end{equation*}
Lemma \ref{lem:wlln} implies
\begin{equation*}
    \frac{1}{G} \sum_{1 \leq g \leq 2 G} \tilde{Y}_g(d) D_g \stackrel{P}{\rightarrow} E\left[\tilde Y_{g}(d)\right] = 0~.
\end{equation*}
Thus, the result follows. 
\end{proof}

\begin{lemma}\label{lemma:tau_G}
If Assumptions \ref{ass:QG} holds, and Assumptions \ref{ass:pair_formX}-\ref{ass:DGP2} hold, then
\begin{equation*}
    \hat{\tau}_{G}^{2} \stackrel{P}{\rightarrow} E\left[\operatorname{Var}\left[\tilde Y_{g}(1) \middle| X_{g}\right]\right]+E\left[\operatorname{Var}\left[\tilde Y_{g}(0) \middle| X_{g}\right]\right]+E\left[\left(E\left[\tilde Y_{g}(1) \middle| X_{g}\right]-E\left[\tilde Y_{g}(0) \middle| X_{g}\right]\right)^{2}\right]
\end{equation*}
in the case where we match on cluster size. Instead, if Assumptions \ref{ass:QG} and \ref{ass:pair_form}-\ref{ass:DGP} hold, then
\begin{equation*}
    \hat{\tau}_{G}^{2} \stackrel{P}{\rightarrow} E\left[\operatorname{Var}\left[\tilde Y_{g}(1) \middle| W_{g}\right]\right]+E\left[\operatorname{Var}\left[\tilde Y_{g}(0) \middle| W_{g}\right]\right]+E\left[\left(E\left[\tilde Y_{g}(1) \middle| W_{g}\right]-E\left[\tilde Y_{g}(0) \middle| W_{g}\right]\right)^{2}\right]
\end{equation*}
in the case where we do not match on cluster size.
\end{lemma}
\begin{proof}
Note that
\begin{equation*}
    \hat{\tau}_{G}^{2}=\frac{1}{G} \sum_{1 \leq j \leq G}\left(\hat Y_{\pi(2 j)}-\hat Y_{\pi(2 j-1)}\right)^{2}=\frac{1}{G} \sum_{1 \leq g \leq 2 G} \hat{Y}_{g}^{2}-\frac{2}{G} \sum_{1 \leq j \leq G} \hat Y_{\pi(2 j)}\hat Y_{\pi(2 j-1)} .
\end{equation*}
Since
\begin{equation*}
    \frac{1}{G} \sum_{1 \leq g \leq 2 G} \hat Y_{g}^{2}=\hat{\sigma}_{G}^{2}(1)-\hat{\mu}_{G}^{2}(1)+\hat{\sigma}_{G}^{2}(0)-\hat{\mu}_{G}^{2}(0)
\end{equation*}
It follows from Lemma \ref{lemma:mu_G} that
\begin{equation*}
    \frac{1}{G} \sum_{1 \leq g \leq 2 G} \hat Y_{g}^{2} \xrightarrow{P} E[\tilde Y_g^2(1)] + E[\tilde Y_g^2(0)]
\end{equation*}
Next, we argue that
\begin{equation*}
    \frac{2}{G} \sum_{1 \leq j \leq G} \hat Y_{\pi(2 j)}\hat Y_{\pi(2 j-1)} \xrightarrow{P}  2 E[\mu_1(W_g) \mu_0(W_g) ]~,
\end{equation*}
where we use the notation $\mu_d(W_g)$ to denote $E[\tilde Y_g(d) | W_g]$. To this end, first note that
\begin{equation*}
    \frac{2}{G} \sum_{1 \leq j \leq G} \hat Y_{\pi(2 j)}\hat Y_{\pi(2 j-1)} = \frac{2}{G} \sum_{1 \leq j \leq G} \tilde Y_{\pi(2 j)}\tilde Y_{\pi(2 j-1)} + \frac{2}{G} \sum_{1 \leq j \leq G} \left ( \hat Y_{\pi(2 j)}\hat Y_{\pi(2 j-1)} - \tilde Y_{\pi(2 j)}\tilde Y_{\pi(2 j-1)} \right )~.
\end{equation*}
Note that
\begin{align*}
    &\frac{2}{G} \sum_{1 \leq j \leq G} \left( \hat Y_{\pi(2 j)}(1)\hat Y_{\pi(2 j-1)}(0) - \tilde Y_{\pi(2 j)}(1)\tilde Y_{\pi(2 j-1)}(0)\right) D_{\pi(2 j)}\\
    &= \frac{2}{G} \sum_{1 \leq j \leq G} \left ( \left(\hat Y_{\pi(2 j)}(1) - \tilde Y_{\pi(2 j)}(1)\right)\hat Y_{\pi(2 j-1)}(0) D_{\pi(2 j)} + \left(\hat Y_{\pi(2 j-1)}(0) -\tilde Y_{\pi(2 j-1)}(0) \right)\tilde Y_{\pi(2 j)}(1) D_{\pi(2 j)} \right ) \\
    &= \frac{2}{G} \sum_{1 \leq j \leq G} \bigg ( \left(\hat Y_{\pi(2 j)}(1)  - \tilde Y_{\pi(2 j)}(1)\right)\tilde Y_{\pi(2 j-1)}(0) D_{\pi(2 j)} \\
    & \hspace{3em} + \left(\hat Y_{\pi(2 j)}(1) - \tilde Y_{\pi(2 j)}(1)\right)\left(\hat Y_{\pi(2 j-1)}(0)- \tilde Y_{\pi(2 j-1)}(0)\right) D_{\pi(2 j)} \\
    & \hspace{3em} + \left(\hat Y_{\pi(2 j-1)}(0) -\tilde Y_{\pi(2 j-1)}(0) \right)\tilde Y_{\pi(2 j)}(1) D_{\pi(2 j)} \bigg )~,
\end{align*}
for which the first term is given as follows:
\begin{align*}
    &\frac{2}{G} \sum_{1 \leq j \leq G} \left(\hat Y_{\pi(2 j)}(1) - \tilde Y_{\pi(2 j)}(1)\right)\tilde Y_{\pi(2 j-1)}(0) D_{\pi(2 j)} \\
    &= \left( \frac{1}{\frac{1}{2G} \sum_{1\leq g \leq 2G} N_g} - \frac{1}{E[N_g]} \right)\left(\frac{2}{G} \sum_{1 \leq j \leq G} N_{\pi(2 j)} \bar Y_{\pi(2 j)}(1) \tilde Y_{\pi(2 j-1)}(0)  D_{\pi(2 j)} \right)\\
    &\quad - \left( \frac{\frac{1}{2G}\sum_{1\leq g \leq 2G} \bar Y_g(1) I\{D_g=1\} N_g}{\left(\frac{1}{2G} \sum_{1\leq g \leq 2G} N_g \right)^2} - \frac{E[\Bar{Y}_g(1) N_g]}{E[N_g]^2} \right) \left(\frac{2}{G} \sum_{1 \leq j \leq G} N_{\pi(2 j)} \tilde Y_{\pi(2 j-1)}(0)D_{\pi(2 j)}\right)~.
\end{align*}
Lemma \ref{lem:cross} implies
\begin{align*}
    \frac{2}{G} \sum_{1 \leq j \leq G} N_{\pi(2 j)} \bar Y_{\pi(2 j)}(1) \tilde Y_{\pi(2 j-1)}(0) D_{\pi(2 j)} &\xrightarrow{P} E[E[N_g \Bar Y_g(1)| X_g] E[ \tilde Y_g(0)| X_g]] \\
    \nonumber \frac{2}{G} \sum_{1 \leq j \leq G} N_{\pi(2 j)} \tilde Y_{\pi(2 j-1)}(0) D_{\pi(2 j)} &\xrightarrow{P} E[E[N_g| X_g]E[\tilde Y_g(0)| X_g]]
\end{align*}
for the case of not matching on cluster sizes. For the case where we match on cluster sizes, 
\begin{align*}
    \frac{2}{G} \sum_{1 \leq j \leq G} N_{\pi(2 j)} \bar Y_{\pi(2 j)}(1) \tilde Y_{\pi(2 j-1)}(0) D_{\pi(2 j)} &\xrightarrow{P} E[N_g E[ \Bar Y_g(1)| W_g] E[ \tilde Y_g(0)| W_g]] \\
    \frac{2}{G} \sum_{1 \leq j \leq G} N_{\pi(2 j)} \tilde Y_{\pi(2 j-1)}(0) D_{\pi(2 j)} &\xrightarrow{P} E[N_g E[\tilde Y_g(0)| W_g]]
\end{align*}
Then, by the weak law of large numbers, Lemma \ref{lem:wlln}, and the continuous mapping theorem, we have 
\begin{align*}
    \frac{2}{G} \sum_{1 \leq j \leq G} \left(\hat Y_{\pi(2 j)}(1) - \tilde Y_{\pi(2 j)}(1)\right)\tilde Y_{\pi(2 j-1)}(0) D_{\pi(2 j)} \xrightarrow{P} 0~.
\end{align*}
By repeating the same arguments for the other two terms, we conclude that
\begin{equation*}
    \frac{2}{G} \sum_{1 \leq j \leq G}\left( \hat Y_{\pi(2 j)}(1)\hat Y_{\pi(2 j-1)}(0) - \tilde Y_{\pi(2 j)}(1)\tilde Y_{\pi(2 j-1)}(0)\right)D_{\pi(2 j)} \xrightarrow{P} 0~,
\end{equation*}
which immediately implies
\begin{equation*}
    \frac{2}{G} \sum_{1 \leq j \leq G} \hat Y_{\pi(2 j)}\hat Y_{\pi(2 j-1)} - \tilde Y_{\pi(2 j)}\tilde Y_{\pi(2 j-1)} \xrightarrow{P} 0 ~.
\end{equation*}
Thus, it is left to show that
\begin{equation*}
    \frac{2}{G} \sum_{1 \leq j \leq G} \tilde Y_{\pi(2 j)}\tilde Y_{\pi(2 j-1)} \xrightarrow{P}  2 E[\mu_1(W_g) \mu_0(W_g) ]~,
\end{equation*}
for the case of matching on cluster sizes, and for the case of not matching on cluster size,
\begin{equation*}
    \frac{2}{G} \sum_{1 \leq j \leq G} \tilde Y_{\pi(2 j)}\tilde Y_{\pi(2 j-1)} \xrightarrow{P}  2 E[\mu_1(X_g) \mu_0(X_g) ]~,
\end{equation*}
both of which follow from Lemmas \ref{lem:cross} and \ref{lemma:assumptions-for-bai-inference}. Hence, in the case where we match on cluster size,
\begin{align*}
\hat{\tau}_{n}^{2} & \stackrel{P}{\rightarrow} E\left[\tilde{Y}_{g}^{2}(1)\right]+E\left[\tilde Y_{g}^{2}(0)\right]-2 E\left[\mu_{1}\left(W_{g}\right) \mu_{0}\left(W_{g}\right)\right] \\
&=E\left[\operatorname{Var}\left[\tilde Y_{g}(1) | W_{g}\right]\right]+E\left[\operatorname{Var}\left[\tilde Y_{g}(0) | W_{g}\right]\right]+E\left[\left(\mu_{1}\left(W_{g}\right)-\mu_{0}\left(W_{g}\right)\right)^{2}\right] \\
&=E\left[\operatorname{Var}\left[\tilde Y_{g}(1) | W_{g}\right]\right]+E\left[\operatorname{Var}\left[\tilde Y_{g}(0) | W_{g}\right]\right]+E\left[\left(E\left[\tilde Y_{g}(1) | X_{i}\right]-E\left[\tilde Y_{g}(0) | W_{g}\right]\right)^{2}\right]~.
\end{align*}
And the corresponding result holds in the case where we do not match on cluster size.
\end{proof}

\begin{lemma}\label{lemma:lambda_G}
If Assumptions \ref{ass:QG} holds, and Assumptions \ref{ass:pair_formX}-\ref{ass:DGP2}, \ref{ass:close-pairs-of-pairsX} hold, then
\begin{equation*}
    \hat{\lambda}_{G}^{2} \stackrel{P}{\rightarrow} E\left[\left(E\left[\tilde Y_{g}(1) | X_{g}\right]-E\left[\tilde Y_{g}(0) | X_{g}\right]\right)^{2}\right]
\end{equation*}
in the case where we do not match on cluster size. Instead, if Assumptions \ref{ass:pair_form}-\ref{ass:DGP}, \ref{ass:close-pairs-of-pairs} hold, then
\begin{equation*}
    \hat{\lambda}_{G}^{2} \stackrel{P}{\rightarrow} E\left[\left(E\left[\tilde Y_{g}(1) | W_{g}\right]-E\left[\tilde Y_{g}(0) | W_{g}\right]\right)^{2}\right]
\end{equation*}
in the case where we match on cluster size.
\end{lemma}
\begin{proof}
Note that
\begin{align*}
    \hat{\lambda}_{G}^{2}=& \frac{2}{G} \sum_{1 \leq j \leq\left\lfloor G/2\right\rfloor}\left(\left(\hat Y_{\pi(4 j-3)}-\hat Y_{\pi(4 j-2)}\right)\left(\hat Y_{\pi(4 j-1)}- \hat Y_{\pi(4 j)}\right)\left(D_{\pi(4 j-3)}-D_{\pi(4 j-2)}\right)\left(D_{\pi(4 j-1)}-D_{\pi(4 j)}\right)\right) \\
    =& \underbrace{ \frac{2}{G} \sum_{1 \leq j \leq\left\lfloor G/2\right\rfloor}\left(\left(\tilde Y_{\pi(4 j-3)}-\tilde Y_{\pi(4 j-2)}\right)\left(\tilde Y_{\pi(4 j-1)}- \tilde Y_{\pi(4 j)}\right)\left(D_{\pi(4 j-3)}-D_{\pi(4 j-2)}\right)\left(D_{\pi(4 j-1)}-D_{\pi(4 j)}\right)\right)}_{:= \tilde{\lambda}_{G}^{2}} \\
    &+ \frac{2}{G} \sum_{1 \leq j \leq\left\lfloor G/2\right\rfloor}\left(\left(\left(\hat Y_{\pi(4 j-3)}-\hat Y_{\pi(4 j-2)}\right)\left(\hat Y_{\pi(4 j-1)}- \hat Y_{\pi(4 j)}\right) - \left(\tilde Y_{\pi(4 j-3)}-\tilde Y_{\pi(4 j-2)}\right)\left(\tilde Y_{\pi(4 j-1)}- \tilde Y_{\pi(4 j)}\right) \right)\right.\\
    &\quad \left.\times\left(D_{\pi(4 j-3)}-D_{\pi(4 j-2)}\right)\left(D_{\pi(4 j-1)}-D_{\pi(4 j)}\right)\right)
\end{align*}
Note that
\begin{align*}
    &\left(\hat Y_{\pi(4 j-3)}(1)-\hat Y_{\pi(4 j-2)}(0)\right)\left(\hat Y_{\pi(4 j-1)}(1)- \hat Y_{\pi(4 j)}(0)\right)D_{\pi(4 j-3)}D_{\pi(4 j-1)}\\
    &\quad - \left(\tilde Y_{\pi(4 j-3)}(1)-\tilde Y_{\pi(4 j-2)}(0)\right)\left(\tilde Y_{\pi(4 j-1)}(1)- \tilde Y_{\pi(4 j)}(0)\right) D_{\pi(4 j-3)}D_{\pi(4 j-1)}\\
    &=\left(\hat Y_{\pi(4 j-3)}(1)-\hat Y_{\pi(4 j-2)}(0)-\left(\tilde Y_{\pi(4 j-3)}(1)-\tilde Y_{\pi(4 j-2)}(0)\right)\right)\left(\tilde Y_{\pi(4 j-1)}(1)- \tilde Y_{\pi(4 j)}(0)\right)D_{\pi(4 j-3)}D_{\pi(4 j-1)} \\
    & + \left(\hat Y_{\pi(4 j-3)}(1)-\hat Y_{\pi(4 j-2)}(0)-\left(\tilde Y_{\pi(4 j-3)}(1)-\tilde Y_{\pi(4 j-2)}(0)\right)\right)\\
    &\hspace{3em} \times\left(\hat Y_{\pi(4 j-1)}(1)- \hat Y_{\pi(4 j)}(0)-\left(\tilde Y_{\pi(4 j-1)}(1)- \tilde Y_{\pi(4 j)}(0)\right)\right)D_{\pi(4 j-3)}D_{\pi(4 j-1)}\\
    &+ \left(\hat Y_{\pi(4 j-1)}(1)- \hat Y_{\pi(4 j)}(0) - \left(\tilde Y_{\pi(4 j-1)}(1)- \tilde Y_{\pi(4 j)}(0)\right)\right)\left(\tilde Y_{\pi(4 j-3)}(1)-\tilde Y_{\pi(4 j-2)}(0)\right)D_{\pi(4 j-3)}D_{\pi(4 j-1)}~.
\end{align*}
Then we can show that each term converges to zero in probability by repeating the arguments in Lemma \ref{lemma:tau_G}. Similar arguments imply the same result holds for other cross products, which implies $\hat{\lambda}_{G}^{2} - \tilde{\lambda}_{G}^{2} \xrightarrow{P} 0$. Finally, by Lemma S.1.7 of \cite{bai2022inference} and Lemma \ref{lemma:assumptions-for-bai-inference}, we have 
\begin{equation*}
    \hat{\lambda}_{G}^{2} = \tilde{\lambda}_{G}^{2} + o_P(1) \stackrel{P}{\rightarrow} E\left[\left(E\left[\tilde Y_{g}(1) \middle| W_{g}\right]-E\left[\tilde Y_{g}(0) \middle| W_{g}\right]\right)^{2}\right] 
\end{equation*}
in the case where we match on cluster size, and
\begin{equation*}
    \hat{\lambda}_{G}^{2} = \tilde{\lambda}_{G}^{2} + o_P(1) \stackrel{P}{\rightarrow} E\left[\left(E\left[\tilde Y_{g}(1) \middle| X_{g}\right]-E\left[\tilde Y_{g}(0) \middle| X_{g}\right]\right)^{2}\right] 
\end{equation*}
in the case where we do not match on cluster size.
\end{proof}

\begin{lemma}\label{lem:rand_numerator}
Let $\tilde{R}_G(t)$ denote the randomization distribution of $\sqrt{G}\hat{\Delta}_G$ (see equation \eqref{eq:rand_num}). Then under the null hypothesis \eqref{eq:delta_zero}, we have that 
\[\sup_{t \in \mathbf{R}}|\tilde{R}_G(t) - \Phi(t/\tau)|\xrightarrow{P} 0~,\]
where, in the case where we match on cluster size,
\[\tau^2 = E[\var[\tilde{Y}_g(1)|W_g]] + E[\var[\tilde{Y}_g(0)|W_g]] + E\left[(E[\tilde{Y}_g(1)|W_g] - E[\tilde{Y}_g(0)|X_g])^2\right]~,\]
and in the case where we do not match on cluster size,
\[\tau^2 = E[\var[\tilde{Y}_g(1)|X_g]] + E[\var[\tilde{Y}_g(0)|X_g]] + E\left[(E[\tilde{Y}_g(1)|X_g] - E[\tilde{Y}_g(0)|X_g])^2\right]~.\]
\end{lemma}
\begin{proof}
For a random transformation of the data, it follows as a consequence of Lemma \ref{lem:wlln} that
\begin{align*}
    \frac{1}{G}\sum_{1 \le g \le 2G}I\{\tilde{D}_g = d\}N_g & \xrightarrow{P} E[N_g]~, \\
    \frac{1}{G}\sum_{1 \le g \le {2G}}(1 - \tilde{D}_g)N_g\bar{Y}_g & \xrightarrow{P} E[N_g\bar{Y}_g(0)]~.
\end{align*}
Combining this with Lemma \ref{lem:rand_joint} and a straightforward modification of Lemma A.3. in \cite{chung2013exact} to two dimensional distributions, we obtain that 
\[\sup_{t \in \mathbf{R}}|\tilde{R}_G(t) - \Phi(t/\tau)|\xrightarrow{P} 0~,\]
where when we match on cluster size
\[\tau^2 = \frac{1}{E[N_g]^2}\left(E[\var(N_g\bar{Y}_g(1)|W_g)] + E[\var(N_g\bar{Y}_g(0)|W_g)] + E\left[(E[N_g\bar{Y}_g(1)|W_g] - E[N_g\bar{Y}_g(0)|W_g])^2\right]\right)~,\]
and when we do \emph{not} match on cluster size
\begin{align*}
\tau^2 &= \frac{1}{E[N_g]^2}\Big(E[\var(N_g\bar{Y}_g(1)|X_g)] + E[\var(N_g\bar{Y}_g(0)|X_g)] + E\left[(E[N_g\bar{Y}_g(1)|X_g] - E[N_g\bar{Y}_g(0)|X_g])^2\right] + \\
& - 2\frac{E[N_g\bar{Y}_g(0)]}{E[N_g]}\left(E[N^2_g\bar{Y}_g(1)] + E[N^2_g\bar{Y}_g(0)] \right. \\
& \left.- \left(E\left[E[N_g\bar{Y}_g(1)|X_g]E[N_g|X_g]\right] + E\left[E[N_g\bar{Y}_g(0)|X_g]E[N_g|X_g]\right]\right)\right) + \left(\frac{E[N_g\bar{Y}_g(0)]}{E[N_g]}\right)^22E[\var(N_g|X_g)]\Big)~.
\end{align*}
Note than, since under the null, $E[N_g \bar Y_g(1)] = E[N_g \bar Y_g(0)]$, we obtain
\begin{align*}
   & E[\var[\tilde Y_g(1) | X_g]] + E[\var[\tilde Y_g(0) | X_g]] + E[(E[\tilde Y_g(1) | X_g] - E[\tilde Y_g(0) | X_g])^2] \\
   & = \frac{E[\var[N_g \bar Y_g(1) | X_g]]}{E[N_g]^2} + \frac{E[\var[N_g \bar Y_g(0) | X_g]]}{E[N_g]^2} + \frac{2 E[\var[N_g | X_g]] E[N_g \bar Y_g(d)]^2}{E[N_g]^4} \\
   & \hspace{3em} + \frac{E[(E[N_g \bar Y_g(1) | X_g] - E[N_g \bar Y_g(0) | X_g])^2}{E[N_g]^2} \\
   & \hspace{3em} - 2 \frac{E[N_g \bar Y_g(1)] (E[N_g^2 \bar Y_g(1)] - E[E[N_g \bar Y_g(1) | X_g] E[N_g | X_g]])}{E[N_g]^3} \\
   & \hspace{3em} - 2 \frac{E[N_g \bar Y_g(0)] (E[N_g^2 \bar Y_g(0)] - E[E[N_g \bar Y_g(0) | X_g] E[N_g | X_g]])}{E[N_g]^3}~.
\end{align*}
The result then follows immediately.
\end{proof}
\begin{lemma}\label{lem:rand_denom}
Let $\check{v}^2_G(\epsilon_1,\ldots, \epsilon_G)$ be defined as in equation \eqref{eq:check_v}. If Assumption \ref{ass:QG} holds, and Assumptions \ref{ass:DGP}-\ref{ass:pair_form} (or Assumptions \ref{ass:DGP2}-\ref{ass:pair_formX}) hold,
\begin{equation*}
    \check{v}_G^2(\epsilon_1,\dots, \epsilon_G) \xrightarrow{P} \tau^2~,
\end{equation*}
where $\tau^2$ is defined in (\ref{lem:rand_numerator}).
\end{lemma}
\begin{proof}
From Lemma \ref{lemma:tau_G}, we see that $\hat{\tau}_G^2 \xrightarrow{P} \tau^2$. It therefore suffices to show that $\check{\lambda}_G^2\left(\epsilon_1, \ldots, \epsilon_G\right) \xrightarrow{P} 0$. In order to do so, note that $\check{\lambda}_G^2\left(\epsilon_1, \ldots, \epsilon_G\right)$ may be decomposed into sums of the form
\begin{equation*}
    \frac{2}{G} \sum_{1 \leq j \leq\left\lfloor\frac{G}{2}\right\rfloor} \epsilon_{2 j-1} \epsilon_{2 j} \hat Y_{\pi(4 j-k)} \hat Y_{\pi(4 j-\ell)} D_{\pi\left(4 j-k^{\prime}\right)} D_{\pi\left(4 j-\ell^{\prime}\right)}~,
\end{equation*}
where $(k,k^\prime) \in \{2,3\}^2$ and $(l,l^\prime)\in \{0,1\}^2$. Note that
\begin{align*}
    &\frac{2}{G} \sum_{1 \leq j \leq\left\lfloor\frac{G}{2}\right\rfloor} \epsilon_{2 j-1} \epsilon_{2 j} \hat Y_{\pi(4 j-k)} \hat Y_{\pi(4 j-\ell)} D_{\pi\left(4 j-k^{\prime}\right)} D_{\pi\left(4 j-\ell^{\prime}\right)}\\
    &= \frac{2}{G} \sum_{1 \leq j \leq\left\lfloor\frac{G}{2}\right\rfloor} \epsilon_{2 j-1} \epsilon_{2 j} \tilde Y_{\pi(4 j-k)} \tilde Y_{\pi(4 j-\ell)} D_{\pi\left(4 j-k^{\prime}\right)} D_{\pi\left(4 j-\ell^{\prime}\right)}\\
    &\quad + \frac{G}{n} \sum_{1 \leq j \leq\left\lfloor\frac{G}{2}\right\rfloor} \epsilon_{2 j-1} \epsilon_{2 j} \left(\hat Y_{\pi(4 j-k)} \hat Y_{\pi(4 j-\ell)}-\tilde Y_{\pi(4 j-k)} \tilde Y_{\pi(4 j-\ell)}\right)  D_{\pi\left(4 j-k^{\prime}\right)} D_{\pi\left(4 j-\ell^{\prime}\right)}~.
\end{align*}
By following the arguments in Lemma S.1.9 of \cite{bai2022inference} and Lemma \ref{lemma:assumptions-for-bai-inference}, we have that
\begin{equation*}
    \frac{2}{G} \sum_{1 \leq j \leq\left\lfloor\frac{G}{2}\right\rfloor} \epsilon_{2 j-1} \epsilon_{2 j} \tilde Y_{\pi(4 j-k)} \tilde Y_{\pi(4 j-\ell)} D_{\pi\left(4 j-k^{\prime}\right)} D_{\pi\left(4 j-\ell^{\prime}\right)} \xrightarrow{P} 0~.
\end{equation*}
As for the second term, we show that it convergences to zero in probability in the case where $k=k^\prime=3$ and $\ell=\ell^\prime=1$. And the other cases should hold by repeating the same arguments.
\begin{align*}
    &\frac{2}{G}  \sum_{1 \leq j \leq\left\lfloor\frac{G}{2}\right\rfloor} \epsilon_{2 j-1} \epsilon_{2 j} \left(\hat Y_{\pi(4 j-3)} \hat Y_{\pi(4 j-1)}-\tilde Y_{\pi(4 j-3)} \tilde Y_{\pi(4 j-1)}\right)  D_{\pi\left(4 j-3\right)} D_{\pi\left(4 j-1^{\prime}\right)}\\
    &=\frac{2}{G}  \sum_{1 \leq j \leq\left\lfloor\frac{G}{2}\right\rfloor} \epsilon_{2 j-1} \epsilon_{2 j} \left(\hat Y_{\pi(4 j-3)}(1) \hat Y_{\pi(4 j-1)}(1)-\tilde Y_{\pi(4 j-3)}(1) \tilde Y_{\pi(4 j-1)}(1)\right)  D_{\pi\left(4 j-3\right)} D_{\pi\left(4 j-1^{\prime}\right)}\\
    &= \frac{2}{G}  \sum_{1 \leq j \leq\left\lfloor\frac{G}{2}\right\rfloor} \epsilon_{2 j-1} \epsilon_{2 j} \left(\hat Y_{\pi(4 j-3)}(1)-\tilde Y_{\pi(4 j-3)}(1)\right)  \tilde Y_{\pi(4 j-1)}(1)  D_{\pi\left(4 j-3\right)} D_{\pi\left(4 j-1^{\prime}\right)} \\
    &\quad + \frac{2}{G}  \sum_{1 \leq j \leq\left\lfloor\frac{G}{2}\right\rfloor} \epsilon_{2 j-1} \epsilon_{2 j} \left(\hat Y_{\pi(4 j-3)}(1)-\tilde Y_{\pi(4 j-3)}(1)\right) \left(\hat Y_{\pi(4 j-1)}(1) - \tilde Y_{\pi(4 j-1)}(1)\right)  D_{\pi\left(4 j-3\right)} D_{\pi\left(4 j-1^{\prime}\right)}\\
    &\quad + \frac{2}{G}  \sum_{1 \leq j \leq\left\lfloor\frac{G}{2}\right\rfloor} \epsilon_{2 j-1} \epsilon_{2 j} \left(\hat Y_{\pi(4 j-1)}(1) - \tilde Y_{\pi(4 j-1)}(1)\right) \tilde Y_{\pi(4 j-3)}(1)  D_{\pi\left(4 j-3\right)} D_{\pi\left(4 j-1^{\prime}\right)}~,
\end{align*}
for which the first term is given as follows:
\begin{align*}
    &\frac{2}{G}  \sum_{1 \leq j \leq\left\lfloor\frac{G}{2}\right\rfloor} \epsilon_{2 j-1} \epsilon_{2 j} \left(\hat Y_{\pi(4 j-3)}(1)-\tilde Y_{\pi(4 j-3)}(1)\right)  \tilde Y_{\pi(4 j-1)}(1)  D_{\pi\left(4 j-3\right)} D_{\pi\left(4 j-1^{\prime}\right)} \\
    &= \left( \frac{1}{\frac{1}{2G} \sum_{1\leq g \leq 2G} N_g} - \frac{1}{E[N_g]} \right)\left(\frac{2}{G} \sum_{1 \leq j \leq \left\lfloor\frac{G}{2}\right\rfloor} \epsilon_{2 j-1} \epsilon_{2 j} N_{\pi(4 j-3)} \bar Y_{\pi(4 j-3)}(1) \right. \\
    & \hspace{3em} \left. \times \tilde Y_{\pi(4 j-1)}(1)  D_{\pi\left(4 j-3\right)} D_{\pi\left(4 j-1^{\prime}\right)} \right)\\
    &\quad - \left( \frac{\frac{1}{2G}\sum_{1\leq g \leq 2G} \bar Y_g(1) I\{D_g=1\} N_g}{\left(\frac{1}{2G} \sum_{1\leq g \leq 2G} N_g \right)^2} - \frac{E[\Bar{Y}_g(1) N_g]}{E[N_g]^2} \right) \left(\frac{2}{G} \sum_{1 \leq j \leq \left\lfloor\frac{G}{2}\right\rfloor} \epsilon_{2 j-1} \epsilon_{2 j} N_{\pi(4 j-3)} \right. \\
    & \hspace{3em} \left. \times \tilde Y_{\pi(4 j-1)}(1)  D_{\pi\left(4 j-3\right)} D_{\pi\left(4 j-1^{\prime}\right)}\right)~.
\end{align*}
by following the same argument in Lemma S.1.7 from \cite{bai2022inference} and Lemma \ref{lemma:assumptions-for-bai-inference}, we have
\begin{align*}
    &\frac{2}{G} \sum_{1 \leq j \leq \left\lfloor\frac{G}{2}\right\rfloor} \epsilon_{2 j-1} \epsilon_{2 j} N_{\pi(4 j-3)} \bar Y_{\pi(4 j-3)}(1) \tilde Y_{\pi(4 j-1)}(1)  D_{\pi\left(4 j-3\right)} D_{\pi\left(4 j-1^{\prime}\right)} \xrightarrow{P} 0 \\
    &\frac{2}{G} \sum_{1 \leq j \leq \left\lfloor\frac{G}{2}\right\rfloor} \epsilon_{2 j-1} \epsilon_{2 j} N_{\pi(4 j-3)} \tilde Y_{\pi(4 j-1)}(1)  D_{\pi\left(4 j-3\right)} D_{\pi\left(4 j-1^{\prime}\right)} \xrightarrow{P} 0 ~.
\end{align*}
Then, by the weak law of large numbers, Lemma \ref{lem:wlln} and the continuous mapping theorem, we have 
\begin{equation*}
    \frac{2}{G}  \sum_{1 \leq j \leq\left\lfloor\frac{G}{2}\right\rfloor} \epsilon_{2 j-1} \epsilon_{2 j} \left(\hat Y_{\pi(4 j-3)}(1)-\tilde Y_{\pi(4 j-3)}(1)\right)  \tilde Y_{\pi(4 j-1)}(1)  D_{\pi\left(4 j-3\right)} D_{\pi\left(4 j-1^{\prime}\right)} \xrightarrow{P} 0~.
\end{equation*}
By repeating the same arguments for the other two terms, we conclude that
\begin{equation*}
    \frac{2}{G}  \sum_{1 \leq j \leq\left\lfloor\frac{G}{2}\right\rfloor} \epsilon_{2 j-1} \epsilon_{2 j} \left(\hat Y_{\pi(4 j-3)} \hat Y_{\pi(4 j-1)}-\tilde Y_{\pi(4 j-3)} \tilde Y_{\pi(4 j-1)}\right)  D_{\pi\left(4 j-3\right)} D_{\pi\left(4 j-1^{\prime}\right)} \xrightarrow{P} 0~.
\end{equation*}
Therefore, for $(k,k^\prime) \in \{2,3\}^2$ and $(l,l^\prime)\in \{0,1\}^2$,
\begin{equation*}
    \frac{2}{G} \sum_{1 \leq j \leq\left\lfloor\frac{G}{2}\right\rfloor} \epsilon_{2 j-1} \epsilon_{2 j} \hat Y_{\pi(4 j-k)} \hat Y_{\pi(4 j-\ell)} D_{\pi\left(4 j-k^{\prime}\right)} D_{\pi\left(4 j-\ell^{\prime}\right)} \xrightarrow{P} 0 ~,
\end{equation*}
which implies $\check{\lambda}_G^2\left(\epsilon_1, \ldots, \epsilon_G\right) \xrightarrow{P} 0$, and thus $\check{\nu}_G^2(\epsilon_1,\dots, \epsilon_G) \xrightarrow{P} \tau^2$.
\end{proof}

\begin{lemma} \label{lem:betatilde}
    Suppose all assumptions in Theorem \ref{thm:adj} hold. Then,
    \begin{align*}
        \frac{1}{G} \sum_{1\leq j \leq G} (\hat\psi_{1,j} - \hat\psi_{0,j}) (\hat\psi_{1,j} - \hat\psi_{0,j})^{\prime} & \xrightarrow{P} 2 E[\psi_g \psi_g^\prime] - 2 E[E[\psi_g| W_g][\psi_g| W_g]'] = 2 E[\var[\psi_g | W_g]] \\
        \frac{1}{G} \sum_{1\leq j \leq G} (\hat\psi_{1,j} - \hat\psi_{0,j}) (\tilde\mu_{1,j} - \tilde\mu_{0,j}) & \xrightarrow{P} E\left[  \cov \left[ \tilde{Y}_g(1) + \tilde{Y}_g(0),  \psi_g \middle | W_g\right] \right] E[N_g]
    \end{align*}
\end{lemma}

\begin{proof}
     Note that 
\begin{align*}
    &\frac{1}{G} \sum_{1\leq j \leq G} (\hat\psi_{1,j} - \hat\psi_{0,j}) (\hat\psi_{1,j} - \hat\psi_{0,j})^{\prime} \\
    &= \frac{1}{G} \sum_{1\leq j \leq G}  \hat\psi_{1,j} \hat\psi_{1,j}^\prime + \hat\psi_{0,j} \hat\psi_{0,j}^\prime - \hat\psi_{1,j} \hat\psi_{0,j}^\prime - \hat\psi_{0,j} \hat\psi_{1,j}^\prime \\
    &= \frac{1}{G} \sum_{1\leq g \leq 2 G} \psi_g \psi_g^\prime D_g +  \frac{1}{G} \sum_{1\leq g \leq 2 G} \psi_g \psi_g^\prime (1-D_g) \\
    &\hspace{1em} -  \frac{1}{G} \sum_{1\leq j \leq G} \psi_{\pi(2j)} \psi_{\pi(2j-1)}^\prime  D_{\pi(2j)}  -  \frac{1}{G} \sum_{1\leq j \leq G}\psi_{\pi(2j-1)} \psi_{\pi(2j)}^\prime  D_{\pi(2j-1)}\\
    &\hspace{1em} -\frac{1}{G} \sum_{1\leq j \leq G}\psi_{\pi(2j)} \psi_{\pi(2j-1)}^\prime  D_{\pi(2j-1)} - \frac{1}{G} \sum_{1\leq j \leq G}\psi_{\pi(2j-1)} \psi_{\pi(2j)}^\prime  D_{\pi(2j)} \\
    &= \frac{1}{G} \sum_{1\leq g \leq 2 G} \psi_g \psi_g^\prime -   \frac{1}{G} \sum_{1\leq j \leq G} (\psi_{\pi(2j)} \psi_{\pi(2j-1)}^\prime + \psi_{\pi(2j-1)} \psi_{\pi(2j)}^\prime) ~.
\end{align*}
Assumptions \ref{ass:QG}, \ref{ass:pair_form}, \ref{ass:DGP}, \ref{ass:indep_xz}, \ref{ass:psi} and Lemma \ref{lem:cross} imply
\[ \frac{1}{G} \sum_{1\leq j \leq G} (\hat\psi_{1,j} - \hat\psi_{0,j}) (\hat\psi_{1,j} - \hat\psi_{0,j})^{\prime} \xrightarrow{P} 2 E[\psi_g \psi_g^\prime] - 2 E[E[\psi_g| W_g][\psi_g| W_g]'] = 2 E[\var[\psi_g | W_g]]~. \]
On the other hand,
\begin{align*}
    &\frac{1}{G} \sum_{1\leq j \leq G}   (\hat\psi_{1,j} - \hat\psi_{0,j}) (\tilde\mu_{1,j} - \tilde\mu_{0,j})  \\
    &= \frac{1}{G} \sum_{1\leq j \leq G}  \hat\psi_{1,j} \tilde\mu_{1,j} + \hat\psi_{0,j} \tilde\mu_{0,j} - \tilde\mu_{1,j} \hat\psi_{0,j} - \tilde\mu_{0,j} \hat\psi_{1,j}\\
    &= \frac{1}{G} \sum_{1\leq g \leq 2 G}  \left(\bar Y_{g}(1) -  \frac{E[\bar Y_{g}(1) N_g]}{E[N_g]}\right) N_{g} \psi_{g} D_{g} + \frac{1}{G} \sum_{1\leq g \leq 2 G} \left(\bar Y_{g}(0) -  \frac{E[\bar Y_{g}(0) N_g]}{E[N_g]}\right) N_{g} \psi_{g} (1-D_{g}) \\
    &\hspace{3em} - \frac{1}{G} \sum_{1\leq j \leq G} \left(\bar Y_{\pi(2j-1)}(1) -  \frac{E[\bar Y_{g}(1) N_g]}{E[N_g]}\right) N_{\pi(2j-1)} \psi_{\pi(2j)} D_{\pi(2j-1)} \\
    &\hspace{3em} - \frac{1}{G} \sum_{1\leq j \leq G} \left(\bar Y_{\pi(2j)}(1) -  \frac{E[\bar Y_{g}(1) N_g]}{E[N_g]}\right) N_{\pi(2j)} \psi_{\pi(2j-1)}D_{\pi(2j)} \\
    &\hspace{3em} - \frac{1}{G} \sum_{1\leq j \leq G} \left(\bar Y_{\pi(2j-1)}(0) -  \frac{E[\bar Y_{g}(0) N_g]}{E[N_g]}\right) N_{\pi(2j-1)} \psi_{\pi(2j)} (1-D_{\pi(2j-1)} ) \\
    &\hspace{3em} - \frac{1}{G} \sum_{1\leq j \leq G} \left(\bar Y_{\pi(2j)}(0) -  \frac{E[\bar Y_{g}(0) N_g]}{E[N_g]}\right) N_{\pi(2j)} \psi_{\pi(2j-1)} (1- D_{\pi(2j)}) ~. 
\end{align*}
Lemma \ref{lem:cross} implies that under Assumptions \ref{ass:QG}, \ref{ass:pair_form}, \ref{ass:DGP}, \ref{ass:indep_xz}, and \ref{ass:psi}, we have
\begin{align*}
    &\frac{1}{G} \sum_{1\leq g \leq 2 G}  \left(\bar Y_{g}(1) -  \frac{E[\bar Y_{g}(1) N_g]}{E[N_g]}\right) N_{g} \psi_{g} D_{g} \xrightarrow{P} E[\bar Y_{g}(1) N_{g} \psi_{g}] - \frac{E[\bar Y_{g}(1) N_g]}{E[N_g]} E[N_{g} \psi_{g}] \\
    &\frac{1}{G} \sum_{1\leq g \leq 2 G} \left(\bar Y_{g}(0) -  \frac{E[\bar Y_{g}(0) N_g]}{E[N_g]}\right) N_{g} \psi_{g} (1-D_{g}) \xrightarrow{P} E[\bar Y_{g}(0) N_{g} \psi_{g}] - \frac{E[\bar Y_{g}(0) N_g]}{E[N_g]} E[N_{g} \psi_{g}] \\
    &\frac{1}{G} \sum_{1\leq j \leq G} \left(\bar Y_{\pi(2j-1)}(1) -  \frac{E[\bar Y_{g}(1) N_g]}{E[N_g]}\right) N_{\pi(2j-1)} \psi_{\pi(2j)} D_{\pi(2j-1)} \\
    & \hspace{3em} \xrightarrow{P} \frac{1}{2} E\left[E\left[ \bar Y_g(1) N_g | W_g\right] E\left[\psi_g | W_g\right]\right] - \frac{1}{2} \frac{E[\bar Y_{g}(1) N_g]}{E[N_g]} E\left[E\left[ N_g | W_g\right] E\left[\psi_g | W_g\right]\right]~.
\end{align*}
Therefore,
\begin{align*}
    & \frac{1}{G} \sum_{1\leq j \leq G}   (\hat\psi_{1,j} - \hat\psi_{0,j}) (\tilde\mu_{1,j} - \tilde\mu_{0,j}) \\
    &\xrightarrow{P} E[(\bar Y_{g}(1) + \bar Y_{g}(0)) N_{g} \psi_{g}] -E\left[E\left[ (\bar Y_g(1) + \bar Y_g(0)) N_g | W_g\right] E\left[\psi_g | W_g\right]\right] \\
    &\hspace{3em} - \frac{E[(\bar Y_{g}(1) + \bar Y_{g}(0)) N_{g}]}{E[N_g]} E[N_{g} \psi_{g}] + \frac{E[(\bar Y_{g}(1) + \bar Y_{g}(0)) N_{g}]}{E[N_g]} E\left[E\left[ N_g | W_g\right] E\left[\psi_g | W_g\right]\right]\\
    & = E\left[  \cov \left[(\bar Y_{g}(1) + \bar Y_{g}(0)) N_{g} - \frac{E[(\bar Y_{g}(1) + \bar Y_{g}(0)) N_{g}]}{E[N_g]} N_g,  \psi_g \middle | W_g\right] \right] \\
    & = E\left[  \cov \left[ \tilde{Y}_g(1) + \tilde{Y}_g(0),  \psi_g \middle | W_g\right] \right] E[N_g]~,
\end{align*}
as desired.
\end{proof}

\begin{lemma} \label{lem:tilde-hat}
    Suppose all assumptions in Theorem \ref{thm:adj} hold. Then, $\tilde \beta_G - \hat \beta_G \xrightarrow{P} 0$.
\end{lemma}

\begin{proof}
    Note that
\begin{align*}
    \tilde \beta_G - \hat \beta_G &= \left(\frac{1}{G} \sum_{1\leq j \leq G} (\hat\psi_{1,j} - \hat\psi_{0,j}) (\hat\psi_{1,j} - \hat\psi_{0,j})^{\prime} \right)^{-1} \left(\frac{1}{G} \sum_{1\leq j \leq G}   (\hat\psi_{1,j} - \hat\psi_{0,j}) (\tilde\mu_{1,j} - \tilde\mu_{0,j} - (\hat\mu_{1,j} - \hat\mu_{0,j}))   \right) ~.
\end{align*}
We want to show that the following term converges to zero:
\begin{align*}
    &\frac{1}{G} \sum_{1\leq j \leq G}   (\hat\psi_{1,j} - \hat\psi_{0,j}) (\tilde\mu_{1,j}  - \hat\mu_{1,j})  \\
    &=\frac{1}{G} \sum_{1\leq j \leq G}   (\hat\psi_{1,j} - \hat\psi_{0,j})\left( \frac{\frac{1}{G}\sum_{1\leq g \leq 2G} \bar{Y}_{g} D_{g}  N_{g}}{\frac{1}{G}\sum_{1\leq g \leq 2G} D_{g}  N_{g}} -  \frac{E[\bar Y_{g}(1) N_g]}{E[N_g]}\right) N_{\pi(2j-1)} D_{\pi(2j-1)}  \\
    &\hspace{3em} + \frac{1}{G} \sum_{1\leq j \leq G}   (\hat\psi_{1,j} - \hat\psi_{0,j})\left(\frac{\frac{1}{G}\sum_{1\leq g \leq 2G} \bar{Y}_{g} D_g  N_g}{\frac{1}{G}\sum_{1\leq g \leq 2G} D_g  N_g}-  \frac{E[\bar Y_{g}(1) N_g]}{E[N_g]}\right) N_{\pi(2j)} D_{\pi(2j)} \\
    &= \left(\frac{\frac{1}{G}\sum_{1\leq g \leq 2G} \bar{Y}_{g} D_g  N_g}{\frac{1}{G}\sum_{1\leq g \leq 2G} D_g N_g}-  \frac{E[\bar Y_{g}(1) N_g]}{E[N_g]}\right) \frac{1}{G} \sum_{1\leq j \leq G}   (\hat\psi_{1,j} - \hat\psi_{0,j}) (N_{\pi(2j-1)} D_{\pi(2j-1)}+N_{\pi(2j)} D_{\pi(2j)}) ~.
\end{align*}
Lemma \ref{lem:cross} implies
\begin{align*}
    & \frac{1}{G} \sum_{1\leq j \leq G} (\hat\psi_{1,j} - \hat\psi_{0,j}) (N_{\pi(2j-1)} D_{\pi(2j-1)}+N_{\pi(2j)} D_{\pi(2j)}) \\
    & = \frac{1}{G} \sum_{1 \leq g \leq 2G} \psi_g N_g D_g - \frac{1}{G} \sum_{1 \leq j \leq G} \psi_{\pi(2j)} N_{\pi(2j - 1)} D_{\pi(2j - 1)} - \frac{1}{G} \sum_{1 \leq j \leq G} \psi_{\pi(2j - 1)} N_{\pi(2j)} D_{\pi(2j)} \\
    & \xrightarrow{P} E[\psi_g N_g] - E[E[\psi_g | W_g] E[N_g | W_g]] \\
    & = E\left[\operatorname{Cov}\left[\psi_g, N_g | W_g\right] \right] ~.
\end{align*}
By Lemma \ref{lem:wlln} and the continuous mapping theorem,
\begin{equation*}
    \left(\frac{\frac{1}{G}\sum_{1\leq g \leq 2G} \bar{Y}_{g} D_g  N_g}{\frac{1}{G}\sum_{1\leq g \leq 2G} D_g  N_g}-  \frac{E[\bar Y_{g}(1) N_g]}{E[N_g]}\right) \xrightarrow{P} 0~,
\end{equation*}
which implies that
\[ \frac{1}{G} \sum_{1\leq j \leq G}   (\hat\psi_{1,j} - \hat\psi_{0,j}) (\tilde\mu_{1,j}  - \hat\mu_{1,j})\xrightarrow{P} 0~. \]
Similarly,
\[ \frac{1}{G} \sum_{1\leq j \leq G}   (\hat\psi_{1,j} - \hat\psi_{0,j}) (\tilde\mu_{0,j}  - \hat\mu_{0,j})\xrightarrow{P} 0~. \]
The result then follows.
\end{proof}

\begin{lemma}\label{lem:E_bounded}
If Assumption \ref{ass:QG} holds, then
\[\left|E[\bar{Y}^r_g(d)|X_g, N_g]\right| \le C \hspace{3mm} a.s.~,\]
for $r \in \{1, 2\}$ for some constant $C > 0$,  
\[E\left[\bar{Y}_g^r(d)N_g^\ell\right] < \infty~,\]
for $r \in \{1, 2\}, \ell \in \{0, 1, 2\}$, and
\[E\left[E[\bar{Y}_g(d)N_g|X_g]^2\right] < \infty~.\]
\end{lemma}
\begin{proof}
We show the first statement for $r = 2$, since the case $r = 1$ follows similarly. By the Cauchy-Schwarz inequality,
\[\bar{Y}_g(d)^2 = \left(\frac{1}{|\mathcal{M}_g|}\sum_{i\in \mathcal{M}_g}Y_{i,g}(d)\right)^2 \le \frac{1}{|\mathcal{M}_g|}\sum_{i\in \mathcal{M}_g}Y_{i,g}(d)^2~,\]
and hence
\[\left|E[\bar{Y}_g(d)^2|X_g, N_g]\right| \le E\left[\frac{1}{|\mathcal{M}_g|}\sum_{i \in \mathcal{M}_g}E[Y_{i,g}(d)^2|X_g, N_g]\Bigg|X_g,N_g\right] \le C~,\]
where the first inequality follows from the above derivation, Assumption \ref{ass:QG}(e) and the law of iterated expectations, and final inequality follows from Assumption \ref{ass:QG}(d).
We show the next statement for $r = \ell = 2$, since the other cases follow similarly. By the law of iterated expectations,
\begin{align*}
E\left[\bar{Y}^2_g(d)N_g^2\right] &= E\left[N_g^2E[\bar{Y}^2_g(d)|X_g,N_g]\right] \\
&\lesssim E\left[N_g^2\right] < \infty~,
\end{align*}
where the final line follows by Assumption \ref{ass:QG}(c). Finally,
\begin{align*}
E\left[E[\bar{Y}_g(d)N_g|X_g]^2\right] &= E\left[E[N_gE[\bar{Y}_g(d)|X_g, N_g]|X_g]^2\right] \\
&\lesssim E\left[E[N_g|X_g]^2\right] < \infty~,
\end{align*}
where the final line follows from Jensen's inequality and Assumption \ref{ass:QG}(c).
%\begin{align*}
%E[\bar{Y}_g(d)|X_g, N_g] &= E\left[\frac{1}{|\mathcal{M}_g|}\sum_{i \in \mathcal{M}_g}E[Y_{i,g}(d)|X_g, N_g]\right] \le C'~,
%\end{align*}
%for some constant $C'$, by the law of iterated expectations, Assumption \ref{ass:QG}(e), Assumption .
\end{proof}

\begin{lemma}\label{lem:pair_form}
Suppose Assumption \ref{ass:pair_form} holds. Then,
\[\frac{1}{G}\sum_{g=1}^GN_{\pi(2g)}^{\ell}\left\|W_{\pi(2g)} - W_{\pi(2g-1)}\right\|^r \xrightarrow{P} 0~,\]
for $\ell \in \{0, 1, 2\}$, $r \in \{1, 2\}$.
\end{lemma}

\begin{proof}
By the Cauchy-Schwarz inequality
\[\frac{1}{G}\sum_{g=1}^GN^{\ell}_{\pi(2g)}|W_{\pi(2g)} - W_{\pi(2g-1)}|^r \le \left[\left(\frac{1}{G}\sum_{g=1}^GN_{\pi(2g)}^{2\ell}\right)\left(\frac{1}{G}\sum_{g=1}^G|W_{\pi(2g)} - W_{\pi(2g-1)}|^{2r}\right)\right]^{1/2}~,\]
$\frac{1}{G}\sum_{g=1}^{G} N^{2\ell}_{\pi(2g)} \le \frac{1}{G}\sum_{g=1}^{2G} N_{g}^{2\ell} = O_P(1)$ by the law of large numbers, $\frac{1}{G}\sum_g\|W_{\pi(2g)} - W_{\pi(2g-1)}\|^{2r} \xrightarrow{P} 0$ by assumption, hence the result follows.
\end{proof}

\begin{lemma}\label{lem:pair_squared}
If Assumptions \ref{ass:QG} and \ref{ass:pair_form} hold,
\[\frac{1}{G}\sum_{g=1}^G\left|N^2_{\pi(2g)} - N^2_{\pi(2g-1)}\right| \xrightarrow{P} 0~.\]
\end{lemma}
\begin{proof}
\begin{align*}
\frac{1}{G}\sum_{g=1}^G\left|N^2_{\pi(2g)} - N^2_{\pi(2g-1)}\right| &= \frac{1}{G}\sum_{g=1}^G\left|N_{\pi(2g)} - N_{\pi(2g-1)}\right|\left|N_{\pi(2g)} + N_{\pi(2g-1)}\right| \\
&\le \left[\left(\frac{1}{G}\sum_{g=1}^G\left|N_{\pi(2g)} - N_{\pi(2g-1)}\right|^2\right)\left(\frac{1}{G}\sum_{g=1}^G\left|N_{\pi(2g)} + N_{\pi(2g-1)}\right|^2\right)\right]^{1/2}~,
\end{align*}
where the inequality follows by Cauchy-Schwarz. It follows from an argument similar to the proof of Lemma \ref{lem:pair_form} that $\frac{1}{G}\sum_{g=1}^G\left|N_{\pi(2g)} + N_{\pi(2g-1)}\right|^2 = O_P(1)$. By Assumption \ref{ass:pair_form}, $\frac{1}{G}\sum_{g=1}^G\left|N_{\pi(2g)} - N_{\pi(2g-1)}\right|^2 \xrightarrow{P} 0$. Hence the result follows.
\end{proof}

\begin{lemma} \label{lem:wlln}
    Let $Z_1, Z_2, \ldots, Z_G$ be i.i.d random variables. Then,
    \begin{enumerate}[\rm (a)]
        \item Suppose $E[|Z_g|] < \infty$, $E[Z_g | X_g = x]$ is Lipschitz,
        \[ Z^{(G)} \independent D^{(G)} | X^{(G)}~, \]
        and conditional on $X^{(G)}$, $(D_{\pi(2j-1)}, D_{\pi(2j)})$, $j = 1, ..., G$ are i.i.d.\ and each uniformly distributed over $\{(0,1), (1,0)\}$, and
        \[ \frac{1}{G} \sum_{1 \leq j \leq G} \|X_{\pi(2j - 1)} - X_{\pi(2j)}\| \stackrel{P}{\to} 0~. \]
        Then, as $G \to \infty$,
        \[ \frac{1}{G} \sum_{1 \leq g \leq 2G} Z_g D_g \stackrel{P}{\to} E[Z_g]~. \]
        \item Suppose $E[Z_g^2] < \infty$, $E[Z_g | W_g = w]$ is Lipschitz, $E[N_g^{2\ell}] < \infty$,
        \[ Z^{(G)} \independent D^{(G)} | W^{(G)}~, \]
        and conditional on $W^{(G)}$, $(D_{\pi(2j-1)}, D_{\pi(2j)})$, $j = 1, ..., G$ are i.i.d.\ and each uniformly distributed over $\{(0,1), (1,0)\}$, and
        \[ \frac{1}{G}  \sum_{1 \leq j \leq G} \|W_{\pi(2j - 1)} - W_{\pi(2j)}\|^2 \stackrel{P}{\to} 0~. \]
        Then, as $G \to \infty$,
        \[ \frac{1}{G} \sum_{1 \leq g \leq 2G} Z_g N_g^\ell D_g \stackrel{P}{\to} E[Z_g N_g^\ell]~. \]
    \end{enumerate}
\end{lemma}

\begin{proof}
    (a) follows from Lemma S.1.5 in \cite{bai2022inference}. (b) follows by combining the arguments in the proofs of that lemma and the proof of Lemma \ref{lem:L_N}.
\end{proof}

\begin{lemma} \label{lem:cross}
    Let $(Z_1, \tilde Z_1), \ldots, (Z_G, \tilde Z_G)$ be i.i.d random vectors. Suppose Assumption \ref{ass:QG} holds, $E[|Z_g|] < \infty$, $E[Z_g | X_g = x]$ and $E[\tilde{Z}_g | X_g = x]$ are Lipschitz, \[ (Z^{(G)}, \tilde Z^{(G)}) \independent D^{(G)} | X^{(G)}~, \]
        and conditional on $X^{(G)}$, $(D_{\pi(2j-1)}, D_{\pi(2j)})$, $j = 1, ..., G$ are i.i.d.\ and each uniformly distributed over $\{(0,1), (1,0)\}$, and
        \[ \frac{1}{G} \sum_{1 \leq j \leq G} \|X_{\pi(2j - 1)} - X_{\pi(2j)}\|^2 \stackrel{P}{\to} 0~, \]
    Then,
    \begin{align*}
        \frac{1}{n} \sum_{1 \leq j \leq G} Z_{\pi(2j - 1)} \tilde Z_{\pi(2j)} & \stackrel{P}{\to} E[E[Z_g | X_g] E[\tilde Z_g | X_g]] \\
        \frac{1}{n} \sum_{1 \leq j \leq G} Z_{\pi(2j - 1)} \tilde Z_{\pi(2j)} D_{\pi(2j - 1)} & \stackrel{P}{\to} \frac{1}{2} E[E[Z_g | X_g] E[\tilde Z_g | X_g]]~.
    \end{align*}
\end{lemma}

\begin{proof}
    The proof is identical to the proof of Lemma S.1.6 in \cite{bai2022inference} and is therefore omitted.
\end{proof}

\begin{lemma}\label{lemma:assumptions-for-bai-inference}
If Assumptions \ref{ass:QG} holds, and additionally Assumptions \ref{ass:pair_formX}-\ref{ass:DGP2}, \ref{ass:close-pairs-of-pairsX} (or Assumptions \ref{ass:pair_form}-\ref{ass:DGP}, \ref{ass:close-pairs-of-pairs}) hold, then
\begin{enumerate}
    \item $E\left[\tilde Y_{g}^{2}(d)\right] < \infty$ for $d \in \{0,1\}$.
    \item $((\tilde Y_{g}(1), \tilde Y_{g}(0)):1\leq g\leq 2G) \independent D^{(G)} | X^{(G)}$ or $((\tilde Y_{g}(1), \tilde Y_{g}(0)):1\leq g\leq 2G) \independent D^{(G)} | W^{(G)}$.
     \item When not matching on cluster size, $\frac{1}{G} \sum_{1 \leq j \leq G}\left|\mu_d(X_{\pi(2j)}) - \mu_d(X_{\pi(2j-1)})\right| \xrightarrow{P} 0$, where we use $\mu_d(X_g)$ to denote $E[\tilde Y_g(d) | X_g]$ for $d\in\{0,1\}$ or when matching on cluster size 
     \[\frac{1}{G} \sum_{1 \leq j \leq G}\left|\mu_d(W_{\pi(2j)}) - \mu_d(W_{\pi(2j-1)})\right| \xrightarrow{P} 0~.\]
    \item When not matching on cluster size, 
    \[\frac{1}{G} \sum_{1 \leq j \leq G}\left|\left(\mu_1(X_{\pi(2j)}) - \mu_1(X_{\pi(2j-1)})\right)\left(\mu_0(X_{\pi(2j)}) - \mu_0(X_{\pi(2j-1)})\right)\right| \xrightarrow{P} 0~,\]
    
    or when matching on cluster size 
    \[\frac{1}{G} \sum_{1 \leq j \leq G}\left|\left(\mu_1(W_{\pi(2j)}) - \mu_1(W_{\pi(2j-1)})\right)\left(\mu_0(W_{\pi(2j)}) - \mu_0(W_{\pi(2j-1)})\right)\right| \xrightarrow{P} 0~.\]
    \item When not matching on cluster size 
    \[\frac{1}{4 G} \sum_{k \in\{2,3\}, \ell \in\{0,1\}} \sum_{1 \leq j \leq \frac{G}{2}}\left(\mu_{d}\left(X_{\pi(4 j-\ell)}\right)-\mu_{d}\left(X_{\pi(4 j-k)}\right)\right)^{2} \xrightarrow{P} 0~,\]
    
    or when matching on cluster size 
    \[\frac{1}{4 G} \sum_{k \in\{2,3\}, \ell \in\{0,1\}} \sum_{1 \leq j \leq \frac{G}{2}}\left(\mu_{d}\left(W_{\pi(4 j-\ell)}\right)-\mu_{d}\left(W_{\pi(4 j-k)}\right)\right)^{2} \xrightarrow{P} 0~.\]
\end{enumerate}
\end{lemma}
\begin{proof}
Note that
\begin{align*}
    E\left[\tilde Y_{g}^{2}(d)\right] &\leq E\left[ N_g^2 \left(\bar{Y}_{g}(d)-\frac{E\left[\bar{Y}_{g}(d) N_{g}\right]}{E\left[N_{g}\right]}\right)^2 \right] \\
    &\lesssim E\left[ N_g^2 \bar{Y}_{g}^2(d)\right] + \left(\frac{E\left[\bar{Y}_{g}(d) N_{g}\right]}{E\left[N_{g}\right]}\right)^2 E[N_g^2] < \infty
\end{align*}
where the inequality follows by Lemma \ref{lem:E_bounded}. The second result follows directly by inspection and Assumption \ref{ass:indep_pairsX} (or Assumption \ref{ass:indep_pairs}). In terms of the third result, by Assumption  \ref{ass:pair_formX} and \ref{ass:DGP2},
\begin{align*}
    \frac{1}{G} \sum_{1 \leq j \leq G}\left|\mu_1(X_{\pi(2j)}) - \mu_1(X_{\pi(2j-1)})\right| \lesssim \frac{1}{G} \sum_{1 \leq j \leq G}  \left\|X_{\pi(2j)}- X_{\pi(2j-1)} \right\| \xrightarrow{P} 0~.
\end{align*}
Meanwhile,
\begin{align*}
    &\frac{1}{G} \sum_{1 \leq j \leq G}\left|\mu_1(W_{\pi(2j)}) - \mu_1(W_{\pi(2j-1)})\right| \\
    & \lesssim \frac{1}{G} \sum_{1 \leq j \leq G}\left| E[N_{\pi(2j)}  \bar{Y}_{\pi(2j)}(d)| W_{\pi(2j)}] - E[N_{\pi(2j-1)}  \bar{Y}_{\pi(2j-1)}(d)| W_{\pi(2j-1)}]\right| \\
    &\hspace{3em} + \frac{1}{G} \sum_{1 \leq j \leq G}\left| E[N_{\pi(2j)} | W_{\pi(2j)}] - E[N_{\pi(2j-1)}  | W_{\pi(2j-1)}]\right| \\
    &\lesssim \frac{1}{G} \sum_{1 \leq j \leq G}\left| N_{\pi(2j)} \left(E[  \bar{Y}_{\pi(2j)}(d)| W_{\pi(2j)}] - E[\bar{Y}_{\pi(2j-1)}(d)| W_{\pi(2j-1)}]\right)\right| + \frac{1}{G} \sum_{1 \leq j \leq G}\left| N_{\pi(2j)}  - N_{\pi(2j-1)} \right| \\
    &\hspace{3em} + \frac{1}{G} \sum_{1 \leq j \leq G}\left| (N_{\pi(2j)}-N_{\pi(2j-1)})  E[  \bar{Y}_{\pi(2j-1)}(d)| W_{\pi(2j-1)}]\right| \\
    &\lesssim \frac{1}{G} \sum_{1 \leq j \leq G} N_{\pi(2j)} \left\|W_{\pi(2j)}- W_{\pi(2j-1)} \right\|~,
\end{align*}
which converges to zero in probability by Assumption \ref{ass:pair_form} and Lemma \ref{lem:pair_form}. To prove the fourth result, by Assumption \ref{ass:pair_formX} and \ref{ass:DGP2}, 
\begin{equation*}
    \frac{1}{G} \sum_{1 \leq j \leq G}\left|\left(\mu_1(X_{\pi(2j)}) - \mu_1(X_{\pi(2j-1)})\right)\left(\mu_0(X_{\pi(2j)}) - \mu_0(X_{\pi(2j-1)})\right)\right| \lesssim \frac{1}{G} \sum_{1 \leq j \leq G}  \left\|X_{\pi(2j)}- X_{\pi(2j-1)} \right\|^2 \xrightarrow{P} 0~.
\end{equation*}
Similarly, 
\begin{align*}
    &\frac{1}{G} \sum_{1 \leq j \leq G}\left|\left(\mu_1(W_{\pi(2j)}) - \mu_1(W_{\pi(2j-1)})\right)\left(\mu_0(W_{\pi(2j)}) - \mu_0(W_{\pi(2j-1)})\right)\right| \\
    & \leq \frac{1}{G} \sum_{1 \leq j \leq G}\left|\mu_1(W_{\pi(2j)}) - \mu_1(W_{\pi(2j-1)})\right|\left|\mu_0(W_{\pi(2j)}) - \mu_0(W_{\pi(2j-1)})\right| \\
    & \lesssim\frac{1}{G} \sum_{1 \leq j \leq G} N_{\pi(2j)}^2 \left\|W_{\pi(2j)}- W_{\pi(2j-1)} \right\|^2 \xrightarrow{P} 0~,
\end{align*}
where the last step follows by Assumption \ref{ass:pair_form} and Lemma \ref{lem:pair_form}. Finally, the fifth result follows the same argument by Assumption \ref{ass:close-pairs-of-pairsX} (or Assumption \ref{ass:close-pairs-of-pairs}).
\end{proof}

\begin{lemma}\label{lem:rand_joint}
\[\rho\left(\mathcal{L}\left((\mathbb{K}_G^{YN}, \mathbb{K}_G^{N})'|Z^{(G)}\right) , N\left(0, \mathbb{V}_R\right)\right) \xrightarrow{P} 0~,\]
where
\[\begin{pmatrix}\mathbb{K}_G^{YN} \\ \mathbb{K}_G^{N} \end{pmatrix} = \begin{pmatrix}
 \frac{1}{\sqrt{G}}\sum_{1 \le j \le G}\epsilon_j\left(N_{\pi(2j)}\bar{Y}_{\pi(2j)} - N_{\pi(2j-1)}\bar{Y}_{\pi(2j-1)}\right)(D_{\pi(2j)} - D_{\pi(2j-1)}) \\
\frac{1}{\sqrt{G}}\sum_{1 \le j \le G}\epsilon_j(N_{\pi(2j)} - N_{\pi(2j-1)})(D_{\pi(2j)} - D_{\pi(2j-1)}) 
\end{pmatrix}~,\]
and where, in the case where we match on cluster size,
\[ \mathbb V_R = \begin{pmatrix}
\mathbb V_R^1 & 0 \\
0 & 0
\end{pmatrix}~,\]
with
\[\mathbb V_R^1 = E[\var(N_g\bar{Y}_g(1)|W_g)] + E[\var(N_g\bar{Y}_g(0)|W_g)] + E\left[(E[N_g\bar{Y}_g(1)|W_g] - E[N_g\bar{Y}_g(0)|W_g])^2\right]~,\]
and when we do not match on cluster size,
\[ \mathbb V_R = \begin{pmatrix}
\mathbb V_R^{1,1} & \mathbb{V}_R^{1,2} \\
\mathbb V_R^{1,2} & \mathbb{V}_R^{2,2}
\end{pmatrix}~,\]
with
\begin{align*}
    \mathbb V_R^{1,1} &= E[\var(N_g\bar{Y}_g(1)|X_g)] + E[\var(N_g\bar{Y}_g(0)|X_g)] + E\left[(E[N_g\bar{Y}_g(1)|X_g] - E[N_g\bar{Y}_g(0)|X_g])^2\right] \\
        \mathbb V_R^{1,2} &= E[N^2_g\bar{Y}_g(1)] + E[N^2_g\bar{Y}_g(0)] - \left(E\left[E[N_g\bar{Y}_g(1)|X_g]E[N_g|X_g]\right] + E\left[E[N_g\bar{Y}_g(0)|X_g]E[N_g|X_g]\right]\right) \\
        \mathbb{V}_R^{2,2} &= 2E[\var(N_g|X_g)]~.
\end{align*}
\end{lemma}
\begin{proof}
Using the fact that $\epsilon_j$, $j = 1, \ldots, G$ and $\epsilon_j(D_{\pi(2j)} - D_{\pi(2j-1)})$, $j = 1, \ldots, G$ have the same distribution conditional on $Z^{(G)}$, it suffices to study the limiting distribution of $(\tilde{\mathbb{K}}_{G}^{YN}, \tilde{\mathbb{K}}_G^{N})'$ conditional on $Z^{(G)}$, where 
\begin{align*}
\tilde{\mathbb{K}}_G^{YN} &:= \frac{1}{\sqrt{G}}\sum_{1 \le j \le G}\epsilon_j\left(N_{\pi(2j)}\bar{Y}_{\pi(2j)} - N_{\pi(2j-1)}\bar{Y}_{\pi(2j-1)}\right)~,\\
\tilde{\mathbb{K}}_G^{N} &:= \frac{1}{\sqrt{G}}\sum_{1 \le j \le G}\epsilon_j\left(N_{\pi(2j)} - N_{\pi(2j-1)}\right)~.
\end{align*}
We will show
\begin{equation}\label{eq:rand_converge}
\rho\left(\mathcal{L}\left((\tilde{\mathbb{K}}_G^{YN}, \tilde{\mathbb{K}}_G^{N})'|Z^{(G)}\right), N(0, \mathbb{V}_R)\right) \xrightarrow{P} 0~,
\end{equation}
where $\mathcal L(\cdot)$ denote the law and $\rho$ is any metric that metrizes weak convergence. To that end, we will employ the Lindeberg central limit theorem in Proposition 2.27 of \cite{van_der_vaart1998asymptotic} and a subsequencing argument. Indeed, to verify \eqref{eq:rand_converge}, note we need only show that for any subsequence $\{G_k\}$ there exists a further subsequence $\{G_{k_l}\}$ such that
\begin{equation}\label{eq:rand_converge_as}
\rho\left(\mathcal{L}\left((\tilde{\mathbb{K}}_{G_{k_l}}^{YN}, \tilde{\mathbb{K}}_{G_{k_l}}^{N})'|Z^{(G_{k_l})}\right), N(0, \mathbb{V}_R)\right) \to 0 \text{ with probability one}~.
\end{equation}
To that end, define
\[ \mathbb V_{R, n} = \begin{pmatrix} \mathbb V_{R, n}^{1, 1} & \mathbb V_{R, n}^{1, 2} \\
\mathbb V_{R, n}^{1, 2} & \mathbb V_{R, n}^{2, 2}
\end{pmatrix} = \var[(\tilde{\mathbb{K}}_G^{YN}, \tilde{\mathbb{K}}_G^{N})' | Z^{(G)}]~, \]
where
\begin{align*}
    \mathbb V_{R, n}^{1, 1} & = \frac{1}{G} \sum_{1 \leq j \leq G} (N_{\pi(2j)} \bar Y_{\pi(2j)} - N_{\pi(2j - 1)} \bar Y_{\pi(2j - 1)})^2 \\
    \mathbb V_{R, n}^{1, 2} & = \frac{1}{G} \sum_{1 \leq j \leq G} (N_{\pi(2j)} \bar Y_{\pi(2j)} - N_{\pi(2j - 1)} \bar Y_{\pi(2j - 1)}) (N_{\pi(2j)} - N_{\pi(2j - 1)}) \\
    \mathbb V_{R, n}^{2, 2} & = \frac{1}{G} \sum_{1 \leq j \leq G} (N_{\pi(2j)} - N_{\pi(2j - 1)})^2~.
\end{align*}
We first show that
\begin{equation} \label{eq:rand-var}
    \mathbb V_{R, n} \stackrel{P}{\to} \mathbb V_R~.
\end{equation}
Consider the case where we match on cluster size. The weak law of large numbers and Lemma \ref{lem:cross} imply 
\[\mathbb V_{R, n}^{1, 1} \xrightarrow{P} E[\var[N_g\bar{Y}_g(1)]|W_g] + E[\var[N_g\bar{Y}_g(0)]|W_g] + E\left[(E[N_g\bar{Y}_g(1)|W_g] - E[N_g\bar{Y}_g(0)|W_g])^2\right]~.\]
Next, we show that in this case $\mathbb V_{R, n}^{1, 2}$ and $\mathbb V_{R, n}^{2, 2}$ are $o_P(1)$. For $\mathbb V_{R, n}^{2, 2}$ this follows immediately from Assumption \ref{ass:pair_form}. For $\mathbb V_{R, n}^{1, 2}$ note that by the Cauchy-Schwarz inequality,
\begin{align*}
&\frac{1}{G}\sum_{1 \le j \le G}\left(\left(N_{\pi(2j)}\bar{Y}_{\pi(2j)} - N_{\pi(2j-1)}\bar{Y}_{\pi(2j-1)}\right)\left(N_{\pi(2j)} - N_{\pi(2j-1)}\right)\right) \\
&\le \left(\left(\frac{1}{G}\sum_{1 \le j \le G}\left(N_{\pi(2j)}\bar{Y}_{\pi(2j)} - N_{\pi(2j-1)}\bar{Y}_{\pi(2j-1)}\right)^2\right)\left(\frac{1}{G}\sum_{1 \le j \le G}\left(N_{\pi(2j)} - N_{\pi(2j-1)}\right)^2\right)\right)^{1/2}~.
\end{align*}
The second term of the product on the RHS is $o_P(1)$ by Assumption \ref{ass:pair_form}. The first term is $O_P(1)$ since
\[\frac{1}{G}\sum_{1 \le j \le G}\left(N_{\pi(2j)}\bar{Y}_{\pi(2j)} - N_{\pi(2j-1)}\bar{Y}_{\pi(2j-1)}\right)^2 \lesssim \frac{1}{G}\sum_{1 \le g \le 2G}N_g^2\bar{Y}_g(1)^2 + \frac{1}{G}\sum_{1 \le g \le 2G}N_g^2\bar{Y}_g(0)^2 = O_P(1)~,\]
where the first inequality follows from exploiting the fact that $|a - b|^2 \le 2(a^2 + b^2)$ and the definition of $\bar{Y}_g$, and the final equality follows from Lemma \ref{lem:E_bounded} and the law of large numbers. We can thus conclude that $\mathbb V_{R, n}^{1, 2} = o_P(1)$ when matching on cluster size.

In the case where we do \emph{not} match on cluster size, again by the weak law of large numbers and Lemma \ref{lem:cross}, it can be shown that \eqref{eq:rand-var} holds. Next, we verify the Lindeberg condition in Proposition 2.27 of \cite{van_der_vaart1998asymptotic}. Note that for an arbitrary $\delta > 0$,
\begin{align*}
    &\frac{1}{G} \sum_{1 \leq j \leq G} E[((\epsilon_j (N_{\pi(2j)}\bar{Y}_{\pi(2j)} - N_{\pi(2j-1)}\bar{Y}_{\pi(2j-1)} ))^2 + (\epsilon_j (N_{\pi(2j)} - N_{\pi(2j-1)} ))^2) \\
    & \hspace{3em} \times I \{((\epsilon_j (N_{\pi(2j)}\bar{Y}_{\pi(2j)} - N_{\pi(2j-1)}\bar{Y}_{\pi(2j-1)} ))^2 + (\epsilon_j (N_{\pi(2j)} - N_{\pi(2j-1)} ))^2) > \delta^2 G\} | Z^{(G)} ] \\
    & = \frac{1}{G} \sum_{1 \leq j \leq G} E[((N_{\pi(2j)}\bar{Y}_{\pi(2j)} - N_{\pi(2j-1)}\bar{Y}_{\pi(2j-1)})^2 + (N_{\pi(2j)} - N_{\pi(2j-1)} )^2) \\
    & \hspace{3em} \times I \{((N_{\pi(2j)}\bar{Y}_{\pi(2j)} - N_{\pi(2j-1)}\bar{Y}_{\pi(2j-1)})^2 + (N_{\pi(2j)} - N_{\pi(2j-1)} )^2) > \delta^2 G\} | Z^{(G)} ] \\
    & \lesssim \frac{1}{G} \sum_{1 \leq j \leq G} (N_{\pi(2j)}\bar{Y}_{\pi(2j)} - N_{\pi(2j-1)}\bar{Y}_{\pi(2j-1)})^2 I \{(N_{\pi(2j)}\bar{Y}_{\pi(2j)} - N_{\pi(2j-1)}\bar{Y}_{\pi(2j-1)})^2 > \delta^2 G / 2\} \\
    & \hspace{3em} + \frac{1}{G} \sum_{1 \leq j \leq G} (N_{\pi(2j)} - N_{\pi(2j-1)})^2 I \{(N_{\pi(2j)} - N_{\pi(2j-1)})^2 > \delta^2 G / 2\}~.
\end{align*}
where the inequality follows from \eqref{eq:indicator} and the fact that $(N_g, \bar Y_g), 1 \leq g \leq 2G$ are all constants conditional on $Z^{(G)}$. The last line converges in probability to zero as long as we can show
\begin{align*}
    & \frac{1}{G} \max_{1 \leq j \leq G} (N_{\pi(2j)}\bar{Y}_{\pi(2j)} - N_{\pi(2j-1)}\bar{Y}_{\pi(2j-1)})^2 \stackrel{P}{\to} 0 \\
    & \frac{1}{G} \max_{1 \leq j \leq G} (N_{\pi(2j)} - N_{\pi(2j-1)})^2 \stackrel{P}{\to} 0~.
\end{align*}
Note
\begin{align*}
\frac{1}{G} \max_{1 \leq j \leq G} (N_{\pi(2j)}\bar{Y}_{\pi(2j)} - N_{\pi(2j-1)}\bar{Y}_{\pi(2j-1)})^2 & \lesssim \frac{1}{G}\max_{1 \le j \le G}\left(N^2_{\pi_{(2j-1)}}\bar{Y}^2_{\pi(2j-1)} + N^2_{\pi_{(2j)}}\bar{Y}^2_{\pi(2j)}\right) \\ 
& \lesssim \frac{1}{G}\max_{1 \le g \le 2G}\left(N^2_{g}\bar{Y}^2_g(1) + N^2_{g}\bar{Y}^2_g(0)\right) \stackrel{P}{\to} 0
\end{align*}
Where the first inequality follows from the fact that $|a - b|^2 \le 2(a^2 + b^2)$, the second by inspection, and the convergence by Lemma S.1.1 in \cite{bai2022inference} along with Assumption \ref{ass:QG}(c) and Lemma \ref{lem:E_bounded}. The second statement follows similarly. Therefore, we have verified both conditions in Proposition 2.27 of \cite{van_der_vaart1998asymptotic} hold in probability, and therefore for each subsequence there must exists a further subsequence along which both conditions hold with probability one, so \eqref{eq:rand_converge_as} holds, and the conclusion of the lemma follows.
\end{proof}

%\begin{equation}
%    \check{\nu}_G^2(\epsilon_1,\dots, \epsilon_G) = \hat{\tau}_G^2-\frac{1}{2}\check{\lambda}_G^2\left(\epsilon_1, \ldots, \epsilon_G\right)
%\end{equation}
%where $\hat{\tau}_G^2$ is defined in (\ref{eqn:define-variance-estimator}) and $\check{\Delta}_G\left(\epsilon_1, \ldots, \epsilon_G\right)$ is defined in (\ref{eq:rand_num}),
% \begin{align*}
%     \check{\lambda}_G^2\left(\epsilon_1, \ldots, \epsilon_G\right) = & \frac{2}{G} \sum_{1 \leq j \leq\left\lfloor G/2\right\rfloor} \epsilon_{2j-1} \epsilon_{2j} \left(\hat Y_{\pi(4 j-3)}-\hat Y_{\pi(4 j-2)}\right)\left(\hat Y_{\pi(4 j-1)}- \hat Y_{\pi(4 j)}\right)\\
%     & \times\left(D_{\pi(4 j-3)}-D_{\pi(4 j-2)}\right)\left(D_{\pi(4 j-1)}-D_{\pi(4 j)}\right)~.
% \end{align*}
%and, independent of $Z^{G}$, $\epsilon_j$, $j=1,\dots,G$ are i.i.d. Rademacher random variables

\section{Addtional Simulations}

\subsection{Simulation Results in Finite Populations}\label{sec:sims-finpop}
In this section, we compare the finite population design-based coverage properties of confidence intervals constructed using our proposed variance estimator $\hat{v}^2_G$ versus the estimators $\hat{\omega}^2_{\rm CR, G}$ and $\hat{\omega}^2_{\rm PCVE, G}$ introduced in Section \ref{sec:variance-estimate}. We revisit the simulation setting considered in Tables \ref{tab:model1X}--\ref{tab:model4XN} in Section \ref{sec:sims-unadj}, but now use each DGP to generate the covariates and outcomes only \emph{once}, and then fix these in repeated samples. 

Tables \ref{tab:FPmodel1X}--\ref{tab:FPmodel2XN} present our results. From Tables \ref{tab:FPmodel1X} and \ref{tab:FPmodel1XN}, we see that both $\hat{\nu}^2_G$ and $\hat{\omega}^2_{\rm PCVE, G}$ are consistent in large populations when there is sufficient ``homogeneity" in treatment effects, but undercover in small populations. This behavior is not surprising given that asymptotically exact inference is often feasible even in the design-based paradigm as long as treatment effects are sufficiently homogeneous; see for instance \cite{bai2024primer} for a discussion in the context of completely randomized experiments. On the other hand, Tables \ref{tab:FPmodel2X} and \ref{tab:FPmodel2XN} illustrate that when there is treatment effect heterogeneity, all three estimators are conservative, leading to a coverage probability of $1$ for all population sizes. However, although all three estimators over-cover, our proposed variance estimator $\hat{v}^2_G$ produces confidence intervals with the shortest average length in all cases.

\begin{table}[htbp]\centering
\begin{threeparttable}[b]
\def\sym#1{\ifmmode^{#1}\else\(^{#1}\)\fi}
\caption{Model 1 - Finite Population - Matching on $X_g$\tnote{1}}\label{cov_rate_m1}
\begin{tabular}{c*{9}{c}}
\hline\hline
    \multicolumn{1}{c}{$N_{max}/N_{min}$}&\multicolumn{1}{c}{VCE}
    &\multicolumn{1}{c}{$G=12$}&\multicolumn{1}{c}{$G=26$}&\multicolumn{1}{c}{$G=50$}
    &\multicolumn{1}{c}{$G=100$}&\multicolumn{1}{c}{$G=150$}
    &\multicolumn{1}{c}{$G=200$}&\multicolumn{1}{c}{$G=250$}\\
\hline
\multicolumn{9}{c}{} \\
\multicolumn{9}{c}{\bf Coverage} \\
\hline
\multirow{3}{*}{1.11}
& $\hat{v}^2$ & $0.8990$ & $0.9295$ & $0.9460$ & $0.9380$ & $0.9470$ & $0.9340$ & $0.9505$ \\
& CR & $1$ & $1$ & $0.9990$ & $1$ & $1$ & $1$ & $0.9995$ \\
& PCVE & $0.9095$ & $0.9270$ & $0.9450$ & $0.9365$ & $0.9470$ & $0.9325$ & $0.9480$ \\
\multicolumn{9}{l}{} \\
\multirow{3}{*}{1.42}
& $\hat{v}^2$ & $0.9060$ & $0.9315$ & $0.9475$ & $0.9375$ & $0.9515$ & $0.9330$ & $0.9465$ \\
& CR & $1$ & $1$ & $0.9990$ & $1$ & $1$ & $1$ & $0.9990$ \\
& PCVE & $0.9085$ & $0.9305$ & $0.9450$ & $0.9370$ & $0.9530$ & $0.9320$ & $0.9480$ \\
\multicolumn{9}{l}{} \\
\multirow{3}{*}{1.99}
& $\hat{v}^2$ & $0.9030$ & $0.9260$ & $0.9450$ & $0.9370$ & $0.9480$ & $0.9375$ & $0.9495$ \\
& CR & $1$ & $1$ & $1$ & $1$ & $1$ & $1$ & $0.9980$ \\
& PCVE & $0.9170$ & $0.9250$ & $0.9450$ & $0.9360$ & $0.9485$ & $0.9330$ & $0.9480$ \\
\multicolumn{9}{l}{} \\
\multirow{3}{*}{3.31}
& $\hat{v}^2$ & $0.8775$ & $0.9190$ & $0.9395$ & $0.9430$ & $0.9425$ & $0.9385$ & $0.9485$ \\
& CR & $1$ & $1$ & $1$ & $1$ & $1$ & $0.9995$ & $0.9965$ \\
& PCVE & $0.9075$ & $0.9175$ & $0.9435$ & $0.9395$ & $0.9435$ & $0.9360$ & $0.9470$ \\
\multicolumn{9}{l}{} \\
\multirow{3}{*}{9.80}
& $\hat{v}^2$ & $0.8880$ & $0.9085$ & $0.9440$ & $0.9390$ & $0.9415$ & $0.9455$ & $0.9405$ \\
& CR & $1$ & $1$ & $1$ & $0.9995$ & $1$ & $0.9965$ & $0.9925$ \\
& PCVE & $0.9075$ & $0.9100$ & $0.9465$ & $0.9400$ & $0.9420$ & $0.9455$ & $0.9410$ \\
\hline
\multicolumn{9}{c}{} \\
\multicolumn{9}{c}{{\bf Average Length}} \\
\hline
\multirow{3}{*}{1.11}
& $\hat{v}^2$ & $1.12824$ & $1.05815$ & $0.84888$ & $0.59101$ & $0.44808$ & $0.41502$ & $0.38434$ \\
& CR & $2.93266$ & $2.25955$ & $1.56492$ & $1.20447$ & $0.90146$ & $0.79447$ & $0.72000$ \\
& PCVE & $1.11395$ & $1.04746$ & $0.84517$ & $0.58917$ & $0.44726$ & $0.41469$ & $0.38418$ \\
\multicolumn{9}{l}{} \\
\multirow{3}{*}{1.42}
& $\hat{v}^2$ & $1.07152$ & $1.06921$ & $0.84835$ & $0.60402$ & $0.45275$ & $0.42419$ & $0.40010$ \\
& CR & $2.98019$ & $2.30454$ & $1.56619$ & $1.21866$ & $0.90291$ & $0.79774$ & $0.72714$ \\
& PCVE & $1.06215$ & $1.05823$ & $0.84533$ & $0.60213$ & $0.45189$ & $0.42370$ & $0.39987$ \\
\multicolumn{9}{l}{} \\
\multirow{3}{*}{1.99}
& $\hat{v}^2$ & $1.05214$ & $1.08426$ & $0.82321$ & $0.62589$ & $0.46226$ & $0.44162$ & $0.42537$ \\
& CR & $3.02136$ & $2.38754$ & $1.56696$ & $1.24393$ & $0.90557$ & $0.80431$ & $0.73828$ \\
& PCVE & $1.04815$ & $1.07367$ & $0.82097$ & $0.62399$ & $0.46142$ & $0.44114$ & $0.42500$ \\
\multicolumn{9}{l}{} \\
\multirow{3}{*}{3.31}
& $\hat{v}^2$ & $1.04528$ & $1.11925$ & $0.82767$ & $0.64119$ & $0.47469$ & $0.47427$ & $0.46200$ \\
& CR & $3.09726$ & $2.42478$ & $1.56226$ & $1.29434$ & $0.91920$ & $0.82070$ & $0.75534$ \\
& PCVE & $1.04739$ & $1.11017$ & $0.82627$ & $0.63952$ & $0.47380$ & $0.47367$ & $0.46149$ \\
\multicolumn{9}{l}{} \\
\multirow{3}{*}{9.80}
& $\hat{v}^2$ & $1.19775$ & $1.19395$ & $0.82358$ & $0.70239$ & $0.51101$ & $0.53635$ & $0.53192$ \\
& CR & $3.19729$ & $2.59330$ & $1.55023$ & $1.39250$ & $0.94697$ & $0.85953$ & $0.79422$ \\
& PCVE & $1.20833$ & $1.18286$ & $0.82301$ & $0.70132$ & $0.51043$ & $0.53549$ & $0.53114$ \\
\hline\hline
\end{tabular}
\label{tab:FPmodel1X}
\begin{tablenotes}
\item [1] Number of clusters $=2G$ with $G=12, 26, 50, 100, 150, 200, 250$. Number of replications for each $G$ is $2000$. $N_{max}=500$.
\end{tablenotes}
\end{threeparttable}
\end{table}

\begin{table}[htbp]\centering
\begin{threeparttable}[b]
\def\sym#1{\ifmmode^{#1}\else\(^{#1}\)\fi}
\caption{Model 1 - Finite Population - Matching on $X_g$ and $N_g$\tnote{1}}\label{CI_length_m1}
\begin{tabular}{c*{9}{c}}
\hline\hline
    \multicolumn{1}{c}{$N_{max}/N_{min}$}&\multicolumn{1}{c}{VCE}
    &\multicolumn{1}{c}{$G=12$}&\multicolumn{1}{c}{$G=26$}&\multicolumn{1}{c}{$G=50$}
    &\multicolumn{1}{c}{$G=100$}&\multicolumn{1}{c}{$G=150$}
    &\multicolumn{1}{c}{$G=200$}&\multicolumn{1}{c}{$G=250$}\\
\hline
\multicolumn{9}{c}{} \\
\multicolumn{9}{c}{\bf Coverage} \\
\hline
\multirow{3}{*}{1.11}
& $\hat{v}^2$ & $0.9225$ & $0.8930$ & $0.9365$ & $0.9475$ & $0.9500$ & $0.9560$ & $0.9505$ \\
& CR & $1$ & $1$ & $1$ & $1$ & $1$ & $1$ & $1$ \\
& PCVE & $0.9055$ & $0.9405$ & $0.9360$ & $0.9470$ & $0.9465$ & $0.9590$ & $0.9510$ \\
\multicolumn{9}{l}{} \\
\multirow{3}{*}{1.42}
& $\hat{v}^2$ & $0.9245$ & $0.9220$ & $0.9410$ & $0.9480$ & $0.9455$ & $0.9530$ & $0.9475$ \\
& CR & $1$ & $1$ & $1$ & $1$ & $1$ & $1$ & $1$ \\
& PCVE & $0.9115$ & $0.9230$ & $0.9390$ & $0.9460$ & $0.9540$ & $0.9545$ & $0.9485$ \\
\multicolumn{9}{l}{} \\
\multirow{3}{*}{1.99}
& $\hat{v}^2$ & $0.9370$ & $0.8555$ & $0.9490$ & $0.9455$ & $0.9515$ & $0.9480$ & $0.9540$ \\
& CR & $1$ & $1$ & $1$ & $1$ & $1$ & $1$ & $1$ \\
& PCVE & $0.9225$ & $0.9290$ & $0.9505$ & $0.9465$ & $0.9490$ & $0.9495$ & $0.9555$ \\
\multicolumn{9}{l}{} \\
\multirow{3}{*}{3.31}
& $\hat{v}^2$ & $0.9070$ & $0.8475$ & $0.9515$ & $0.9610$ & $0.9625$ & $0.9665$ & $0.9545$ \\
& CR & $1$ & $1$ & $1$ & $1$ & $1$ & $1$ & $1$ \\
& PCVE & $0.9035$ & $0.9425$ & $0.9515$ & $0.9595$ & $0.9610$ & $0.9615$ & $0.9550$ \\
\multicolumn{9}{l}{} \\
\multirow{3}{*}{9.80}
& $\hat{v}^2$ & $0.9020$ & $0.8175$ & $0.9415$ & $0.9580$ & $0.9665$ & $0.9635$ & $0.9645$ \\
& CR & $1$ & $1$ & $1$ & $1$ & $1$ & $1$ & $1$ \\
& PCVE & $0.8980$ & $0.9155$ & $0.9475$ & $0.9580$ & $0.9635$ & $0.9655$ & $0.9640$ \\
\hline
\multicolumn{9}{c}{} \\
\multicolumn{9}{c}{{\bf Average Length}} \\
\hline
\multirow{3}{*}{1.11}
& $\hat{v}^2$ & $1.06353$ & $0.54293$ & $0.39449$ & $0.26347$ & $0.20698$ & $0.14455$ & $0.13742$ \\
& CR & $2.93374$ & $2.26348$ & $1.56660$ & $1.20512$ & $0.90170$ & $0.79480$ & $0.72016$ \\
& PCVE & $1.04531$ & $0.54223$ & $0.39200$ & $0.26298$ & $0.20627$ & $0.14448$ & $0.13730$ \\
\multicolumn{9}{l}{} \\
\multirow{3}{*}{1.42}
& $\hat{v}^2$ & $1.03963$ & $0.80633$ & $0.29493$ & $0.19849$ & $0.16039$ & $0.12190$ & $0.09622$ \\
& CR & $2.98061$ & $2.30678$ & $1.56787$ & $1.21943$ & $0.90334$ & $0.79804$ & $0.72736$ \\
& PCVE & $1.01824$ & $0.80046$ & $0.29340$ & $0.19762$ & $0.16023$ & $0.12191$ & $0.09628$  \\
\multicolumn{9}{l}{} \\
\multirow{3}{*}{1.99}
& $\hat{v}^2$ & $1.09840$ & $0.63621$ & $0.25458$ & $0.16747$ & $0.14000$ & $0.12914$ & $0.09993$ \\
& CR & $3.01789$ & $2.38973$ & $1.56826$ & $1.24480$ & $0.90602$ & $0.80477$ & $0.73865$ \\
& PCVE & $1.08265$ & $0.63690$ & $0.25379$ & $0.16716$ & $0.13959$ & $0.12888$ & $0.09985$ \\
\multicolumn{9}{l}{} \\
\multirow{3}{*}{3.31}
& $\hat{v}^2$ & $1.02165$ & $0.71836$ & $0.26920$ & $0.21766$ & $0.17826$ & $0.13358$ & $0.09376$ \\
& CR & $3.09474$ & $2.42593$ & $1.56316$ & $1.29498$ & $0.91953$ & $0.82124$ & $0.75591$ \\
& PCVE & $1.00793$ & $0.71943$ & $0.26743$ & $0.21631$ & $0.17711$ & $0.13257$ & $0.09323$ \\
\multicolumn{9}{l}{} \\
\multirow{3}{*}{9.80}
& $\hat{v}^2$ & $1.13033$ & $0.88192$ & $0.28810$ & $0.26255$ & $0.18366$ & $0.12748$ & $0.10254$ \\
& CR & $3.19046$ & $2.59270$ & $1.55106$ & $1.39307$ & $0.94746$ & $0.86046$ & $0.79523$ \\
& PCVE & $1.11007$ & $0.87854$ & $0.28778$ & $0.26048$ & $0.18279$ & $0.12726$ & $0.10232$ \\
\hline\hline
\end{tabular}
\label{tab:FPmodel1XN}
\begin{tablenotes}
\item [1] Number of clusters $=2G$ with $G=12, 26, 50, 100, 150, 200, 250$. Number of replications for each $G$ is $2000$. $N_{max}=500$.
\end{tablenotes}
\end{threeparttable}
\end{table} 

\begin{table}[htbp]\centering
\begin{threeparttable}[b]
\def\sym#1{\ifmmode^{#1}\else\(^{#1}\)\fi}
\caption{Model 2 - Finite Population - Matching on $X_g$\tnote{1}}\label{cov_rate_m2}
\begin{tabular}{c*{9}{c}}
\hline\hline
    \multicolumn{1}{c}{$N_{max}/N_{min}$}&\multicolumn{1}{c}{VCE}
    &\multicolumn{1}{c}{$G=12$}&\multicolumn{1}{c}{$G=26$}&\multicolumn{1}{c}{$G=50$}
    &\multicolumn{1}{c}{$G=100$}&\multicolumn{1}{c}{$G=150$}
    &\multicolumn{1}{c}{$G=200$}&\multicolumn{1}{c}{$G=250$}\\
\hline
\multicolumn{9}{c}{} \\
\multicolumn{9}{c}{\bf Coverage} \\
\hline
\multirow{3}{*}{1.11}
& $\hat{v}^2$ & $1$ & $1$ & $0.9995$ & $1$ & $1$ & $1$ & $0.9990$ \\
& CR & $1$ & $1$ & $1$ & $1$ & $1$ & $1$ & $1$ \\
& PCVE & $1$ & $1$ & $1$ & $1$ & $1$ & $1$ & $1$ \\
\multicolumn{9}{l}{} \\
\multirow{3}{*}{1.42}
& $\hat{v}^2$ & $1$ & $1$ & $0.9990$ & $1$ & $1$ & $1$ & $0.9990$ \\
& CR & $1$ & $1$ & $1$ & $1$ & $1$ & $1$ & $1$ \\
& PCVE & $1$ & $1$ & $1$ & $1$ & $1$ & $1$ & $1$ \\
\multicolumn{9}{l}{} \\
\multirow{3}{*}{1.99}
& $\hat{v}^2$ & $1$ & $1$ & $0.9995$ & $1$ & $1$ & $1$ & $0.9985$ \\
& CR & $1$ & $1$ & $1$ & $1$ & $1$ & $1$ & $0.9995$ \\
& PCVE & $1$ & $1$ & $1$ & $1$ & $1$ & $1$ & $0.9995$ \\
\multicolumn{9}{l}{} \\
\multirow{3}{*}{3.31}
& $\hat{v}^2$ & $1$ & $1$ & $0.9990$ & $1$ & $0.9990$ & $0.9985$ & $0.9970$ \\
& CR & $1$ & $1$ & $1$ & $1$ & $1$ & $1$ & $0.9995$ \\
& PCVE & $1$ & $1$ & $1$ & $1$ & $1$ & $1$ & $0.9995$ \\
\multicolumn{9}{l}{} \\
\multirow{3}{*}{9.80}
& $\hat{v}^2$ & $1$ & $1$ & $1$ & $0.9995$ & $0.9990$ & $0.9965$ & $0.9960$ \\
& CR & $1$ & $1$ & $1$ & $1$ & $1$ & $0.9995$ & $0.9985$ \\
& PCVE & $1$ & $1$ & $1$ & $1$ & $1$ & $0.9995$ & $0.9985$ \\
\hline
\multicolumn{9}{c}{} \\
\multicolumn{9}{c}{{\bf Average Length}} \\
\hline
\multirow{3}{*}{1.11}
& $\hat{v}^2$ & $1.51070$ & $1.10752$ & $0.81935$ & $0.63852$ & $0.44747$ & $0.39393$ & $0.35735$ \\
& CR & $1.66339$ & $1.31058$ & $0.92939$ & $0.76490$ & $0.52908$ & $0.47240$ & $0.42471$ \\
& PCVE & $1.67962$ & $1.31421$ & $0.93901$ & $0.76591$ & $0.53029$ & $0.47223$ & $0.42367$ \\
\multicolumn{9}{l}{} \\
\multirow{3}{*}{1.42}
& $\hat{v}^2$ & $1.53829$ & $1.15173$ & $0.81013$ & $0.64981$ & $0.45659$ & $0.39764$ & $0.36359$ \\
& CR & $1.72251$ & $1.36383$ & $0.92401$ & $0.77462$ & $0.53466$ & $0.47511$ & $0.43120$ \\
& PCVE & $1.73073$ & $1.37272$ & $0.92403$ & $0.77627$ & $0.53808$ & $0.47500$ & $0.42954$ \\
\multicolumn{9}{l}{} \\
\multirow{3}{*}{1.99}
& $\hat{v}^2$ & $1.45130$ & $1.16632$ & $0.79474$ & $0.67449$ & $0.45573$ & $0.40764$ & $0.37492$ \\
& CR & $1.69166$ & $1.40618$ & $0.92103$ & $0.79970$ & $0.53349$ & $0.48243$ & $0.44014$ \\
& PCVE & $1.64456$ & $1.39143$ & $0.90836$ & $0.80309$ & $0.53384$ & $0.48525$ & $0.43991$ \\
\multicolumn{9}{l}{} \\
\multirow{3}{*}{3.31}
& $\hat{v}^2$ & $1.51039$ & $1.23004$ & $0.82204$ & $0.71496$ & $0.47173$ & $0.42133$ & $0.38757$ \\
& CR & $1.73747$ & $1.46680$ & $0.92359$ & $0.84163$ & $0.54618$ & $0.49257$ & $0.45049$ \\
& PCVE & $1.72595$ & $1.47169$ & $0.92085$ & $0.84367$ & $0.54881$ & $0.49426$ & $0.45014$ \\
\multicolumn{9}{l}{} \\
\multirow{3}{*}{9.80}
& $\hat{v}^2$ & $1.71776$ & $1.31631$ & $0.80818$ & $0.79366$ & $0.48406$ & $0.44659$ & $0.41584$ \\
& CR & $1.86668$ & $1.60387$ & $0.90901$ & $0.92440$ & $0.55159$ & $0.51513$ & $0.47517$ \\
& PCVE & $1.93059$ & $1.59637$ & $0.89610$ & $0.92753$ & $0.54784$ & $0.51592$ & $0.47486$ \\
\hline\hline
\end{tabular}
\label{tab:FPmodel2X}
\begin{tablenotes}
\item [1] Number of clusters $=2G$ with $G=12, 26, 50, 100, 150, 200, 250$. Number of replications for each $G$ is $2000$. $N_{max}=500$.
\end{tablenotes}
\end{threeparttable}
\end{table}

\begin{table}[htbp]\centering
\begin{threeparttable}[b]
\def\sym#1{\ifmmode^{#1}\else\(^{#1}\)\fi}
\caption{Model 2 - Finite Population - Matching on $X_g$ and $N_g$\tnote{1}}\label{CI_length_m2}
\begin{tabular}{c*{9}{c}}
\hline\hline
    \multicolumn{1}{c}{$N_{max}/N_{min}$}&\multicolumn{1}{c}{VCE}
    &\multicolumn{1}{c}{$G=12$}&\multicolumn{1}{c}{$G=26$}&\multicolumn{1}{c}{$G=50$}
    &\multicolumn{1}{c}{$G=100$}&\multicolumn{1}{c}{$G=150$}
    &\multicolumn{1}{c}{$G=200$}&\multicolumn{1}{c}{$G=250$}\\
\hline
\multicolumn{9}{c}{} \\
\multicolumn{9}{c}{\bf Coverage} \\
\hline
\multirow{3}{*}{1.11}
& $\hat{v}^2$ & $1$ & $1$ & $1$ & $1$ & $1$ & $1$ & $1$ \\
& CR & $1$ & $1$ & $1$ & $1$ & $1$ & $1$ & $1$ \\
& PCVE & $1$ & $1$ & $1$ & $1$ & $1$ & $1$ & $1$ \\
\multicolumn{9}{l}{} \\
\multirow{3}{*}{1.42}
& $\hat{v}^2$ & $1$ & $1$ & $1$ & $1$ & $1$ & $1$ & $1$ \\
& CR & $1$ & $1$ & $1$ & $1$ & $1$ & $1$ & $1$ \\
& PCVE & $1$ & $1$ & $1$ & $1$ & $1$ & $1$ & $1$ \\
\multicolumn{9}{l}{} \\
\multirow{3}{*}{1.99}
& $\hat{v}^2$ & $1$ & $1$ & $1$ & $1$ & $1$ & $1$ & $1$ \\
& CR & $1$ & $1$ & $1$ & $1$ & $1$ & $1$ & $1$ \\
& PCVE & $1$ & $1$ & $1$ & $1$ & $1$ & $1$ & $1$ \\
\multicolumn{9}{l}{} \\
\multirow{3}{*}{3.31}
& $\hat{v}^2$ & $1$ & $1$ & $1$ & $1$ & $1$ & $1$ & $1$ \\
& CR & $1$ & $1$ & $1$ & $1$ & $1$ & $1$ & $1$ \\
& PCVE & $1$ & $1$ & $1$ & $1$ & $1$ & $1$ & $1$ \\
\multicolumn{9}{l}{} \\
\multirow{3}{*}{9.80}
& $\hat{v}^2$ & $1$ & $1$ & $1$ & $1$ & $1$ & $1$ & $1$ \\
& CR & $1$ & $1$ & $1$ & $1$ & $1$ & $1$ & $1$ \\
& PCVE & $1$ & $1$ & $1$ & $1$ & $1$ & $1$ & $1$ \\
\hline
\multicolumn{9}{c}{} \\
\multicolumn{9}{c}{{\bf Average Length}} \\
\hline
\multirow{3}{*}{1.11}
& $\hat{v}^2$ & $1.43001$ & $0.98768$ & $0.71225$ & $0.57552$ & $0.38898$ & $0.34712$ & $0.30917$ \\
& CR & $1.66632$ & $1.31199$ & $0.93037$ & $0.76575$ & $0.52947$ & $0.47207$ & $0.42466$ \\
& PCVE & $1.66130$ & $1.30434$ & $0.94045$ & $0.76739$ & $0.52682$ & $0.47164$ & $0.42340$ \\
\multicolumn{9}{l}{} \\
\multirow{3}{*}{1.42}
& $\hat{v}^2$ & $1.35210$ & $1.08903$ & $0.68790$ & $0.58216$ & $0.39551$ & $0.34579$ & $0.31063$ \\
& CR & $1.71907$ & $1.36641$ & $0.92554$ & $0.77532$ & $0.53521$ & $0.47431$ & $0.43152$ \\
& PCVE & $1.71252$ & $1.36754$ & $0.92406$ & $0.77703$ & $0.53714$ & $0.47354$ & $0.43086$ \\
\multicolumn{9}{l}{} \\
\multirow{3}{*}{1.99}
& $\hat{v}^2$ & $1.36855$ & $1.04579$ & $0.68163$ & $0.60169$ & $0.38793$ & $0.35447$ & $0.31701$ \\
& CR & $1.68436$ & $1.40552$ & $0.92186$ & $0.80133$ & $0.53444$ & $0.48159$ & $0.44058$ \\
& PCVE & $1.64990$ & $1.37699$ & $0.91400$ & $0.80336$ & $0.53179$ & $0.48544$ & $0.43940$ \\
\multicolumn{9}{l}{} \\
\multirow{3}{*}{3.31}
& $\hat{v}^2$ & $1.43146$ & $1.11080$ & $0.69613$ & $0.64046$ & $0.40438$ & $0.36571$ & $0.32568$ \\
& CR & $1.73042$ & $1.46136$ & $0.92487$ & $0.84401$ & $0.54673$ & $0.49209$ & $0.45137$ \\
& PCVE & $1.71754$ & $1.45545$ & $0.92046$ & $0.84452$ & $0.55122$ & $0.49459$ & $0.44999$ \\
\multicolumn{9}{l}{} \\
\multirow{3}{*}{9.80}
& $\hat{v}^2$ & $1.62023$ & $1.24723$ & $0.68231$ & $0.71972$ & $0.41039$ & $0.37921$ & $0.34731$ \\
& CR & $1.85014$ & $1.59673$ & $0.91020$ & $0.92797$ & $0.55260$ & $0.51529$ & $0.47639$ \\
& PCVE & $1.92935$ & $1.60166$ & $0.90340$ & $0.93148$ & $0.54945$ & $0.51515$ & $0.47589$ \\
\hline\hline
\end{tabular}
\label{tab:FPmodel2XN}
\begin{tablenotes}
\item [1] Number of clusters $=2G$ with $G=12, 26, 50, 100, 150, 200, 250$. Number of replications for each $G$ is $2000$. $N_{max}=500$.
\end{tablenotes}
\end{threeparttable}
\end{table} 

\subsection{Simulation Results for Different Choices of $|\mathcal{M}_g|$}\label{sec:sims-sample}
In this section, we repeat the simulation exercise from Section \ref{sec:sims-unadj} for different choices of the second stage sample size $|\mathcal{M}_g| = \lfloor \rho\cdot N_g \rfloor$ for $\rho \in \{0.5, 0.6, 0.7, 0.8, 0.9\}$. In each case, we generate samples as in Section \ref{sec:sims-unadj}, but sample a fraction $\rho\cdot N_g$ of each cluster without replacement when computing $\hat{\Delta}_G$ and $\hat{v}^2_G$. Results for $G = 50$ and $G = 250$ are presented in Tables \ref{tab:rho_50_1}--\ref{tab:rho_250_2}. In each table, the results stay roughly the same across different values of $\rho$, with the average lengths of the confidence intervals slightly decreasing when $\rho$ increases. The stability across $\rho$ is not surprising in our model given the heavy dependence across the units within the same cluster.

\begin{table}[htbp]\centering
\begin{threeparttable}[b]
\def\sym#1{\ifmmode^{#1}\else\(^{#1}\)\fi}
\caption{Model 1 - $|\mathcal{M}_g| = \rho\cdot N_g$ with $G = 50$ - Matching on $X_g$ and  $N_g$ \tnote{1}}
\begin{tabular}{c*{7}{c}}
\hline\hline
    \multicolumn{1}{c}{$N_{max}/N_{min}$}&\multicolumn{1}{c}{VCE}
    &\multicolumn{1}{c}{$\rho=0.5$}&\multicolumn{1}{c}{$\rho=0.6$}&\multicolumn{1}{c}{$\rho=0.7$}
    &\multicolumn{1}{c}{$\rho=0.8$}&\multicolumn{1}{c}{$\rho=0.9$}\\
\hline
\multicolumn{7}{c}{} \\
\multicolumn{7}{c}{\bf Coverage} \\
\hline
\multirow{3}{*}{1.11}
& $\hat{v}^2$ & $0.9435$ & $0.9315$ & $0.9335$ & $0.9335$ & $0.9405$ \\
& CR & 1 & 1 & 1 & 1 & 1 \\
& PCVE & $0.9440$ & $0.9335$ & $0.9330$ & $0.9360$ & $0.9420$ \\
\multicolumn{7}{l}{} \\
\multirow{3}{*}{1.42}
& $\hat{v}^2$ & $0.9455$ & $0.9320$ & $0.9455$ & $0.9385$ & $0.9405$ \\
& CR & 1 & 1 & 1 & 1 & 1 \\
& PCVE & $0.9485$ & $0.9325$ & $0.9465$ & $0.9365$ & $0.9405$ \\
\multicolumn{7}{l}{} \\
\multirow{3}{*}{1.99}
& $\hat{v}^2$ & $0.9345$ & $0.9350$ & $0.9450$ & $0.9400$ & $0.9380$ \\
& CR & 1 & 1 & 1 & 1 & 1 \\
& PCVE & $0.9380$ & $0.9400$ & $0.9460$ & $0.9430$ & $0.9410$ \\
\multicolumn{7}{l}{} \\
\multirow{3}{*}{3.31}
& $\hat{v}^2$ & $0.9395$ & $0.9370$ & $0.9345$ & $0.9420$ & $0.9395$ \\
& CR & 1 & 1 & 1 & 1 & 1 \\
& PCVE & $0.9400$ & $0.9405$ & $0.9380$ & $0.9480$ & $0.9380$ \\
\multicolumn{7}{l}{} \\
\multirow{3}{*}{9.80}
& $\hat{v}^2$ & $0.9425$ & $0.9410$ & $0.9495$ & $0.9385$ & $0.9270$ \\
& CR & 1 & 1 & 1 & 1 & 1 \\
& PCVE & $0.9435$ & $0.9445$ & $0.9505$ & $0.9370$ & $0.9325$ \\
\hline
\multicolumn{7}{c}{} \\
\multicolumn{7}{c}{{\bf Average Length}} \\
\hline
\multirow{3}{*}{1.11}
& $\hat{v}^2$ & $0.40207$ & $0.39959$ & $0.39833$ & $0.39692$ & $0.39552$ \\
& CR & $1.62141$ & $1.62149$ & $1.62140$ & $1.62101$ & $1.62086$ \\
& PCVE & $0.39996$ & $0.39815$ & $0.39623$ & $0.39516$ & $0.39415$ \\
\multicolumn{7}{l}{} \\
\multirow{3}{*}{1.42}
& $\hat{v}^2$ & $0.35158$ & $0.34891$ & $0.34777$ & $0.34562$ & $0.34375$ \\
& CR & $1.63392$ & $1.63325$ & $1.63252$ & $1.63225$ & $1.63232$ \\
& PCVE & $0.35029$ & $0.34731$ & $0.34535$ & $0.34384$ & $0.34229$ \\
\multicolumn{7}{l}{} \\
\multirow{3}{*}{1.99}
& $\hat{v}^2$ & $0.35889$ & $0.35386$ & $0.35342$ & $0.35057$ & $0.34797$ \\
& CR & $1.65320$ & $1.65233$ & $1.65205$ & $1.65185$ & $1.65086$ \\
& PCVE & $0.35634$ & $0.35293$ & $0.35167$ & $0.34858$ & $0.34715$ \\
\multicolumn{7}{l}{} \\
\multirow{3}{*}{3.31}
& $\hat{v}^2$ & $0.38701$ & $0.38306$ & $0.37956$ & $0.37682$ & $0.37493$ \\
& CR & $1.68841$ & $1.68715$ & $1.68610$ & $1.68575$ & $1.68542$ \\
& PCVE & $0.38437$ & $0.37985$ & $0.37746$ & $0.37493$ & $0.37263$ \\
\multicolumn{7}{l}{} \\
\multirow{3}{*}{9.80}
& $\hat{v}^2$ & $0.44908$ & $0.44416$ & $0.44082$ & $0.43789$ & $0.43528$ \\
& CR & $1.75885$ & $1.75848$ & $1.75757$ & $1.75719$ & $1.75705$ \\
& PCVE & $0.44459$ & $0.43984$ & $0.43769$ & $0.43398$ & $0.43209$ \\
\hline\hline
\end{tabular}
\label{tab:rho_50_1}
\begin{tablenotes}
\item [1] Number of clusters $=2G$ with $G=50$ throughout. Number of replications for each $\rho$ is $2000$. $N_{max}=500$.
\end{tablenotes}
\end{threeparttable}
\end{table} 

\begin{table}[htbp]\centering
\begin{threeparttable}[b]
\def\sym#1{\ifmmode^{#1}\else\(^{#1}\)\fi}
\caption{Model 2 - $|\mathcal{M}_g| = \rho\cdot N_g$ with $G = 50$ - Matching on $X_g$ and  $N_g$\tnote{1}}
\begin{tabular}{c*{7}{c}}
\hline\hline
    \multicolumn{1}{c}{$N_{max}/N_{min}$}&\multicolumn{1}{c}{VCE}
    &\multicolumn{1}{c}{$\rho=0.5$}&\multicolumn{1}{c}{$\rho=0.6$}&\multicolumn{1}{c}{$\rho=0.7$}
    &\multicolumn{1}{c}{$\rho=0.8$}&\multicolumn{1}{c}{$\rho=0.9$}\\
\hline
\multicolumn{7}{c}{} \\
\multicolumn{7}{c}{\bf Coverage} \\
\hline
\multirow{3}{*}{1.11}
& $\hat{v}^2$ & $0.9540$ & $0.9540$ & $0.9455$ & $0.9530$ & $0.9515$ \\
& CR & $0.9870$ & $0.9880$ & $0.9890$ & $0.9895$ & $0.9890$ \\
& PCVE & $0.9870$ & $0.9875$ & $0.9890$ & $0.9895$ & $0.9890$ \\
\multicolumn{7}{l}{} \\
\multirow{3}{*}{1.42}
& $\hat{v}^2$ & $0.9530$ & $0.9525$ & $0.9525$ & $0.9560$ & $0.9565$ \\
& CR & $0.9865$ & $0.9900$ & $0.9880$ & $0.9865$ & $0.9890$ \\
& PCVE & $0.9870$ & $0.9900$ & $0.9875$ & $0.9870$ & $0.9890$ \\
\multicolumn{7}{l}{} \\
\multirow{3}{*}{1.99}
& $\hat{v}^2$ & $0.9500$ & $0.9485$ & $0.9475$ & $0.9455$ & $0.9520$ \\
& CR & $0.9860$ & $0.9885$ & $0.9870$ & $0.9890$ & $0.9880$ \\
& PCVE & $0.9860$ & $0.9895$ & $0.9870$ & $0.9885$ & $0.9880$ \\
\multicolumn{7}{l}{} \\
\multirow{3}{*}{3.31}
& $\hat{v}^2$ & $0.9460$ & $0.9470$ & $0.9475$ & $0.9480$ & $0.9470$ \\
& CR & $0.9850$ & $0.9890$ & $0.9875$ & $0.9840$ & $0.9845$ \\
& PCVE & $0.9845$ & $0.9905$ & $0.9870$ & $0.9845$ & $0.9850$ \\
\multicolumn{7}{l}{} \\
\multirow{3}{*}{9.80}
& $\hat{v}^2$ & $0.9475$ & $0.9420$ & $0.9450$ & $0.9475$ & $0.9455$ \\
& CR & $0.9790$ & $0.9850$ & $0.9820$ & $0.9860$ & $0.9835$ \\
& PCVE & $0.9785$ & $0.9855$ & $0.9820$ & $0.9865$ & $0.9835$ \\
\hline
\multicolumn{7}{c}{} \\
\multicolumn{7}{c}{{\bf Average Length}} \\
\hline
\multirow{3}{*}{1.11}
& $\hat{v}^2$ & $0.73376$ & $0.73321$ & $0.73231$ & $0.73105$ & $0.72948$ \\
& CR & $0.96896$ & $0.96879$ & $0.96889$ & $0.96688$ & $0.96575$ \\
& PCVE & $0.96936$ & $0.96884$ & $0.96858$ & $0.96676$ & $0.96545$ \\
\multicolumn{7}{l}{} \\
\multirow{3}{*}{1.42}
& $\hat{v}^2$ & $0.73555$ & $0.73355$ & $0.73364$ & $0.73220$ & $0.73197$ \\
& CR & $0.97830$ & $0.97690$ & $0.97677$ & $0.97497$ & $0.97623$ \\
& PCVE & $0.97814$ & $0.97681$ & $0.97712$ & $0.97551$ & $0.97590$ \\
\multicolumn{7}{l}{} \\
\multirow{3}{*}{1.99}
& $\hat{v}^2$ & $0.74875$ & $0.74732$ & $0.74460$ & $0.74303$ & $0.74426$ \\
& CR & $0.99345$ & $0.99257$ & $0.98995$ & $0.98866$ & $0.99003$ \\
& PCVE & $0.99326$ & $0.99258$ & $0.99005$ & $0.98826$ & $0.99013$ \\
\multicolumn{7}{l}{} \\
\multirow{3}{*}{3.31}
& $\hat{v}^2$ & $0.77167$ & $0.77033$ & $0.76704$ & $0.76421$ & $0.76631$ \\
& CR & $1.01607$ & $1.01609$ & $1.01194$ & $1.00929$ & $1.01166$ \\
& PCVE & $1.01571$ & $1.01559$ & $1.01192$ & $1.00913$ & $1.01135$ \\
\multicolumn{7}{l}{} \\
\multirow{3}{*}{9.80}
& $\hat{v}^2$ & $0.81196$ & $0.81261$ & $0.81153$ & $0.80961$ & $0.80766$ \\
& CR & $1.05338$ & $1.05499$ & $1.05399$ & $1.05304$ & $1.05132$ \\
& PCVE & $1.05429$ & $1.05492$ & $1.05480$ & $1.05326$ & $1.05128$ \\
\hline\hline
\end{tabular}
\label{tab:rho_50_2}
\begin{tablenotes}
\item [1] Number of clusters $=2G$ with $G=50$ throughout. Number of replications for each $\rho$ is $2000$. $N_{max}=500$.
\end{tablenotes}
\end{threeparttable}
\end{table} 

\begin{table}[htbp]\centering
\begin{threeparttable}[b]
\def\sym#1{\ifmmode^{#1}\else\(^{#1}\)\fi}
\caption{Model 1 - $|\mathcal{M}_g| = \rho\cdot N_g$ with $G = 250$ - Matching on $X_g$ and  $N_g$\tnote{1}}
\begin{tabular}{c*{7}{c}}
\hline\hline
    \multicolumn{1}{c}{$N_{max}/N_{min}$}&\multicolumn{1}{c}{VCE}
    &\multicolumn{1}{c}{$\rho=0.5$}&\multicolumn{1}{c}{$\rho=0.6$}&\multicolumn{1}{c}{$\rho=0.7$}
    &\multicolumn{1}{c}{$\rho=0.8$}&\multicolumn{1}{c}{$\rho=0.9$}\\
\hline
\multicolumn{7}{c}{} \\
\multicolumn{7}{c}{\bf Coverage} \\
\hline
\multirow{3}{*}{1.11}
& $\hat{v}^2$ & $0.9460$ & $0.9385$ & $0.9540$ & $0.9550$ & $0.9535$ \\
& CR & 1 & 1 & 1 & 1 & 1 \\
& PCVE & $0.9460$ & $0.9395$ & $0.9530$ & $0.9540$ & $0.9510$ \\
\multicolumn{7}{l}{} \\
\multirow{3}{*}{1.42}
& $\hat{v}^2$ & $0.9505$ & $0.9455$ & $0.9570$ & $0.9425$ & $0.9555$ \\
& CR & 1 & 1 & 1 & 1 & 1 \\
& PCVE & $0.9530$ & $0.9470$ & $0.9570$ & $0.9400$ & $0.9555$ \\
\multicolumn{7}{l}{} \\
\multirow{3}{*}{1.99}
& $\hat{v}^2$ & $0.9505$ & $0.9470$ & $0.9530$ & $0.9565$ & $0.9365$ \\
& CR & 1 & 1 & 1 & 1 & 1 \\
& PCVE & $0.9495$ & $0.9500$ & $0.9555$ & $0.9575$ & $0.9370$ \\
\multicolumn{7}{l}{} \\
\multirow{3}{*}{3.31}
& $\hat{v}^2$ & $0.9410$ & $0.9475$ & $0.9400$ & $0.9450$ & $0.9455$ \\
& CR & 1 & 1 & 1 & 1 & 1 \\
& PCVE & $0.9425$ & $0.9465$ & $0.9395$ & $0.9440$ & $0.9465$ \\
\multicolumn{7}{l}{} \\
\multirow{3}{*}{9.80}
& $\hat{v}^2$ & $0.9510$ & $0.9485$ & $0.9455$ & $0.9495$ & $0.9405$ \\
& CR & 1 & 1 & 1 & 1 & 1 \\
& PCVE & $0.9470$ & $0.9480$ & $0.9500$ & $0.9510$ & $0.9430$ \\
\hline
\multicolumn{7}{c}{} \\
\multicolumn{7}{c}{{\bf Average Length}} \\
\hline
\multirow{3}{*}{1.11}
& $\hat{v}^2$ & $0.14449$ & $0.14312$ & $0.14249$ & $0.14173$ & $0.14127$ \\
& CR & $0.73103$ & $0.73070$ & $0.73057$ & $0.73044$ & $0.73034$ \\
& PCVE & $0.14444$ & $0.14309$ & $0.14229$ & $0.14165$ & $0.14116$ \\
\multicolumn{7}{l}{} \\
\multirow{3}{*}{1.42}
& $\hat{v}^2$ & $0.10899$ & $0.10714$ & $0.10574$ & $0.10481$ & $0.10393$ \\
& CR & $0.73644$ & $0.73611$ & $0.73590$ & $0.73575$ & $0.73559$ \\
& PCVE & $0.10897$ & $0.10709$ & $0.10560$ & $0.10464$ & $0.10387$ \\
\multicolumn{7}{l}{} \\
\multirow{3}{*}{1.99}
& $\hat{v}^2$ & $0.10480$ & $0.10230$ & $0.10073$ & $0.09930$ & $0.09825$ \\
& CR & $0.74537$ & $0.74501$ & $0.74487$ & $0.74471$ & $0.74447$ \\
& PCVE & $0.10477$ & $0.10234$ & $0.10059$ & $0.09919$ & $0.09814$ \\
\multicolumn{7}{l}{} \\
\multirow{3}{*}{3.31}
& $\hat{v}^2$ & $0.11023$ & $0.10740$ & $0.10511$ & $0.10385$ & $0.10256$ \\
& CR & $0.76179$ & $0.76141$ & $0.76113$ & $0.76098$ & $0.76078$ \\
& PCVE & $0.11014$ & $0.10734$ & $0.10523$ & $0.10372$ & $0.10248$ \\
\multicolumn{7}{l}{} \\
\multirow{3}{*}{9.80}
& $\hat{v}^2$ & $0.12613$ & $0.12277$ & $0.12007$ & $0.11823$ & $0.11673$ \\
& CR & $0.79667$ & $0.79620$ & $0.79573$ & $0.79560$ & $0.79545$ \\
& PCVE & $0.12599$ & $0.12262$ & $0.12008$ & $0.11836$ & $0.11675$ \\
\hline\hline
\end{tabular}
\label{tab:rho_250_1}
\begin{tablenotes}
\item [1] Number of clusters $=2G$ with $G=250$ throughout. Number of replications for each $\rho$ is $2000$. $N_{max}=500$.
\end{tablenotes}
\end{threeparttable}
\end{table} 

\begin{table}[htbp]\centering
\begin{threeparttable}[b]
\def\sym#1{\ifmmode^{#1}\else\(^{#1}\)\fi}
\caption{Model 2 - $|\mathcal{M}_g| = \rho\cdot N_g$ with $G = 250$ - Matching on $X_g$ and  $N_g$\tnote{1}}
\begin{tabular}{c*{7}{c}}
\hline\hline
    \multicolumn{1}{c}{$N_{max}/N_{min}$}&\multicolumn{1}{c}{VCE}
    &\multicolumn{1}{c}{$\rho=0.5$}&\multicolumn{1}{c}{$\rho=0.6$}&\multicolumn{1}{c}{$\rho=0.7$}
    &\multicolumn{1}{c}{$\rho=0.8$}&\multicolumn{1}{c}{$\rho=0.9$}\\
\hline
\multicolumn{7}{c}{} \\
\multicolumn{7}{c}{\bf Coverage} \\
\hline
\multirow{3}{*}{1.11}
& $\hat{v}^2$ & $0.9525$ & $0.9500$ & $0.9515$ & $0.9545$ & $0.9490$ \\
& CR & $0.9935$ & $0.9935$ & $0.9940$ & $0.9940$ & $0.9955$ \\
& PCVE & $0.9930$ & $0.9930$ & $0.9940$ & $0.9940$ & $0.9955$ \\
\multicolumn{7}{l}{} \\
\multirow{3}{*}{1.42}
& $\hat{v}^2$ & $0.9525$ & $0.9520$ & $0.9505$ & $0.9545$ & $0.9515$ \\
& CR & $0.9935$ & $0.9945$ & $0.9965$ & $0.9935$ & $0.9950$ \\
& PCVE & $0.9935$ & $0.9945$ & $0.9970$ & $0.9935$ & $0.9960$ \\
\multicolumn{7}{l}{} \\
\multirow{3}{*}{1.99}
& $\hat{v}^2$ & $0.9490$ & $0.9480$ & $0.9535$ & $0.9555$ & $0.9515$ \\
& CR & $0.9950$ & $0.9940$ & $0.9945$ & $0.9925$ & $0.9950$ \\
& PCVE & $0.9945$ & $0.9940$ & $0.9945$ & $0.9925$ & $0.9940$ \\
\multicolumn{7}{l}{} \\
\multirow{3}{*}{3.31}
& $\hat{v}^2$ & $0.9470$ & $0.9510$ & $0.9480$ & $0.9480$ & $0.9465$ \\
& CR & $0.9930$ & $0.9925$ & $0.9940$ & $0.9950$ & $0.9935$ \\
& PCVE & $0.9925$ & $0.9925$ & $0.9935$ & $0.9950$ & $0.9935$ \\
\multicolumn{7}{l}{} \\
\multirow{3}{*}{9.80}
& $\hat{v}^2$ & $0.9505$ & $0.9510$ & $0.9520$ & $0.9550$ & $0.9470$ \\
& CR & $0.9935$ & $0.9915$ & $0.9935$ & $0.9935$ & $0.9935$ \\
& PCVE & $0.9925$ & $0.9915$ & $0.9935$ & $0.9930$ & $0.9940$ \\
\hline
\multicolumn{7}{c}{} \\
\multicolumn{7}{c}{{\bf Average Length}} \\
\hline
\multirow{3}{*}{1.11}
& $\hat{v}^2$ & $0.32094$ & $0.32029$ & $0.31989$ & $0.31952$ & $0.31931$ \\
& CR & $0.43789$ & $0.43732$ & $0.43698$ & $0.43672$ & $0.43658$ \\
& PCVE & $0.43788$ & $0.43732$ & $0.43700$ & $0.43676$ & $0.43657$ \\
\multicolumn{7}{l}{} \\
\multirow{3}{*}{1.42}
& $\hat{v}^2$ & $0.32054$ & $0.32012$ & $0.31967$ & $0.31917$ & $0.31898$ \\
& CR & $0.44196$ & $0.44168$ & $0.44144$ & $0.44098$ & $0.44075$ \\
& PCVE & $0.44193$ & $0.44176$ & $0.44142$ & $0.44099$ & $0.44083$ \\
\multicolumn{7}{l}{} \\
\multirow{3}{*}{1.99}
& $\hat{v}^2$ & $0.32540$ & $0.32455$ & $0.32406$ & $0.32367$ & $0.32335$ \\
& CR & $0.44862$ & $0.44792$ & $0.44768$ & $0.44744$ & $0.44705$ \\
& PCVE & $0.44870$ & $0.44802$ & $0.44771$ & $0.44746$ & $0.44718$ \\
\multicolumn{7}{l}{} \\
\multirow{3}{*}{3.31}
& $\hat{v}^2$ & $0.33416$ & $0.33324$ & $0.33299$ & $0.33244$ & $0.33192$ \\
& CR & $0.45933$ & $0.45865$ & $0.45869$ & $0.45818$ & $0.45777$ \\
& PCVE & $0.45940$ & $0.45876$ & $0.45880$ & $0.45823$ & $0.45785$ \\
\multicolumn{7}{l}{} \\
\multirow{3}{*}{9.80}
& $\hat{v}^2$ & $0.35255$ & $0.35044$ & $0.34980$ & $0.34945$ & $0.34852$ \\
& CR & $0.48124$ & $0.47943$ & $0.47896$ & $0.47883$ & $0.47811$ \\
& PCVE & $0.48147$ & $0.47949$ & $0.47922$ & $0.47898$ & $0.47819$ \\
\hline\hline
\end{tabular}
\label{tab:rho_250_2}
\begin{tablenotes}
\item [1] Number of clusters $=2G$ with $G=250$ throughout. Number of replications for each $\rho$ is $2000$. $N_{max}=500$.
\end{tablenotes}
\end{threeparttable}
\end{table} 

\end{document}